\documentclass[reqno,11pt]{article}
\usepackage{fullpage}
\usepackage{amsfonts}
\usepackage{amssymb,cite}
\usepackage{mathrsfs}
\usepackage[colorlinks,linkcolor=red,linkcolor=blue,linktocpage]{hyperref}
\hypersetup{colorlinks=true, citecolor=blue}
\usepackage[top=1 in,bottom=1 in,left=1.0 in,right=1.0 in]{geometry}
\usepackage{texdraw}
\usepackage[all]{xy}
\usepackage{latexsym}
\usepackage{color}
\usepackage{bm}
\usepackage{graphics,amsmath,amssymb,amsthm,accents,epic}
\usepackage{lineno}
\allowdisplaybreaks
\makeatletter

\newtheorem{thm}{Theorem}
\newtheorem{prop}{Proposition}
\numberwithin{equation}{section}
\usepackage{amsmath}

\newtheorem{remark}{Remark}
\usepackage{bm}
\DeclareMathAccent{\wtilde}{\mathord}{largesymbols}{"65}
\DeclareMathAccent{\what}{\mathord}{largesymbols}{"62}

\def\m@th{\mathsurround=0pt}
\mathchardef\bracell="0365
\def\upbrall{$\m@th\bracell$}
\def\undertilde#1{\mathop{\vtop{\ialign{##\crcr
    $\hfil\displaystyle{#1}\hfil$\crcr
     \noalign
     {\kern1.5pt\nointerlineskip}
     \upbrall\crcr\noalign{\kern1pt
   }}}}\limits}

\newcommand{\wh}{\widehat}
\newcommand{\wt}{\widetilde}

\newcommand{\Ps}{\boldsymbol{\psi}}
\newcommand{\Ph}{\boldsymbol{\phi}}
\newcommand{\Rho}{\boldsymbol{\rho}}
\newcommand{\Om}{\boldsymbol{\Omega}}
\newcommand{\Ome}{\boldsymbol{\omega}}
\newcommand{\Ga}{\boldsymbol{\Gamma}}

\newcommand{\al}{\alpha}
\newcommand{\be}{\beta}

\newcommand{\nn}{\nonumber}
\newcommand{\st}{\hbox{\tiny\it{T}}}

\newcommand{\ba}{\boldsymbol{a}}
\newcommand{\bb}{\boldsymbol{b}}
\newcommand{\bc}{\boldsymbol{c}}
\newcommand{\bd}{\boldsymbol{d}}
\newcommand{\bA}{\boldsymbol{A}}
\newcommand{\bB}{\boldsymbol{B}}
\newcommand{\bXi}{\boldsymbol{\Xi}}

\newcommand{\bP}{\boldsymbol{P}}
\newcommand{\bQ}{\boldsymbol{Q}}

\newcommand{\bG}{\boldsymbol{G}}

\newcommand{\bM}{\boldsymbol{M}}
\newcommand{\bK}{\boldsymbol{K}}
\newcommand{\bL}{\boldsymbol{L}}
\newcommand{\bT}{\boldsymbol{T}}
\newcommand{\bI}{\boldsymbol{I}}

\def \c#1{\accentset{\circ}{#1}}
\def \dc#1{\underaccent{\circ}{#1}}
\def \d#1{\accentset{\bullet}{#1}}
\def \dd#1{\underaccent{\bullet}{#1}}
\newcommand{\wth}[1]{\what{\wtilde{#1}}}
\newcommand{\wtd}[1]{\d{\wtilde{#1}}}
\newcommand{\whd}[1]{\d{\what{#1}}}
\newcommand{\wthd}[1]{\d{\what{\wtilde{#1}}}}
\def \dth#1{\underaccent{\what}{\underaccent{\wtilde}{#1}}}
\def \dh#1{\underaccent{\what}{#1}}
\def \dt#1{\underaccent{\wtilde}{#1}}

\def \dhc#1{\underaccent{\circ}{\underaccent{\what}{#1}}}
\def \dtc#1{\underaccent{\circ}{\underaccent{\wtilde}{#1}}}
\def \dthc#1{\underaccent{\circ}{\underaccent{\what}{\underaccent{\wtilde}{#1}}}}
\def \dhd#1{\underaccent{\bullet}{\underaccent{\what}{#1}}}
\def \dtd#1{\underaccent{\bullet}{\underaccent{\wtilde}{#1}}}
\def \dthd#1{\underaccent{\bullet}{\underaccent{\what}{\underaccent{\wtilde}{#1}}}}

\begin{document}
	
\title{Bilinear structures of the fourth-order lattice Gel'fand-Dikii equations}	

\author{Song-lin Zhao$^{1}$,~~Han Wang$^{1}$,~~Da-jun Zhang$^{2,3}$\footnote{
Corresponding author. Email: djzhang@staff.shu.edu.cn}\\
{\small $^{1}$School of Mathematical Sciences, Zhejiang University of Technology, Hangzhou 310023, China}\\
{\small $^{2}$Department  of Mathematics, Shanghai University, Shanghai 200444, China}\\
{\small $^{3}$Newtouch Center for Mathematics of Shanghai University,  Shanghai 200444,  China}
}

\maketitle

\begin{abstract}

In this paper we derive bilinear forms and present their solutions in Casoratians
for several fourth-order lattice Gel'fand-Dikii (lattice GD-4)  equations.
These equations were recently formulated from the  direct linearization approach
and exhibit the multidimensionally consistent property in multi-component form.
Based on the obtained soliton solutions,
we are able to extend these equations by introducing a parameter $\delta$.
These $\delta$-extended lattice GD-4  type equations are still consistent around the cube,
and their bilinear forms together with Casoratian solutions are provided.

\begin{description}
\item[Keywords:] fourth-order lattice Gel'fand-Dikii  equations;
bilinearization, Casoratian solution
%\item[Mathematics Subject Classification:] XXXX
\end{description}

\end{abstract}

%\tableofcontents

\section{Introduction}\label{sec-1}

Compared with the continuous and differential-difference case,
the discrete Gel'fand-Dikii (GD) type equations play important roles
in the research of discrete integrable systems.
This is because it is too hard to derive them as reductions from the lattice Kadomtsev-Petviashvili (KP) equations.
This makes the difficulty in particular in understanding the high order lattice GD-type equations
(e.g. the lattice Boussinesq (BSQ)-type)
from the lattice KP equations.
The discretization (in both space and time) of the GD hierarchy was proposed for the first time in \cite{NPCQ}
from the perspectives of direct linearization approach,
where the resulting lattice hierarchy takes
the (regular) quadrilateral  lattice potential KdV (lpKdV) equation and
the $3\times 3$ stencil lattice potential BSQ (lpBSQ) equation as two simplest members
in the second order and third order, respectively.
They have Miura-related counterparts,
namely, the lattice potential modified KdV (lpmKdV) equation \cite{NC-lKdV}
and lattice potential modified BSQ (lpmBSQ) equation \cite{NPCQ}.
In addition, there are also lattice Schwarzian KdV (lSKdV) equation \cite{NC-lKdV}
and lattice Schwarzian BSQ (lSBSQ) equation \cite{Nij-Schwarzian-1997}.
The lattice KdV type equations can be formulated from
the  lpKdV  equation together with its eigenfunction equations \cite{WZZZ-JPA-2024,SZZ-MPAG-2025,SZZ-PA-2026},
and they are all included in the
Adler-Bobenko-Suris (ABS) list \cite{ABS-2003} which are consistent-around-cube (CAC).
The lattice BSQ type equations are also CAC in three-component form \cite{Hie,Tongas,Walker,ZZN-SAPM}.
For the latter equations one can refer to the review paper \cite{HZ-BSQ-book} and references therein.
Recently, in \cite{Tela}, Tela et al. utilized the direct linearization approach to present
some lattice fourth-order GD (GD-4) type equations in quadrilateral form,
which are related to a quartic discrete dispersion relation (see Eq.\eqref{e-curve-4}).
These lattice equations can be viewed as higher-order members in the lattice GD hierarchy,
while represent the regular, modified and Schwarzian type equations,
but appear as much more complicated forms of five-component difference equations.
Compared with the lattice BSQ case (see \cite{Hie,HZ-BSQ-book,ZZN-SAPM}) that are all multidimensionally consistent,
the lattice GD-4 type equations have something special \cite{Tela}.
More concretely,
the number of equations in a quadrilateral lattice GD-4 system
has to be bigger than the number of dependent variables,
so that the evolution and the multidimensional consistency property of the system
are guaranteed simultaneously.
These lattice GD-4 type equations are expected to play more roles in understanding  discrete integrable systems.
More recently, in \cite{Kamp-quadrilinear},
periodic reductions and Laurent property for a lattice GD-4 equation were investigated, where
a Somos-like integer sequence for the associated quadrilinear equations was established.

It is  hard to study  the nonlinear lattice GD type equations from reductions of the lattice KP equations.
In fact,  it is also not clear enough how their bilinear forms are related to
the bilinear lattice KP equation, namely, the Hirota equation or the Hirota-Miwa equation.
A recent progress is due to \cite{WZM-PD-2024}
which builds the connections between the bilinear forms  of some ABS equations
and the Hirota-Miwa equation.
However, for the lattice BSQ-type equations, how their bilinear forms
(see \cite{HZ-BSQ-book}) are obtained as reductions from the Hirota-Miwa equation is not known yet.

The purpose of this paper is to explore the lattice GD-4 type equations obtained in \cite{Tela}
from a bilinear perspective.
For the lattice equations in the GD hierarchy,
some ABS equations have been bilinearized and solutions are presented in Casoratians \cite{HZ-Bi}.
The study was  extended to the lattice BSQ type equations
in a series of papers \cite{HieZh-BSQ-2010,HieZh-BSQ--2011,Nong-Phd-2014}
and a comprehensive review can be found in \cite{HZ-BSQ-book}.
So far the bilinearization of the lattice GD-4 type equations is still unrevealed.
In this paper we will study the lattice  GD-4 (A-2), (B-2) and (C-3) equations obtained in \cite{Tela}
and present their bilinear forms and Casoratian solutions.
Note that these names of equations follow those given by Hietarinta
for the lattice BSQ type equations \cite{Hie}.

The paper is organized as follows.
Section \ref{sec-2} concerns some preliminaries, involving the lattice GD-4 type equations,
discrete bilinear forms and Casoratians.
Section \ref{sec-3} is contributed to the bilinearization
of the lattice GD-4 type equations.
Then, in Section \ref{sec-4}, for each  lattice GD-4 type equation, we provide its
soliton solutions in Casoratian form.
Moreover, a parameter-extension of each equation
is obtained, which allows us to have more types of exact solutions, such as solitons, Jordan-block solutions
and quasi-rational solutions.
We provide enough proving details in this section for the GD-4 (B-2) equation
and sketch the results and proofs for the other equations.
Conclusions are given in Section \ref{sec-6}.
There are four appendices where some formulae for the shifts of Casoratians are listed.

\section{Preliminaries}\label{sec-2}

The present section deals with some preliminaries,
including the lattice GD-4 type equations in Hietarinta's form,
basic concepts and properties on discrete bilinear forms,
as well as definition and properties for Casoratians.

\subsection{Lattice GD-4 type equations}\label{sec-2.1}

The lattice GD-4 type equations derived in \cite{Tela}
are determined by a quartic dispersion relation
\begin{align}
\label{e-curve-4}
G(\omega,k):=\omega^4-k^4+\alpha_3(\omega^3-k^3)+\alpha_2(\omega^2-k^2)+\alpha_1(\omega-k)=0,
\end{align}
where $\{\alpha_j\}$ are complex coefficients.
Note that in the direct linearization scheme,  the lattice KdV type
and BSQ type equations are governed by
a quadratic relation ($\omega^2-k^2+\al_1(\omega-k)=0$) and cubic relation
($\omega^3-k^3+\al_2(\omega^2-k^2)+\al_1(\omega-k)=0$) \cite{ZZN-SAPM}, respectively.

These lattice GD-4 type equations  are listed as follows.\footnote{The names (A-2), (B-2) and (C-3)
correspond to Hietarinta's naming for lattice BSQ type equations \cite{Hie} (also see \cite{ZZN-SAPM}).}

\noindent\textbf{GD-4 (A-2):}
\begin{subequations}
\label{eq:A2-111}
\begin{align}
\label{eq:A2-ex_a}
& \wt{y}=z\wt{x}-x, \quad \wh{y}=z\wh{x}-x, \\
\label{eq:A2-ex_b}
& \wt{\eta}\wh{x}-\wh{\eta}\wt{x}=y(\wt{x}-\wh{x}), \\
\label{eq:A2-ex_bb}
& \wh{\xi}-\wt{\xi}=(\wh{z}-\wt{z})\wh{\wt{z}}, \\
\label{eq:A2-ex_eta-th}
&\wh{\wt{\eta}}=\frac{(\wt{y}\wh{z}-\wh{y}\wt{z})z-\xi(\wt{y}-\wh{y})}{z(\wt{z}-\wh{z})},\\
\label{eq:A2-ex_c}
& \eta=-x\wh{\wt{\xi}}+y\wh{\wt{z}}+\alpha_3(y-x\wh{\wt{z}})-\alpha_2 x
+\frac{G(-p,-a)\,\wt{x} -G(-q,-a)\,\wh{x}}{\wh{z}-\wt{z}}.
\end{align}
\end{subequations}

\noindent\textbf{Alternative GD-4 (A-2):}
\begin{subequations}
\label{eq:A2-222}
\begin{align}
\label{eq:A2-ex-N4-1 a}
& y=\wt{z} x-\wt{x}, \quad y=\wh{z}x-\wh{x}, \\
\label{eq:A2-ex-N4-1 b}
& \wt{\eta}\wh{x}-\wh{\eta}\wt{x}=\wh{\wt{y}}(\wt{x}-\wh{x}), \\
\label{eq:A2-ex-N4-1 bb}
& \wh{\xi}-\wt{\xi}=(\wh{z}-\wt{z})z, \\
\label{eq:alt-A2-ex_eta}
& \wh{\wt{\xi}}=\frac{-\eta(\wh{x}-\wt{x})+y(\wh{y}-\wt{y})+ \wt{x}\wh{y}
- \wh{x}\wt{y}}{x(\wh{x}-\wt{x})},\\
\label{eq:A2-ex-N4-1 c}
&\wh{\wt{\eta}}=-\wh{\wt{x}}\xi+\wh{\wt{y}}z-\alpha_3(\wh{\wt{y}}-\wh{\wt{x}}z)
-\alpha_2 \wh{\wt{x}}
-\frac{G(-p,-b)\wh{x}-G(-q,-b)\wt{x}}{\wh{z}-\wt{z}}.
\end{align}
\end{subequations}

\noindent\textbf{GD-4 (B-2):}
\begin{subequations}
\label{eq:B222}
\begin{align}
\label{eq:BSQ-ex_a}
& \wt{z}=x\wt{x}-y, \quad \wh{z}=x\wh{x}-y, \\
\label{eq:BSQ-ex_b}
& \wh{\xi}-\wt{\xi}=(\wh{x}-\wt{x})\wh{\wt{y}}, \\
\label{eq:BSQ-ex_bb}
&  \wh{\eta}-\wt{\eta}=(\wh{x}-\wt{x})z, \\
\label{eq:B2-ex_eta-th}
& \wh{\wt{\eta}}=\xi + \frac{(\wh{x}\wt{y}-\wt{x}\wh{y})x-y(\wt{y}-\wh{y})}{\wt{x}-\wh{x}},\\
\label{eq:BSQ-ex_c}
& \eta=\wh{\wt{\xi}}-\wh{\wt{y}}x+\wh{\wt{x}}z+\alpha_3(\wh{\wt{y}}-x\wh{\wt{x}}+z)
+\alpha_2(\wh{\wt{x}}-x)+\alpha_1+\frac{G(-p,-q)}{\wh{x}-\wt{x}}.
\end{align}
\end{subequations}

\noindent\textbf{GD-4 (C-3):}
\begin{subequations}
\label{eq:xyz-MSBSQ}
\begin{align}
\label{eq:xyz-MSBSQ-a}
& x-\wt{x}=\wt{y}z, \quad x-\wh{x}=\wh{y}z, \\
\label{eq:xyz-MSBSQ-b}
& (\wt{y}_1+y)\wh{y}=(\wh{y}_1+y)\wt{y}, \quad (\wh{z}_1-\wh{\wt{z}})\wt{z}
=(\wt{z}_1-\wh{\wt{z}})\wh{z}, \\
\label{eq:xyz-MSBSQ-c}
& y \wh{\wt{z}}_1-y_1\wh{\wt{z}}=
z\frac{G(-p,-b)\wh{z}\wt{y}-G(-q,-b)\wt{z}\wh{y}}{\wh{z}-\wt{z}}-\alpha_3y\wh{\wt{z}}
+G(-a,-b)\wh{\wt{x}}.
\end{align}
\end{subequations}

\noindent\textbf{Alternative GD-4 (C-3):}
\begin{subequations}
\label{eq:xyz-MSBSQ-1}
\begin{align}
\label{eq:xyz-MSBSQ-1a}
& x-\wt{x}=\wt{y}z, \quad x-\wh{x}=\wh{y}z, \\
\label{eq:xyz-MSBSQ-1b}
& (\wt{y}_1+y)\wh{y}=(\wh{y}_1+y)\wt{y},
\quad (\wh{z}_1-\wh{\wt{z}})\wt{z}=(\wt{z}_1-\wh{\wt{z}})\wh{z}, \\
\label{eq:xyz-MSBSQ-1c}
& y\wh{\wt{z}}_1-y_1 \wh{\wt{z}}=z\frac{G(-p,-a)\wh{z}\wt{y}-G(-q,-a)\wt{z}\wh{y}}{\wh{z}-\wt{z}}
-\alpha_3y\wh{\wt{z}}+G(-a,-b)x.
\end{align}
\end{subequations}
The notations adopted in \eqref{eq:A2-111}-\eqref{eq:xyz-MSBSQ-1} are as follows.
All dependent variables $x, y, z, \xi, \eta, y_1$ and $z_1$ are functions
defined on $\mathbb{Z}^2$ with discrete coordinates $(n,m)\in\mathbb{Z}^2$;
$p$ and $q$ respectively serve as spacing parameters in $n$- and $m$-direction;
$a$ and $b$  are parameters.
Conventional shorthands for shifts in the two directions are employed, which are, e.g.
\[x=x_{n, m},~~ \wt{x}=x_{n+1,m},~~ \wh{x}=x_{n,m+1},~~ \wh{\wt{x}}=x_{n+1,m+1}.\]
In addition, the backward shifts are denoted as $\dt{x}=x_{n-1,m}$, $\dh{x}=x_{n,m-1}$, etc.

Among these equations, the
GD-4 (A-2) equation \eqref{eq:A2-111} and its alternative form \eqref{eq:A2-222},
and the GD-4 (B-2) equation \eqref{eq:B222},
are respectively composed by  six equations but with five dependent variables.
The involvement of equations \eqref{eq:A2-ex_eta-th},
\eqref{eq:alt-A2-ex_eta} and \eqref{eq:B2-ex_eta-th} enable the corresponding systems
to allow up-right evolution with staircase initial values \cite{Tela}.
Although some equations (e.g. \eqref{eq:A2-ex_b}, \eqref{eq:A2-ex-N4-1 bb}
and \eqref{eq:BSQ-ex_bb}) are not necessary to the systems where they are to get up-right evolutions,
they are indispensable for their systems to be multidimensionally consistent \cite{Tela}.
The equation \eqref{eq:A2-222} is named as the alternative
form of \eqref{eq:A2-111} since it is related to \eqref{eq:A2-111} by reflection transformations
\begin{align}
p\rightarrow -p,\quad q\rightarrow -q,\quad n\rightarrow -n,\quad m\rightarrow -m,\quad
\alpha_3\rightarrow -\alpha_3,\quad \alpha_1\rightarrow -\alpha_1,\quad b\rightarrow -a.
\end{align}
Similar reflection transformation also exists between the GD-4 (C-3) equation \eqref{eq:xyz-MSBSQ}
and its alternative form \eqref{eq:xyz-MSBSQ-1}.
All the equations \eqref{eq:A2-111}-\eqref{eq:xyz-MSBSQ-1}
appeared in a same direct linearization scheme (e.g. Eq.(4.8), (4.10), (4.12), (4.15), (4.16) in \cite{Tela}),
which means they have solutions in terms of Cauchy matrix (cf.\cite{ZZN-SAPM}).
Note that in the GD-4 (B-2) equation \eqref{eq:B222},
there is constant $\alpha_1$ appearing explicitly in the right side of the equation \eqref{eq:BSQ-ex_c},
which can be removed by transformation
$\xi\rightarrow \xi-\frac{\alpha_1}{4}(n+m+2),~~\eta\rightarrow \eta-\frac{\alpha_1}{4}(n+m)$
(cf. \cite{Nong-Phd-2014}).

\subsection{Discrete Hirota's bilinear form and Casoratians}
\label{sec-2.2}

We say that an equation is in discrete Hirota's bilinear form if it can be expressed as \cite{HZ-Bi}
\begin{equation}
\label{eq:HB}
\sum_j\, c_j\, f_{j}(n+\nu_{j}^+,m+\mu_{j}^+)\,
g_{j}(n+\nu_{j}^-,m+\mu_{j}^-)=0,
\end{equation}
where the index sums $\nu_{j}^++\nu_{j}^-=\nu^s ,
\mu_{j}^++\mu_{j}^-=\mu^s$ do not depend on $j$.
Such an equation is gauge invariant, i.e.,
if a set of  functions $\{f_j,g_j\}$ solve \eqref{eq:HB},
then, so do the gauge transformed functions
\begin{equation}
\label{eq:gauge}
f'_j(n,m)=A^nB^m\, f_j(n,m),\quad  g'_j(n,m)=A^nB^m\, g_j(n,m),
\end{equation}
where $A$ and $B$ are independent of $(n,m)$.

A Casoratian (or Casorati determinant) is a discrete version of a Wronskian, defined as
\begin{align}
\label{eq:NXN}
|\Ps(n,m,l_1),\Ps(n,m,l_2),...,\Ps(n,m,l_N)|=|l_1,l_2,...,l_N|,
\end{align}
with the basic column vector $\Ps(n,m,l)=(\psi_1(n,m,l),\psi_2(n,m,l),\ldots,\psi_N(n,m,l))^{\st}$.
Following the shorthand notation given in \cite{FN}, some $N$-th order Casoratians that we
often use are denoted by
\begin{align}
& |0,1,...,N-1|=|\wh{N-1}|, \quad |0,1,...,N-2,N|=|\wh{N-2},N|, \nn \\
& |0,1,...,N-3,N-1,N|=|\wh{N-3},N-1,N|, \label{shorthand}
\end{align}
etc.

In the procedure of the verification of a solution in Casoratian form for a bilinear equation,
one needs some determinantal identities to simplify high order shifts.
Finally, the bilinear equation to be verified is reduced to the Pl\"ucker relation
(Laplace expansion of a zero-valued determinant).
Let us go through these determinantal identities.

\begin{prop}\label{thm-2-3-1}\cite{Zhang-KdV-2006}
Let  $\bXi\in \mathbb{C}^{N\times N}$ and denote its column vectors as $\{\bXi_j\}$;
let $\Om=(\Omega_{i,j})_{N\times N}$ be an operator matrix (i.e. $\Omega_{i,j}$ are operators),
and denote its column vectors as  $\{\Om_j\}$. The following relation holds
\begin{align}\label{id-w-2}
\sum^N_{j=1} |\Om_j*\bXi|
=\sum^N_{j=1}|(\Om^{\st})_{j}*\bXi^{\st}|,
\end{align}
where
\begin{align*}
|\bA_j * \bXi|=|\bXi_1,\cdots,\bXi_{j-1},~\bA_j \circ\bXi_j,~\bXi_{j+1},\cdots, \bXi_{N}|,
\end{align*}
and $\bA_j \circ\bXi_j$ stands for
\begin{align*}
\bA_j \circ \bB_j=(A_{1,j}B_{1,j},~A_{2,j}B_{2,j},\cdots, A_{N,j}B_{N,j})^{\st},
\end{align*}
in which $\bA_j =(A_{1,j},~A_{2,j},\cdots, A_{N,j})^{\st},~ \bB_j=(B_{1,j},~B_{2,j},\cdots, B_{N,j})^{\st}$
are $N$-th order vectors.
\end{prop}

\begin{prop}\label{cor-2-3-1}\cite{HJN-book-2016}
Let $\bP\in \mathbb{C}^{N\times(N-1)}$, $\bQ \in \mathbb{C}^{N\times(N-k+1)}$ be the remained of $\bP$
after removing its arbitrary $k-2$ columns where
$3\leq k < N$, and $\ba_i,~ i=1, 2, \cdots, k$, be $N$-th order column vectors. Then one has
\begin{align}
\label{plu-r-2}
\sum^{k}_{i=1}(-1)^{i-1} |\bP, \ba_i| \cdot
     |\bQ, \ba_1, \cdots, \ba_{i-1},\ba_{i+1}, \cdots,\ba_{k}|=0,
\quad k \geq 3.
\end{align}
\end{prop}
When $k=3$, \eqref{plu-r-2} yields the usual  identity used in Wronskian/Casoratian verification \cite{FN}:
\begin{align}
\label{plu-r-1}
|\bM, \ba,\bb||\bM, \bc, \bd|-|\bM, \ba, \bc||\bM, \bb, \bd|
+|\bM, \ba, \bd||\bM, \bb, \bc|=0,
\end{align}
where we have taken $\bP=(\bM, \ba)$, $\bQ=\bM$,
$\ba_1=\bb$, $\ba_2=\bc$ and $\ba_3=\bd$.

In order to construct   solutions in Casoratian form for the lattice GD-4 type equations,
we extend the basic column vector $\Ps(n,m,l)$ with two more lattice variables
$\al, \be$ corresponding to the spacing parameters $a, b$, respectively,
and thus we introduce an extension of the basic ingredients
\begin{equation}
\label{eq:psij}
\psi_s{(n,m,\alpha,\beta,l)}=\sum_{j=1}^4\rho_{j,s}^{(0)}
(-\omega_j(k_s))^l(p-\omega_j(k_s))^n(q-\omega_j(k_s))^m(a-\omega_j(k_s))^\alpha(b-\omega_j(k_s))^\beta,
\end{equation}
where
$\rho_{j,s}^{(0)}, k_s\in \mathbb{C}$, and
$\{\omega_j(k)\}$ represent the four roots of the quartic dispersion relation curve \eqref{e-curve-4}
with $\omega_4(k)\equiv k$.
In addition, following the operations given in \cite{NZSZ,ZhaZ-JNMP-2019}
we also introduce the following basic elements
\begin{equation}
\label{psi-jd}
\phi_s{(n,m,\alpha,\beta,l)}=\sum_{j=1}^4\rho_{j,s}^{(0)}
(\delta-\omega_j(k_s))^l(p-\omega_j(k_s))^n(q-\omega_j(k_s))^m(a-\omega_j(k_s))^\alpha(b-\omega_j(k_s))^\beta,
\end{equation}
which will be used to derive the quasi-rational solutions of the lattice GD-4 type equations.

Since $\alpha$, $\beta$ and $l$ play the same role as the lattice variables $n$ and $m$ in \eqref{eq:psij} and
\eqref{psi-jd}, we employ
new notations to represent the elementary shifts in these three directions, i.e.,
\begin{align}
\label{shifts}
f(n,m,\alpha+1,\beta,l)=\c{f}, \quad  f(n,m,\alpha,\beta+1,l)=\d{f}, \quad f(n,m,\alpha,\beta,l+1)=\Bar{f}.
\end{align}

In order to get quasi-rational solutions,
we will make use of the $N$-th order matrix in the following form
\begin{equation}
\mathcal{A}=\left(\begin{array}{cccccc}
a_0 & 0    & 0   & \cdots & 0   & 0 \\
a_1 & a_0  & 0   & \cdots & 0   & 0 \\
a_2 & a_1  & a_0 & \cdots & 0   & 0 \\
\vdots &\vdots &\cdots &\vdots &\vdots &\vdots \\
%0   & 0    & 0   & \cdots & k_j & 0 \\
a_{N-1} & a_{N-2} & a_{N-3}  & \cdots &  a_1   & a_0
\end{array}\right)_{N\times N}
\label{A}
\end{equation}
with scalar elements $\{a_j\}$.
This is called a $N$-th order lower triangular Toeplitz matrix.
All such matrices (of same order) compose a commutative set $\mathcal{H}$ with respect to matrix multiplication.
Such kind of matrices play useful roles in the expression of exact solution for multiple-pole equations.
For more properties of such matrices one can refer to Refs.\cite{Zhang-KdV-2006,ZZSZ}.

With these preliminaries in hand,
in the follow-up sections we will demonstrate the bilinear forms and Casoratian solutions of the
lattice GD-4 type equations.
For the sake of simplicity, we introduce the following notations:
\begin{subequations}\label{2.18}
\begin{align}
& H(p,q)=p^2+pq+q^2-\alpha_3(p+q)+\alpha_2,\\
\label{pa}
& p_a=\frac{G(-p,-a)}{p-a}, \quad p_b=\frac{G(-p,-b)}{p-b}, \quad
q_a=\frac{G(-q,-a)}{q-a}, \quad q_b=\frac{G(-q,-b)}{q-b}, \\
&\sigma_1=p-\delta,\quad \sigma_2=q-\delta,\quad \sigma_3=a-\delta,\quad \sigma_4=b-\delta, \quad \sigma^{\pm}_{ij}=\sigma_i\pm\sigma_j, \\
& \sigma^+_{ij\kappa}=\sigma_i+\sigma_j+\sigma_\kappa, \quad \sigma^-_{ij\kappa}=(\sigma_i-\sigma_j)(\sigma_i-\sigma_\kappa)(\sigma_j-\sigma_\kappa), \\
& \sigma^+_{ij\kappa\ell}=\sigma_i+\sigma_j+\sigma_\kappa+\sigma_\ell, \\
& \sigma^-_{ij\kappa\ell}=(\sigma_i-\sigma_j)(\sigma_i-\sigma_\kappa)(\sigma_i-\sigma_\ell)
(\sigma_j -\sigma_\kappa)
(\sigma_j-\sigma_\ell)(\sigma_\kappa-\sigma_\ell), \\
& \epsilon_1=4\delta-\alpha_3, \quad \epsilon_2=6\delta^2-3\delta\alpha_3+\alpha_2, \quad \epsilon_3=4\delta^3-3\delta^2\alpha_3+2\delta\alpha_2-\alpha_1, \\
\label{eq:Pt}
& P(\sigma_i,\sigma_j)=\sigma_i^4-\sigma_j^4+\epsilon_1(\sigma_i^3-\sigma_j^3)
+\epsilon_2(\sigma_i^2-\sigma_j^2)+\epsilon_3(\sigma_i-\sigma_j), \\
\label{eq:Qt}
& Q(\sigma_i,\sigma_j)=\sigma^{+^2}_{ij}-\sigma_i\sigma_j+\epsilon_1\sigma^+_{ij}+\epsilon_2,
\quad S(\sigma_i)=\sigma_i^4+\epsilon_1\sigma_i^3+\epsilon_2\sigma_i^2+\epsilon_3\sigma_i, \\
\label{eq:Rt}
& R(\sigma_i,\sigma_j)=\sigma^+_{ij}(\sigma_i^2+\sigma_j^2)+\epsilon_1(\sigma^{+^2}_{ij}-\sigma_i\sigma_j)
+\epsilon_2\sigma^+_{ij}+\epsilon_3,
\end{align}
\end{subequations}
with $i<j<\kappa<\ell$. It is clear that
\begin{align}
\label{sig-S}
\sigma^-_{ij}R(\sigma_i,\sigma_j)=P(\sigma_i,\sigma_j)=S(\sigma_i)-S(\sigma_j).
\end{align}

\section{Bilinear structures of the lattice GD-4 type equations}\label{sec-3}

In this section we show how the lattice GD-4 type equations are bilinearized and present their bilinear forms.

\subsection{GD-4 (B-2) equation}\label{sec-3-1}

To bilinearize the GD-4 (B-2) equation \eqref{eq:B222}, we consider
its deformation (Eq.(4.6) in \cite{Tela}), which reads
\begin{subequations}
\label{eq:B2o}
\begin{align}
\label{eq:GD4-B2o-ex_a}
& B_{21}\equiv pu_{0}+u_{1,0}-(p+u_{0})\wt{u}_{0}+\wt{u}_{0,1}=0,\\
\label{eq:GD4-B2o-ex_b}
& B_{22}\equiv qu_{0}+u_{1,0}-(q+u_{0})\wh{u}_{0}+\wh{u}_{0,1}=0,\\
\label{eq:GD4-B2o-ex_c}
& B_{23}\equiv\wh{u}_{2,0}-\wt{u}_{2,0}+p\wh{u}_{1,0}-q\wt{u}_{1,0}
-(p-q+\wh{u}_0-\wt{u}_0)\wth{u}_{1,0}=0, \\
\label{eq:GD4-B2o-ex_d}
& B_{24}\equiv\wh{u}_{0,2}-\wt{u}_{0,2}+p\wt{u}_{0,1}-q\wh{u}_{0,1}-(p-q+\wh{u}_0-\wt{u}_0)u_{0,1}=0, \\
\label{eq:GD4-B2o-ex_e}
& B_{25}\equiv 2\wth{u}_{0,2}-(p+q)(u_{1,0}+\wth{u}_{0,1})+(p-\wth{u}_{0})\wh{u}_{0,1}
+ (q-\wth{u}_{0})\wt{u}_{0,1}\nn \\
&\qquad\quad +(p+u_0)\wt{u}_{1,0}+(q+u_0)\wh{u}_{1,0}-2u_{2,0}=0, \\
\label{eq:GD4-B2o-ex_f}
& B_{26}\equiv\frac{G(-p,-q)}{p-q}-\frac{G(-p,-q)}{p-q+\wh{u}_0-\wt{u}_0}
+H(p,q)(u_0-\wth{u}_0) \nn \\
& \qquad\quad -(p+q-\alpha_3)(u_0\wh{\wt{u}}_0-\wth{u}_{1,0}-u_{0,1})
+u_0\wth{u}_{1,0}-\wh{\wt{u}}_0u_{0,1}
-\wth{u}_{2,0}+u_{0,2}=0.
\end{align}
\end{subequations}
We also call the above system the GD-4 (B-2) equation if not making
confusion with \eqref{eq:B222}.
This system is connected with Eq.\eqref{eq:B222} via a set of point transformations \cite{Tela}:
\begin{subequations}
\label{B2-tr}
\begin{align}
& u_0=x-x_{0}, \quad u_{1,0}=y-x_0u_0-y_0, \quad u_{0,1}=z-x_0u_0-z_0, \\
& u_{2,0}=\xi-x_0y+xz_0+\mu, \quad u_{0,2}=\eta-x_0z+xy_0+\nu,
\end{align}
\end{subequations}
where
\begin{subequations}
\label{eq:xyz0-ex-N4}
\begin{align}
& x_0=-pn-qm-c_1, \\
& y_0= \frac{1}{2}\left(x_0^2+(np^2+mq^2+c_2)\right) +c_3,\\
& z_0= \frac{1}{2}\left(x_0^2-(np^2+mq^2+c_2)\right) -c_3,\\
& \mu_0=np^{3}+mq^{3}+c_4,\\
& \mu=-\frac{1}{6}(x^{3}_{0} -3(y_{0}-z_{0})x_{0}-2 \mu_0),\\
& \nu=-\frac{1}{6}( x^{3}_{0} +3(y_{0}-z_{0})x_{0}-2\mu_0),
\end{align}
\end{subequations}
and $\{c_{i}\}$ are constants.

To render the GD-4 (B-2) equation \eqref{eq:B2o} in bilinear form,
we adopt the following rational transformations\footnote{
These $u_0=u_{0,0}$ and $u_{i,j}$ are functions defined in the direct linearization (DL) scheme \cite{Tela}.
$f$ serves as the  tau function which is inherent in this scheme.
Therefore the form \eqref{eq:B2o} in DL variables is more natural and direct to be used in the bilinearization.
The auxiliary function $\theta$ in \eqref{eq:B2-h1}
is actually implied from a DL variable  $u_{1,1}=\theta/f$ (see \cite{Tela}).
Combining  \eqref{B2-tr} and \eqref{eq:tr-GD-4},
one can express Hietarinta's variables $x, y,\cdots$ in terms of $f, g, \cdots$,
but \eqref{eq:xyz0'-ex-N4} will be involved as background solutions.
}
\begin{align}
\label{eq:tr-GD-4}
u_{0}=\frac{g}{f}, \quad u_{1,0}=\frac{h}{f}, \quad u_{0,1}=\frac{s}{f}, \quad
u_{2,0}=\frac{\mu}{f},\quad u_{0,2}=\frac{\nu}{f},
\end{align}
under which the equation \eqref{eq:B2o} is bilinearized as
\begin{subequations}
\label{eq:B2-h1}
\begin{align}
\label{eq:GD4-B12-ex_a}
& \mathcal{H}_{11}\equiv \wt{f}(pg+h)-\wt{g}(pf+g)+f\wt{s}=0,\\
\label{eq:GD4-B12-ex_b}
& \mathcal{H}_{12}\equiv \wh{f}(qg+h)-\wh{g}(qf+g)+f\wh{s}=0, \\
\label{eq:GD4-B12-ex_c}
& \mathcal{H}_{13}\equiv\wt{f}(ph+\mu)-\wt{h}(pf+g)+f\wt{\theta}=0, \\
\label{eq:GD4-B12-ex_d}
& \mathcal{H}_{14}\equiv\wh{f}(qh+\mu)-\wh{h}(qf+g)+f\wh{\theta}=0, \\
\label{eq:GD4-B12-ex_e}
& \mathcal{H}_{15}\equiv\wt{f}(ps+\theta)-f(p\wt{s}-\wt{\nu})-\wt{g}s=0, \\
\label{eq:GD4-B12-ex_f}
& \mathcal{H}_{16}\equiv \wh{f}(q s+\theta)-f(p\wh{s}-\wh{\nu})-\wh{g}s=0, \\
\label{eq:GD4-B12-ex_g}
& \mathcal{H}_{17}\equiv(p-q)(\wt{f}\wh{f}-f\wth{f})+\wt{f}\wh{g}-\wh{f}\wt{g}=0, \\
\label{eq:GD4-B12-ex_h}
& \mathcal{H}_{18}\equiv\frac{G(-p,-q)}{p-q}(f\wth{f}-\wt{f}\wh{f})+H(p,q)(g\wth{f}-f\wth{g}) \nn \\
& \qquad\quad -(p+q-\alpha_3)(g\wth{g}-f\wth{h}-\wth{f}s)+g\wth{h}-s\wth{g}-f\wth{\mu}+\nu\wth{f}=0,
\end{align}
\end{subequations}
where $\theta$ is an auxiliary variable.
One can  verify that the nonlinear system \eqref{eq:B2o}
and the above bilinear forms \eqref{eq:B2-h1} are related by
\begin{subequations}
\begin{align*}
& B_{21}=\frac{\mathcal{H}_{11}}{f\wt{f}},\quad
B_{22}=\frac{\mathcal{H}_{12}}{f\wh{f}}, \quad
B_{23}=\frac{\wh{\mathcal{H}}_{13}}{\wh{f}\wth{f}}-\frac{\wt{\mathcal{H}}_{14}}{\wt{f}\wth{f}}, \quad
B_{24}=\frac{\mathcal{H}_{16}}{f\wh{f}}-\frac{\mathcal{H}_{15}}{f\wt{f}}, \\
& B_{25}=\frac{\wh{\mathcal{H}}_{15}}{\wh{f}\wth{f}}+
\frac{\wt{\mathcal{H}}_{16}}{\wt{f}\wth{f}}
-\frac{\mathcal{H}_{13}}{f\wt{f}}
-\frac{\mathcal{H}_{14}}{f\wh{f}}, \quad
B_{26}=\frac{1}{f\wth{f}}\mathcal{H}_{18}+\frac{G(-p,-q)\wt{f}\wh{f}}{(p-q)f\wth{f}\mathcal{H}_{17}
+(p-q)^2(f\wth{f})^2} \mathcal{H}_{17}.
\end{align*}
\end{subequations}

\subsection{GD-4 (A-2) equation and its alternative form}\label{sec-3-2}

\subsubsection{GD-4 (A-2) equation}\label{sec-3-2-1}

For the bilinearization of the GD-4 (A-2) equation \eqref{eq:A2-111},
we make use of  the following  deformation   (Eq.(4.2) in \cite{Tela}):
\begin{subequations}
\label{eq:A21o}
\begin{align}
\label{eq:GD4-A21-ex_a}
& A_{11}\equiv \wt{s}_a-(p+u_0)\wt{v}_a+(p-a) v_a=0,\\
\label{eq:GD4-A21-ex_b}
& A_{12}\equiv\wh{s}_a-(q+u_0)\wh{v}_a+(q-a) v_a=0,\\
\label{eq:GD4-A21-ex_c}
& A_{13}\equiv(\wt{r}_a-p\wt{s}_a)/\wt{v}_a+(p-a)s_a/\wt{v}_a
-(\wh{r}_a-q\wh{s}_a)/\wh{v}_a-(q-a)s_a/\wh{v}_a=0, \\
\label{eq:GD4-A21-ex_d}
& A_{14}\equiv 2\wh{\wt{r}}_a-(p+q)\wh{\wt{s}}_a+(p-a)\wh{s}_a+(q-a)\wt{s}_a \nn \\
& \qquad\quad -\wh{\wt{v}}_a(p\wt{u}_0+q\wh{u}_0-(p+q-\wt{u}_0-\wh{u}_0)u_0-2u_{1,0})=0, \\
\label{eq:GD4-A21-ex_e}
& A_{15}\equiv\wh{u}_{1,0}-\wt{u}_{1,0}-(p-q+\wh{u}_0-\wt{u}_0)\wh{\wt{u}}_0+p\wh{u}_0-q\wt{u}_0=0, \\
\label{eq:GD4-A21-ex_f}
& A_{16}\equiv\big[H(p,q)+\wh{\wt{u}}_{1,0}-(p+q-\alpha_3)\wh{\wt{u}}_{0}
+\big((p+q-\alpha_3-\wh{\wt{u}}_{0})s_a+r_a\big)/v_a\big] \nn \\
& \qquad\quad \cdot (p-q+\wh{u}_0-\wt{u}_0)-(p_a\wt{v}_a-q_a\wh{v}_a)/v_a=0,
\end{align}
\end{subequations}
where $p_a$ and $q_a$ are defined as \eqref{pa}.
It is connected with \eqref{eq:A2-111} by the transformation \cite{Tela}:
\begin{subequations}
\label{A2-tr}
\begin{align}
& v_a=x/x_{a}, \quad u_0=z-z_0, \quad s_a=(y-v_ay_a)/x_a, \\
& u_{1,0}=\xi-z_0u_0-\xi_0, \quad r_a =(\eta-z_0y+\xi_0x)/x_a,
\end{align}
\end{subequations}
where
\begin{subequations}
\label{eq:xyz-0-A2}
\begin{align}
& x_a=(p-a)^{-n}(q-a)^{-m}c_0, \quad z_0=-pn-qm-c_1, \quad y_a=x_az_0, \\
& \xi_0= \left(z_0^2+(np^2+mq^2+c_2)\right)/2+c_3,
\end{align}
\end{subequations}
and $\{c_i\}$ are constants.
Consider the following transformation
\begin{align}
\label{A21-tran}
& u_{0}=\frac{g}{f}, \quad u_{1,0}=\frac{h}{f}, \quad
v_a=\frac{\dc{f}}{f}, \quad s_a=\frac{a\dc{f}+\dc{g}}{f}, \quad r_a=\frac{a^2\dc{f}+a\dc{g}+\dc{s}}{f},
\end{align}
where the notation $\dc{f}$ is defined in \eqref{shifts},
denoting a shift in the auxiliary $\alpha$-direction.
With this transformation,
the GD-4 (A-2) equation \eqref{eq:A21o} can be converted into the bilinear form
\begin{subequations}
\label{eq:A21-h}
\begin{align}
\label{eq:GD4-A211-ex_a}
& \mathcal{H}_{21}\equiv f(a\wt{\dc{f}}+\wt{\dc{g}})+(p-a)\wt{f}\dc{f}-(pf+g)\wt{\dc{f}}=0,\\
\label{eq:GD4-A211-ex_b}
& \mathcal{H}_{22}\equiv f(a\wh{\dc{f}}+\wh{\dc{g}})+(q-a)\wh{f}\dc{f}-(qf+g)\wh{\dc{f}}=0,\\
\label{eq:GD4-A211-ex_c}
& \mathcal{H}_{23}\equiv f(a^2\wt{\dc{f}}+a\wt{\dc{g}}+\wt{\dc{s}})-pf(a\wt{\dc{f}}+\wt{\dc{g}})
+(p-a)\wt{f}(a\dc{f}+\dc{g})-\wt{\dc{f}}s=0, \\
\label{eq:GD4-A211-ex_d}
& \mathcal{H}_{24}\equiv f(a^2\wh{\dc{f}}+a\wh{\dc{g}}+\wh{\dc{s}})-qf(a\wh{\dc{f}}
+\wh{\dc{g}})+(q-a)\wh{f}(a\dc{f}+\dc{g})-\wh{\dc{f}}s=0,
\\
\label{eq:GD4-A211-ex_e}
& \mathcal{H}_{25}\equiv \wt{f}(pg+h)-\wt{g}(pf+g)+f\wt{s}=0,\\
\label{eq:GD4-A211-ex_f}
& \mathcal{H}_{26}\equiv \wh{f}(qg+h)-\wh{g}(qf+g)+f\wh{s}=0,\\
\label{eq:GD4-A211-ex_g}
& \mathcal{H}_{27}\equiv (p-q)(\wt{f}\wh{f}-f\wth{f})+\wt{f}\wh{g}-\wh{f}\wt{g}=0,\\
\label{eq:GD4-A211-ex_h}
& \mathcal{H}_{28}\equiv\big[H(p,q)\wth{f}\dc{f}+\wth{h}\dc{f}-(p+q-\alpha_3)\wth{g}\dc{f}
+((p+q-\alpha_3)\wth{f}
-\wth{g})(a\dc{f}+\dc{g})\nn\\
&\quad\quad\quad +\wth{f}(a^2\dc{f}+a\dc{g}+\dc{s})\big](p-q)-p_a\wh{f}\wt{\dc{f}}+q_a\wt{f}\wh{\dc{f}}=0.
\end{align}
\end{subequations}
It is related to  the nonlinear equation \eqref{eq:A21o} by
\begin{align*}
& A_{11}=\frac{\mathcal{H}_{21}}{f\wt{f}},\quad A_{12}=\frac{\mathcal{H}_{22}}{f\wh
f},\quad A_{13}=\frac{\mathcal{H}_{23}}{f\wt{f}\wt{v}_a}-\frac{\mathcal{H}_{24}}{f\wh{f}\wh{v}_a},
\quad A_{14}=\frac{\wh{\mathcal{H}}_{23}}{\wh{f}\wth{f}}+\frac{\wt{\mathcal{H}}_{24}}{\wt{f}\wth{f}}
+ \frac{\wth{v}_a\mathcal{H}_{25}}{f\wt{f}}+\frac{\wth{v}_a\mathcal{H}_{26}}{f\wh{f}}, \nn \\
& A_{15}=\frac{\wh{\mathcal{H}}_{25}}{\wh{f}\wth{f}}-\frac{\wt{\mathcal{H}}_{26}}{\wt{f}\wth{f}}, \quad
A_{16}=\frac{\mathcal{H}_{27}\mathcal{H}_{28}+(p_a\wh{f}\wt{\dc{f}}
-q_a\wt{f}\wh{\dc{f}})\mathcal{H}_{27}}
{(p-q)\dc{f}\wt{f}\wh{f}\wth{f}}+\frac{f\mathcal{H}_{28}}{\dc{f}\wt{f}\wh{f}}.
\end{align*}

\subsubsection{Alternative GD-4 (A-2) equation}\label{sec-3-2-2}

For the alternative GD-4 (A-2) equation \eqref{eq:A2-222},
we start from the following deformation of \eqref{eq:A2-222}:
\begin{subequations}
\label{eq:A22o}
\begin{align}
\label{eq:A22-a}
& A_{21}=t_b-(p-b)\wt{w}_b+(p-\wt{u}_0)w_b=0,\\
\label{eq:A22-b}
& A_{22}=t_b-(q-b)\wh{w}_b+(q-\wh{u}_0)w_b=0, \\
\label{eq:A22-c}
& A_{23}=(\wh{z}_b+p\wh{t}_b)/\wh{w}_b-(p-b){\wth{t}}_b/\wh{w}_b
-(\wt{z}_b+q\wt{t}_b)/\wt{w}_b+(q-b)\wth{t}_b/\wt{w}_b=0, \\
\label{eq:A22-d}
& A_{24}=\wh{u}_{0,1}-\wt{u}_{0,1}-(p-q+\wh{u}_0-\wt{u}_0)u_0+p\wt{u}_0-q\wh{u}_0=0, \\
\label{eq:A22-e}
& A_{25}=2z_b-(p-b)\wt{t}_b-(q-b)\wh{t}_b+(p+q){t}_b\nn\\
&\qquad \quad -w_b((p+q+\wt{u}_0+\wh{u}_0)\wth{u}_0-2\wth{u}_{0,1}-q\wt{u}_0-p\wh{u}_0)=0, \\
\label{eq:A22-f}
&A_{26}=\big[H(p,q)+u_{0,1}+(p+q-\alpha_3)u_{0}
- \big((p+q-\alpha_3+u_{0})\wth{t}_b-\wth{z}_b\big)/\wth{w}_b\big] \nn \\
&\qquad\quad \cdot (p-q+\wh{u}_0-\wt{u}_0)-(p_b\wh{w}_b-q_b\wt{w}_b)/\wth{w}_b=0,
\end{align}
\end{subequations}
where $p_b$ and $q_b$ are defined as in \eqref{pa}.
This system was already presented in \cite{Tela} and is connected with  \eqref{eq:A2-222}
via the transformation \cite{Tela}
\begin{subequations}
\label{eq:xyz-A2-1}
\begin{align}
& w_b=x/x_b, \quad u_0=z-z_0, \quad t_b=(y-w_by_b)/x_b, \\
& u_{0,1}=\xi-z_0u_0-\xi_0, \quad z_b=(\eta-z_0y+x\xi_0)/x_b,
\end{align}
\end{subequations}
with
\begin{subequations}
\label{eq:xyz-0-A2-1}
\begin{align}
& x_b =(-p+b)^{n}(-q+b)^{m}c_0, \quad z_0=-pn-qm-c_1, \quad y_b=x_bz_0,\\
& \xi_0 = \left(z_0^2-(np^2+mq^2+c_2)\right)/2-c_3,
\end{align}
\end{subequations}
and constants $\{c_i\}$.
We need to make use the auxiliary $\beta$-direction whose shift is denoted by a ``dot'' as defined in \eqref{shifts}.
Through the  transformation
\begin{align}
\label{A22-tran}
u_{0}=\frac{g}{f}, \quad u_{0,1}=\frac{s}{f}, \quad w_b=\frac{\d{f}}{f},
\quad t_b=\frac{\d{g}-b\d{f}}{f}, \quad z_b=\frac{\d{h}-b\d{g}+b^2\d{f}}{f},
\end{align}
the GD-4 (A-2) system \eqref{eq:A22o} can be written as
\begin{subequations}
\begin{align*}
& A_{21}=\frac{\mathcal{H}_{31}}{f\wt{f}},\quad A_{22}=\frac{\mathcal{H}_{32}}{f\wh f},
\quad A_{23}=\frac{\wh{\mathcal{H}}_{33}}{\wh{f}\wth{f}\wh{w}_a}
-\frac{\wt{\mathcal{H}}_{34}}{\wt{f}\wth{f}\wt{w}_a},\\
&A_{24}=\frac{\mathcal{H}_{36}}{f\wh{f}}-\frac{\mathcal{H}_{35}}{f\wt{f}},
\quad A_{25}=\frac{\mathcal{H}_{33}}{f\wt{f}}+\frac{\mathcal{H}_{34}}{f\wh{f}}
+\frac{w_b\wt{\mathcal{H}}_{36}}{\wt{f}\wth{f}}+\frac{w_b\wh{\mathcal{H}}_{35}}{\wh{f}\wth{f}},\nn\\
&A_{26}=\frac{\mathcal{H}_{37} \mathcal{H}_{38}+(p_b\wt{f}\whd{f}
- q_b\wh{f}\wtd{f})\mathcal{H}_{37}}{(p-q)f\wt{f}\wh{f}\wthd{f}}
+\frac{\wth{{f}} \mathcal{H}_{38}}{\wt{f}\wh{f}\wthd{f}},
\end{align*}
\end{subequations}
where
\begin{subequations}
\label{eq:A2'-b0}
\begin{align}
\label{eq:A22-b0-ex_a}
\mathcal{H}_{31}\equiv&  \wt{f}(p\d{f}+\d{g})-(p-b)f\wtd{f}-(b\wt{f}+\wt{g})\d{f}=0,\\
\label{eq:A22-b0-ex_b}
\mathcal{H}_{32}\equiv&  \wh{f}(q\d{f}+\d{g})-(q-b)f\whd{f}-(b\wh{f}+\wh{g})\d{f}=0,
\\
\label{eq:A22-b0-ex_c}
\mathcal{H}_{33}\equiv& \wt{f}(b^2\d{f}-b\d{g}+\d{h})+p\wt{f}(\d{g}-b\d{f})-(p-b)f(\wtd{g}-b\wtd{f})-\d{f}\wt{h}=0,\\
\label{eq:A22-b0-ex_d}
\mathcal{H}_{34}\equiv &  \wh{f}(b^2\d{f}-b\d{g}+\d{h})+q\wh{f}(\d{g}-b\d{f})-(q-b)f(\whd{g}-b\whd{f})-\d{f}\wh{h}=0,\\
\label{eq:A22-b0-ex_e}
\mathcal{H}_{35}\equiv &  \wt{f}(pg+h)-\wt{g}(pf+g)+f\wt{s}=0,\\
\label{eq:A22-b0-ex_f}
\mathcal{H}_{36}\equiv &  \wh{f}(qg+h)-\wh{g}(qf+g)+f\wh{s}=0,\\
\label{eq:A22-b0-ex_g}
\mathcal{H}_{37}\equiv &  (p-q)(\wt{f}\wh{f}-f\wth{f})+\wt{f}\wh{g}-\wh{f}\wt{g}=0,\\
\label{eq:A22-b0-ex_h}
\mathcal{H}_{38}\equiv& \big[H(p,q)f\wthd{f}+s\wthd{f}+(p+q-\alpha_3)g\wthd{f}
-((p+q-\alpha_3)f+g)(\wthd{g}-b\wthd{f})\nn\\
&+f(\wthd{h}-b\wthd{g}+b^2\wthd{f})\big](p-q)-p_b\wt{f}\whd{f}+q_b\wh{f}\wtd{f}=0,
\end{align}
\end{subequations}
which compose the bilinear form of the alternative GD-4 (A-2) equation \eqref{eq:A22o}.

\subsection{GD-4 (C-3) equation}\label{sec-3-3}

We consider the following deformation of  the GD-4 (C-3) equation \eqref{eq:xyz-MSBSQ}: \cite{Tela}
\begin{subequations}
\label{eq:C3o2}
\begin{align}
\label{eq:C3o2-a}
& C_{31}=(p-a)S_{a,b}-(p-b)\wt{S}_{a,b}-\wt{v}_a w_b=0,\\
\label{eq:C3o2-b}
& C_{32}=(q-a)S_{a,b}-(q-b)\wh{S}_{a,b}-\wh{v}_a w_b=0,\\
\label{eq:C3o2-c}
& C_{33}=\wt{s}_a/\wt{v}_a-\wh{s}_a/\wh{v}_a
-(p-q)+v_a((p-a)/\wt{v}_a-(q-a)/\wh{v}_a)=0, \\
\label{eq:C3o2-d}
& C_{34}=\wt{t}_b/\wt{w}_b-\wh{t}_b/\wh{w}_b
-(p-q)+\wh{\wt{w}}_b((p-b)/\wh{w}_b-(q-b)/\wt{w}_b)=0, \\
\label{eq:C3o2-e}
&C_{35}=v_{a}\wh{\wt{t}}_{b}-s_a\wh{\wt{w}}_{b}+w_b\frac{\frac{p_b}{p-a}\wh{w}_b\wt{v}_a
-\frac{q_b}{q-a}\wt{w}_b\wh{v}_a}{(p-b)\wt{w}_b-(q-b)\wh{w}_b}-(p+q-\alpha_3)v_a\wh{\wt{w}}_{b}\nn\\
&\quad\qquad-\frac{G(-a,-b)}{(p-a)(q-a)}\wh{\wt{S}}_{a,b}=0,
\end{align}
\end{subequations}
which is related to \eqref{eq:xyz-MSBSQ} through the transformation
\begin{subequations}
\label{eq:MSBSQ-tran}
\begin{align}
& S_{a,b}= \big( (p-a)/(p-b)\big)^n \big( (q-a)/(q-b)\big)^m x, \\
& v_a=(p-a)^n(q-a)^m y, \quad w_b=(p-b)^{-n}(q-b)^{-m} z, \\
& s_a=(p-a)^n(q-a)^m(y_1-z_0y), \\
& t_b=(p-b)^{-n}(q-b)^{-m}(z_1-z_0z),
\end{align}
\end{subequations}
where $z_0=-pn-qm-c_{1}$ with a constant $c_1$.
To bilinearize \eqref{eq:C3o2}, we employ both $\alpha$ and $\beta$ auxiliary independent variables.
By the following transformation
\begin{equation}
\label{C3-tran}
v_a=\frac{\dc{f}}{f},\quad w_b=\frac{\d{f}}{f}, \quad s_a=\frac{a\dc{f}+\dc{g}}{f},\quad t_b=\frac{\d{g}-b\d{f}}{f},
\quad  S_{a,b}=\frac{1}{b-a}\frac{\d{\dc{f}}}{f},
\end{equation}
from \eqref{eq:C3o2} we have its bilinear form
\begin{subequations}
\label{eq:C3-b0}
\begin{align}
\label{eq:C3-b0-ex_a}
& \mathcal{H}_{41}\equiv (p-a)\wt{f}\d{\dc{f}}-(p-b)f\wtd{\dc{f}}-(b-a)\d{f}\wt{\dc{f}}=0,\\
\label{eq:C3-b0-ex_b}
& \mathcal{H}_{42}\equiv (q-a)\wh{f}\d{\dc{f}}-(q-b)f\whd{\dc{f}}-(b-a)\d{f}\wh{\dc{f}}=0,\\
\label{eq:C3-b0-ex_c}
& \mathcal{H}_{43}\equiv f(a\wt{\dc{f}}+\wt{\dc{g}})+(p-a)\wt{f}\dc{f}-(pf+g)\wt{\dc{f}}=0,\\
\label{eq:C3-b0-ex_d}
& \mathcal{H}_{44}\equiv f(a\wh{\dc{f}}+\wh{\dc{g}})+(q-a)\wh{f}\dc{f}-(qf+g)\wh{\dc{f}}=0,\\
\label{eq:C3-b0-ex_e}
& \mathcal{H}_{45}\equiv \wt{f}(p\d{f}+\d{g})-(p-b)f\wtd{f}-(b\wt{f}+\wt{g})\d{f}=0,\\
\label{eq:C3-b0-ex_f}
& \mathcal{H}_{46}\equiv \wh{f}(q\d{f}+\d{g})-(q-b)f\whd{f}-(b\wh{f}+\wh{g})\d{f}=0,\\
\label{eq:C3-b0-ex_g}
& \mathcal{H}_{47}\equiv (p-b)\wtd{f}\wh{f}-(q-b)\whd{f}\wt{f}-(p-q)\wth{f}\d{f}=0,\\
\label{eq:C3-b0-ex_h}
& \mathcal{H}_{48}\equiv (a\dc{f}+\dc{g})\wthd{f}-\dc{f}(\wthd{g}-b\wthd{f})-\frac{p_b\whd{f}\wt{\dc{f}}}{(p-q)(p-a)}
+\frac{q_b\wtd{f}\wh{\dc{f}}}{(p-q)(q-a)}\nn\\
&\qquad\quad +(p+q-\alpha_3)\dc{f}\wthd{f}-\frac{a_b}{(p-a)(q-a)}f\wthd{\dc{f}}=0,
\end{align}
\end{subequations}
where  $a_b=p_b|_{p=a}$  and the connection with \eqref{eq:C3o2} is
\begin{subequations}
\begin{align*}
& C_{31}=\frac{\mathcal{H}_{41}}{(b-a)f\wt{f}},\quad C_{32}=\frac{\mathcal{H}_{42}}{(b-a)f\wh{f}}, \quad
C_{33}=\frac{\mathcal{H}_{43}}{f\wt{f}\wt{v}_a}-\frac{\mathcal{H}_{44}}{f\wh{f}\wh{v}_a}, \quad
C_{34}=\frac{\wt{\mathcal{H}}_{46}}{\wt{f}\wth{f}\wt{w}_b}
-\frac{\wh{\mathcal{H}}_{45}}{\wh{f}\wth{f}\wh{w}_b},\\
&C_{35}=-\frac{1}{f\wth{f}}\mathcal{H}_{48}-\frac{1}{(p-q)(p-a)(q-a)f\wth{f}}
\frac{(q-a)p_b\whd{f}\wt{\dc{f}}-(p-a)q_b\wtd{f}\wh{\dc{f}}}{\mathcal{H}_{47}
+ (p-q)\wth{f}\d{f}}\mathcal{H}_{47}.
\end{align*}
\end{subequations}

\section{Casoratian solutions and $\delta$-extended GD-4 type equations}\label{sec-4}

In this section, first, we present soliton solutions for the bilinear equations we derived in Sec.\ref{sec-3}.
Based on these soliton solutions,
we are able to extend the bilinear equations by introducing a parameter $\delta$,
and then recover the nonlinear $\delta$-extended lattice GD-4  type equations.
These equation are still consistent around the cube.
In addition, the parameter $\delta$ allows us to obtain quasi-rational solutions for the lattice GD-4  type equations.
We will study these equations one by one.
We start from the GD-4 (B-2) equation and
as an demonstrative example we will present details of the investigation,
while for GD-4 (A-2)  and (C-3) equations, we only list main results and sketch their proofs.

\subsection{GD-4 (B-2) equation}\label{sec-4-1}

\subsubsection{Soliton solutions in Casoratian}

We work on the bilinear GD-4 (B-2) system \eqref{eq:B2-h1}.
Consider scalar functions \eqref{eq:psij}, i.e.
\begin{equation}
\label{eq:psij-sec-4}
\psi_s{(n,m,\alpha,\beta,l)}=\sum_{j=1}^4\rho_{j,s}^{(0)}
(-\omega_j(k_s))^l(p-\omega_j(k_s))^n(q-\omega_j(k_s))^m(a-\omega_j(k_s))^\alpha(b-\omega_j(k_s))^\beta,
\end{equation}
and column vector
\begin{equation}\label{Ps}
\Ps(n,m,\alpha,\beta,l)=(\psi_1(n,m,\alpha,\beta,l),\psi_2(n,m,\alpha,\beta,l),\ldots,\psi_N(n,m,\alpha,\beta,l))^{\st}.
\end{equation}
Using such $\Ps(n,m,\alpha,\beta,l)$ as the basic column vector, we compose the following Casoratians
with respect to the shifts of $l$:
\begin{align}\label{ss-B2-f}
& f=|\wh{N-1}|,\quad g=|\wh{N-2},N|,\quad h=|\wh{N-2},N+1|, \quad s=|\wh{N-3},N-1,N|, \nn \\
& \mu=|\wh{N-2},N+2|,\quad \theta=|\wh{N-3},N-1,N+1|,\quad\nu=|\wh{N-4},N-2,N-1,N|.
\end{align}
Here we have employed the compact expression as described in \eqref{shorthand}.
For their explicit forms, for example, we have
\begin{align*}
& f=|\Ps(n,m,\alpha,\beta,l=0), \Ps(n,m,\alpha,\beta,l=1), \cdots, \Ps(n,m,\alpha,\beta,l=N-1)|, \\
& g=|\Ps(n,m,\alpha,\beta,0), \Ps(n,m,\alpha,\beta,1), \cdots, \Ps(n,m,\alpha,\beta,N-2), \Ps(n,m,\alpha,\beta,N)|.
\end{align*}
The number  $N$ represents the order of the Casoratians as well as the number of solitons.
For the above Casoratians  with order $N=1,2,3$, we list them  below respectively:
\begin{align*}
& f=|0|, ~g=|1|, ~h=|2|, ~s=0, ~\mu=|3|, ~\theta=0, ~\nu=0,\\
& f=|0,1|, ~g=|0,2|, ~h=|0,3|, ~s=|1,2|, ~\mu=|0,4|, ~\theta=|1,3|, ~\nu=0,\\
& f=|0,1,2|,~ g=|0,1,3|,~ h=|0,1,4|,~ s=|0,2,3|,~ \mu=|0,1,5|, ~\theta=|0,2,4|, ~\nu=|1,2,3|.
\end{align*}

The bilinear system \eqref{eq:B2-h1} consists of 8 equations but with only 7 involving functions.
However, it admits exact solutions formulated by the Casoratians given \eqref{ss-B2-f}.

\begin{thm}
\label{ss-B2}
The Casoratians defined in \eqref{ss-B2-f} with the entries given by \eqref{eq:psij-sec-4}
provide solutions to the bilinear system \eqref{eq:B2-h1}
\end{thm}

For the proof of this theorem, one can refer to the proof for Theorem \ref{prop Sigm4-B2}
in the following Sec.\ref{sec-3-2},
where a more general case with a parameter $\delta$ is considered.

\subsubsection{$\delta$-extension}\label{sec-4-1-2}

The lattice BSQ (B-2) equation allows a parameter extension \cite{HieZh-BSQ-2010,NZSZ},
which results from introducing a parameter $\delta$.
For the lattice GD-4 (B-2) equation, this can be done by replacing $\psi_s$ \eqref{eq:psij-sec-4}
by the $\delta$-dependent function \eqref{psi-jd},
i.e.
\begin{equation}
\label{psi-jd-sec-4}
\phi_s{(n,m,\alpha,\beta,l)}=\sum_{j=1}^4\rho_{j,s}^{(0)}
(\delta-\omega_j(k_s))^l(p-\omega_j(k_s))^n(q-\omega_j(k_s))^m(a-\omega_j(k_s))^\alpha(b-\omega_j(k_s))^\beta,
\end{equation}
and consider the Casoratians composed by the vector
\begin{equation}\label{Ph}
\Ph(n,m,\alpha,\beta,l)=(\phi_1(n,m,\alpha,\beta,l), \phi_2(n,m,\alpha,\beta,l), \cdots, \phi_N(n,m,\alpha,\beta,l))^{\st}.
\end{equation}
For the clarity and  convenience, we  denote those Casoratians composed by $\Ps$
as $f, g$, etc, while the Casoratians composed by $\Ph$
as $f', g'$, etc, i.e.
\begin{align}
\label{eq:moB2'}
& f'(\Ph)=|\wh{N-1}|,\quad g'(\Ph)=|\wh{N-2},N|,\quad h'(\Ph)=|\wh{N-2},N+1|, \quad s'(\Ph)=|\wh{N-3},N-1,N|, \nn \\
& \theta'(\Ph)=|\wh{N-3},N-1,N+1|, \quad \mu'(\Ph)=|\wh{N-2},N+2|,\quad \nu'(\Ph)=|\wh{N-4},N-2,N-1,N|.
\end{align}
Expanding $(\delta-\omega_j(k_s))^l$ in terms of $\delta$ indicates the following relations
\begin{subequations}
\label{eq:dand realtions}
\begin{align}
\label{fafd}
& f'=f,\quad g'=g+N\delta f,\\
\label{hahd}
& h'=h+(N+1)\delta g+\frac{(N+1)N}{2}\delta^2 f,\quad s'=s+(N-1)\delta g+\frac{N(N-1)}{2}\delta^2 f,\\
\label{vavd}
& \mu'=\mu+(N+2)\delta h+\frac{(N+2)(N+1)}{2}\delta^2 g+\frac{(N+2)(N+1)N}{6}\delta^3 f,\\
\label{wawd}
& \nu'=\nu+(N-2)\delta s+\frac{(N-1)(N-2)}{2}\delta^2 g+\frac{N(N-1)(N-2)}{6}\delta^3 f, \\
\label{theta_theta-d}
& \theta'=\theta+(N+1)\delta s+(N-1)\delta h+(N+1)(N-1)\delta^2 g+\frac{(N+1)N(N-1)}{3}\delta^3 f.
\end{align}
\end{subequations}
Substituting these relations into the bilinear system \eqref{eq:B2-h1} gives rise to
\begin{subequations}
\label{eq:B2-h2}
\begin{align}
\label{eq:GD4-B22-d-1}
& \mathcal{H}_{11}^\delta \equiv \wt{f}'(\sigma_1 g'+h')-\wt{g}'(\sigma_1 f'+g')+f'\wt{s}'=0, \\
\label{eq:GD4-B22-d-2}
& \mathcal{H}_{12}^\delta \equiv \wh{f}'(\sigma_2 g'+h')-\wh{g}'(\sigma_2 f'+g')+f'\wh{s}'=0, \\
\label{eq:GD4-B22-d-3}
& \mathcal{H}_{13}^\delta \equiv\wt{f}'(\sigma_1 h'+\mu')-\wt{h}'(\sigma_1 f'+g')+f'\wt{\theta}'=0, \\
\label{eq:GD4-B22-d-4}
& \mathcal{H}_{14}^\delta \equiv\wh{f}'(\sigma_2 h'+\mu')-\wh{h}'(\sigma_2 f'+g')+f'\wh{\theta}'=0, \\
\label{eq:GD4-B22-d-5}
& \mathcal{H}_{15}^\delta \equiv\wt{f}'(\sigma_1 s'+\theta')-f'(\sigma_1 \wt{s}'-\wt{\nu}')-\wt{g}'s'=0, \\
\label{eq:GD4-B22-d-6}
& \mathcal{H}_{16}^\delta \equiv\wh{f}'(\sigma_2 s'+\theta')-f'(\sigma_2 \wh{s}'-\wh{\nu}')-\wh{g}'s'=0, \\
\label{eq:GD4-B22-d-7}
& \mathcal{H}_{17}^\delta \equiv \wt{f}'\wh{g}'-\wh{f}'\wt{g}'+\sigma^{-}_{12}(\wt{f}'\wh{f}'-f'\wth{f}')=0, \\
\label{eq:GD4-B22-d-8}
& \mathcal{H}_{18}^\delta \equiv R(\sigma_1,\sigma_2)(f'\wth{f}'-\wt{f}'\wh{f}')
+Q(\sigma_1,\sigma_2)(g'\wth{f}'
-f'\wth{g}') \nn \\
& \qquad\quad -(\sigma^{+}_{12}+\epsilon_1)\big(g'\wth{g}'-f'\wth{h}'-\wth{f}'s'\big)+g'\wth{h}'-s'\wth{g}'
-f'\wth{\mu}'+\nu'\wth{f}'=0,
\end{align}
\end{subequations}
where $\{\sigma_j\}$, $\{\epsilon_j\}$, $Q$ and $R$ are defined in \eqref{2.18}.
Introduce new variables (cf.\eqref{eq:tr-GD-4})
\begin{align}
\label{eq:tr-GD-4'}
u_{0}'=\frac{g'}{f'}, \quad u_{1,0}'=\frac{h'}{f'}, \quad u_{0,1}'=\frac{s'}{f'}, \quad
u_{2,0}'=\frac{\mu'}{f'},\quad u_{0,2}'=\frac{\nu'}{f'}.
\end{align}
It then follows from  \eqref{eq:dand realtions} that
\begin{subequations}
\label{B2-B2_delta}
\begin{align}
& u_{0}=u_{0}'-N\delta, \\
& u_{1,0}=u_{1,0}'-(N+1)\delta u_{0}'+\frac{(N+1)N}{2}\delta^2,\\
& u_{0,1}=u_{0,1}'-(N-1)\delta u_{0}'+\frac{N(N-1)}{2}\delta^2,\\
& u_{2,0}=u_{2,0}'-(N+2)\delta u_{1,0}'+\frac{(N+2)(N+1)}{2}\delta^2u_{0}'-\frac{(N+2)(N+1)N}{6}\delta^3, \\
& u_{0,2}=u_{0,2}'-(N-2)\delta u_{0,1}'+\frac{(N-1)(N-2)}{2}\delta^2u_{0}'-\frac{N(N-1)(N-2)}{6}\delta^3,
\end{align}
\end{subequations}
and by them we are led to a $\delta$-extension of the GD-4 (B-2) equation \eqref{eq:B2o}:
\begin{subequations}
\label{eq:B2-2d}
\begin{align}
\label{eq:GD4-B2d-ex_a}
& B_{21}^\delta\equiv\sigma_1 u'_{0}+u'_{1,0}-(\sigma_1+u'_{0})\wt{u}'_{0}+\wt{u}'_{0,1}= 0, \\
\label{eq:GD4-B2-ex_b}
& B_{22}^\delta \equiv \sigma_2 u'_{0}+u'_{1,0}-(\sigma_2+u'_{0})\wh{u}'_{0}+\wh{u}'_{0,1}= 0, \\
\label{eq:GD4-B2-ex_c}
& B_{23}^\delta \equiv \wh{u}'_{2,0}-\wt{u}'_{2,0}+\sigma_1\wh{u}'_{1,0}-\sigma_2\wt{u}'_{1,0}-(\sigma^{-}_{12}
+\wh{u}'_0-\wt{u}'_0)\wth{u}'_{1,0} = 0, \\
\label{eq:GD4-B2-ex_d}
& B_{24}^\delta \equiv \wh{u}'_{0,2}-\wt{u}'_{0,2}+\sigma_1\wt{u}'_{0,1}-\sigma_2\wh{u}'_{0,1}-(\sigma^{-}_{12}
+\wh{u}'_0-\wt{u}'_0)u'_{0,1} = 0, \\
\label{eq:GD4-B2-ex_e}
& B_{25}^\delta \equiv 2\wth{u}'_{0,2} -\sigma^+_{12}(u'_{1,0} + \wth{u}'_{0,1})
+ (\sigma_1-\wth{u}'_{0})\wh{u}'_{0,1}+ (\sigma_2-\wth{u}'_{0})\wt{u}'_{0,1} \nonumber\\
& \qquad\quad   +(\sigma_1+u'_0)\wt{u}'_{1,0} + (\sigma_2 + u'_0)\wh{u}'_{1,0} - 2u'_{2,0} = 0, \\
\label{eq:GD4-B2-ex_f}
& B_{26}^\delta \equiv \frac{R(\sigma_1,\sigma_2)}{\sigma^{-}_{12}
+\wh{u}'_0-\wt{u}'_0}(\wh{u}'_0-\wt{u}'_0)
+ Q(\sigma_1,\sigma_2)(u'_0-\wth{u}'_0) \nn\\
& \qquad\quad
-(\sigma^{+}_{12}+\epsilon_1)(u'_0 \wth{u}'_0 - \wth{u}'_{1,0}
- u'_{0,1})+ u'_0 \wth{u}'_{1,0} - \wth{u}'_0 u'_{0,1} - \wth{u}'_{2,0} + u'_{0,2}= 0.
\end{align}
\end{subequations}
In this case, the analogue of \eqref{eq:B222} is
\begin{subequations}
\label{eq:B222'}
\begin{align}
\label{eq:BSQ-ex_a'}
& \wt{z}'=x'\wt{x}'-y', \quad \wh{z}'=x'\wh{x}'-y', \\
\label{eq:BSQ-ex_b'}
& \wh{\xi}'-\wt{\xi}'=(\wh{x}'-\wt{x}')\wth{y}', \\
\label{eq:BSQ-ex_bb'}
&  \wh{\eta}'-\wt{\eta}'=(\wh{x}'-\wt{x}')z', \\
\label{eq:B2-ex_eta-th'}
& \wth{\eta}'=\xi' + \frac{(\wh{x}'\wt{y}'-\wt{x}'\wh{y}')x'-y'(\wt{y}'-\wh{y}')}{\wt{x}'-\wh{x}'},\\
\label{eq:BSQ-ex_c'}
& \eta'=\wth{\xi}'-\wth{y}'x'+\wth{x}'z'-\epsilon_1(\wth{y}'-x'\wth{x}'+z')
+\epsilon_2(\wth{x}'-x')-\epsilon_3+\frac{P(\sigma_1,\sigma_2)}{\wh{x}'-\wt{x}'},
\end{align}
\end{subequations}
which is connected with \eqref{eq:B2-2d} via the transformations
\begin{subequations}
\label{B2'-tr}
\begin{align}
& u'_0=x'-x'_{0}, \quad u'_{1,0}=y'-x'_0u'_0-y'_0, \quad u'_{0,1}=z'-x'_0u'_0-z'_0, \\
& u'_{2,0}=\xi'-x'_0y'+x'z'_0+\mu', \quad u'_{0,2}=\eta'-x'_0z'+x'y'_0+\nu',
\end{align}
\end{subequations}
where
\begin{subequations}
\label{eq:xyz0'-ex-N4}
\begin{align}
& x'_0=-\sigma_1n-\sigma_2m-c'_1, \\
& y'_0= \frac{1}{2}\left((x'_0)^2+(n\sigma_1^2+m\sigma_2^2+c'_2)\right)+c'_3,\\
& z'_0= \frac{1}{2}\left((x'_0)^2-(n\sigma_1^2+m\sigma_2^2+c'_2)\right)-c'_3,\\
& \mu'_0=n\sigma_1^{3}+m\sigma_2^{3}+c'_4,\\
& \mu'=-\frac{1}{6}((x'_{0})^3 -3(y'_{0}-z'_{0})x'_{0} -2\mu'_0),\\
& \nu'=-\frac{1}{6}((x'_{0})^3 +3(y'_{0}-z'_{0})x'_{0} -2\mu'_0),
\end{align}
\end{subequations}
and $\{c'_{i}\}$ are constants.
Note that the parameter $\delta$ is involved in \eqref{eq:B222'} via $\epsilon_i$ and $P(\sigma_1,\sigma_2)$.
In addition, one can check that this quadrilateral system \eqref{eq:B222'} is CAC
along the line in \cite{Tela} for the GD-4 (B-2) equation \eqref{eq:B222}.

Obviously, the $\delta$-extended GD-4 (B-2) equation \eqref{eq:B2-2d} has solutions
in terms of the Casoratians defined in \eqref{eq:moB2'}.
However, in the following we will show that it allows more general solutions.
Consider the column vector $\Ph$  defined by the  following linear
equation set:
\begin{subequations}
\label{CES-B2}
\begin{align}
&\sigma_1\dt{\Ph}=\Ph-\bar{\dt{\Ph}}, \\
&\sigma_2\dh{\Ph}=\Ph-\bar{\dh{\Ph}}, \\
&\sigma_3\dc{\Ph}=\Ph-\bar{\dc{\Ph}}, \\
\label{sig4-dd}
&\sigma_4\dd{\Ph}=\Ph-\bar{\dd{\Ph}}, \\
\label{bKPh}
&\bK\Ph=\bar{\bar{\bar{\bar{\Ph}}}}-\epsilon_1\bar{\bar{\bar{\Ph}}}
+\epsilon_2\bar{\bar{\Ph}}-\epsilon_3\bar{\Ph},
\end{align}
\end{subequations}
where  $\{\sigma_j\}$ and $\{\epsilon_j\}$ are defined in \eqref{2.18},
and $\bK$ is a $N$-th order constant matrix.
When $\bK$ is diagonal (see later equation \eqref{diag-K}),
the above equation set has a solution \eqref{Ph} composed by \eqref{psi-jd-sec-4}.
Here $\bK$ can be more general, and so are $\Ph$ and the related Casoratians.

\begin{thm}
\label{prop Sigm4-B2}
The $\delta$-extended bilinear GD-4 (B-2) system   \eqref{eq:B2-h2} has Casoratian solutions:
\begin{align}
\label{eq:moB2}
& f'=|\wh{N-1}|,\quad g'=|\wh{N-2},N|,\quad h'=|\wh{N-2},N+1|, \quad s'=|\wh{N-3},N-1,N|, \nn \\
& \theta'=|\wh{N-3},N-1,N+1|, \quad \mu'=|\wh{N-2},N+2|,\quad \nu'=|\wh{N-4},N-2,N-1,N|,
\end{align}
which are composed by $\Ph$ define by \eqref{CES-B2}.
\end{thm}

\begin{proof}
It is notable that $\mathcal{H}_{11}^\delta$, $\mathcal{H}_{12}^\delta$ and
$\mathcal{H}_{17}^\delta$ have appeared in \cite{NZSZ}.
They can be proved to be zero in a similar way.
Thus here we just consider
the verifications of other equations in the system \eqref{eq:B2-h2}.

Let us  focus on  equation $\mathcal{H}_{13}^\delta$ and consider its down-tilde version:
\begin{align}
\label{eq:GD4-B22-delta_3-ut}
\dt{\mathcal{H}}_{13}^\delta\equiv f'(\sigma_1\dt{h}' + \dt{\mu}') - h'(\sigma_1\dt{f}' + \dt{g}') + \dt{f}'\theta',
\end{align}
in which
\[f'=|\wh{N-1}|,~~ g'=|\wh{N-2},N|,~~ h'=|\wh{N-2},N+1|,~~\theta'=|\wh{N-3},N-1,N+1|,~~\mu'=|\wh{N-2},N+2|.\]
In Appendix \ref{gd4:formulae} we have collected some formulae  for the (shifted) Casoratians,
which will be used in the proof.
For $\dt{f}'$, from \eqref{B2:fd} where we take $(i,j)=(1,N-2)$, we have
\begin{equation}\label{f'}
\dt{f}'=- \sigma_1^{-(N-2)}  |\wh{N-2},\dt\Ph(N-2)|.
\end{equation}
Here and after $\Ph(N-j) \doteq\Ph(n,m,\alpha,\beta,N-j)$ without making confusion.
For $\sigma_1\dt{f}' + \dt{g}'$ and $\sigma_1\dt{h}' + \dt{\mu}'$,
from \eqref{eq:fg-mu} with $i=1$, and \eqref{eq:hv-mu} with $i=1$, respectively,
we have
\begin{align*}
& \sigma_1\dt{f}' + \dt{g}'= -\sigma_1^{-(N-2)}  |\wh{N-3}, N-1, \dt\Ph(N-2)|, \\
&\sigma_1\dt{h}' + \dt{\mu}'=-\sigma_1^{-(N-2)} |\wh{N-3}, N+1, \dt\Ph(N-2)|.
\end{align*}
Then, substituting the above formulae into \eqref{eq:GD4-B22-delta_3-ut} yields
\begin{align*}
&-\sigma_1^{N-2}\big(f'(\sigma_1\dt{h}' + \dt{\mu}') - h'(\sigma_1\dt{f}' + \dt{g}')+ \dt{f}'\theta'\big)    \\
=~& |\wh{N-1}| |\wh{N-3}, N+1, \dt{\Ph}(N-2)| - |\wh{N-2}, N+1| |\wh{N-3}, N-1, \dt{\Ph}(N-2)|   \\
&\quad + |\wh{N-2}, \dt{\Ph}(N-2)| |\wh{N-3}, N-1, N+1|,
\end{align*}
which is zero in light of \eqref{plu-r-1} with
\begin{align*}
\bM=(\wh{N-3}), \quad \ba=\Ph(N-2), \quad \bb=\Ph(N-1), \quad \bc=\Ph(N+1), \quad \bd=\dt{\Ph}(N-2).
\end{align*}
Likewise, we can prove $\mathcal{H}_{14}^\delta$.

Next, we look at $\mathcal{H}_{15}^\delta$, the down-tilde shift of which is
\begin{align}
\label{eq:B2-B5-dt}
\dt{\mathcal{H}}_{15}^\delta \equiv f'(\sigma_1\dt{s}'+\dt{\theta}')-\dt{f}'(\sigma_1 s'-\nu')-g'\dt{s}'.
\end{align}
For the involved Casoratians,
in addition to those in \eqref{eq:moB2},
$\dt{f}'$ is presented in \eqref{f'},
by \eqref{eq:gs-mu} with $i=1$ and \eqref{eq:ssita-mu} with $i=1$, respectively, we have
\begin{align*}
&\sigma_1^{N-3}\dt s' = |\wh{N-4}, N-2, N-1, \dt\Ph(N-3)|+|\wh{N-3}, N-1, \dt\Ph(N-2)|, \\
& \sigma_1^{N-3}  (\dt\theta'+\sigma_1 \dt s') = |\wh{N-4}, N-2, N, \dt\Ph(N-3)|+|\wh{N-3}, N, \dt\Ph(N-2)|.
\end{align*}
After substituting them into \eqref{eq:B2-B5-dt}, we can split the resulting equation into two parts, i.e.,
\begin{align*}
& \sigma_1^{N-3}(f'(\sigma_1\dt{s}'+\dt{\theta}')-\dt{f}'(\sigma_1 s'-\nu')-g'\dt{s}') \\
 =~&\Big(|\wh{N-1}||\wh{N-4},N-2,N,\dt{\Ph}(N-3)|-|\wh{N-2},N||\wh{N-4},N-2,N-1,\dt{\Ph}(N-3)| \\
& \quad +|\wh{N-2},\dt{\Ph}(N-3)||\wh{N-4},N-2,N-1,N|\Big)\\
~ & +\Big(|\wh{N-1}||\wh{N-3},N,\dt{\Ph}(N-2)|
 -|\wh{N-2},N||\wh{N-3},N-1,\dt{\Ph}(N-2)|\\
& \quad +|\wh{N-2},\dt{\Ph}(N-2)|||\wh{N-3},N-1,N|\Big),
\end{align*}
of which each part is zero by applying the identity \eqref{plu-r-1} with, respectively,
\begin{align*}
& \bM=(\wh{N-4},N-2),~\ba=\Ph(N-3),~\bb=\Ph(N-1),~\bc=\Ph(N),~\bd=\dt{\Ph}(N-3), \\
& \bM=(\wh{N-3}),~ \ba=\Ph(N-2),~\bb=\Ph(N-1),~\bc=\Ph(N),~\bd=\dt{\Ph}(N-2).
\end{align*}
In a similar manner, we can also prove \eqref{eq:GD4-B22-d-6}.

Finally, we consider the last equation \eqref{eq:GD4-B22-d-8}.
The proof is long.
First, we reformulate $\dth{\mathcal{H}}_{18}^\delta$ in the following manner
\begin{align}
\label{dthH18}
&\sigma^{-}_{12}\sigma_1^{N-2}\sigma_2^{N-2}\dth{\mathcal{H}}_{18}^\delta \nn \\
=~&
\sigma^{-}_{12}\sigma_1^{N-2}\sigma_2^{N-2}\{f'[\sigma_{12}^{+}(\sigma_1^2+\sigma_2^2)\dth{f}'
+ (\sigma^{+^2}_{12}-\sigma_1\sigma_2)\dth{g}'
+\sigma_{12}^{+}\dth{s}'+\dth{\nu}'] \nn\\
& +(\epsilon_1f'-g')[(\sigma^{+^2}_{12}-\sigma_1\sigma_2)\dth{f}'+\sigma_{12}^{+}\dth{g}'
+\dth{s}']+(\epsilon_2f'-\epsilon_1g'+h')\nn\\
& \cdot (\sigma_{12}^{+}\dth{f}'+\dth{g}')+[\epsilon_3f'-\epsilon_2g'+\epsilon_1h'-\mu']\dth{f}'
-R(\sigma_1,\sigma_2)\dt{f}'\dh{f}'\}.
\end{align}
Apart from the Casoratians given in \eqref{eq:moB2}, using the formulae \eqref{eq:f-th}-\eqref{eq:fgws-th}, respectively,
we have
\begin{align*}
& \sigma_1^{N-2}\sigma_2^{N-2}\sigma^-_{12}\dth{f}'= |\wh{N-3}, \dh{\Ph}(N-2),\dt{\Ph}(N-2)|, \\
&\sigma_1^{N-2}\sigma_2^{N-2}\sigma^-_{12}\Bigl(\,\dth{g}'+\sigma^+_{12}\dth{f}'\Bigr)=
\sigma_2^{N-2}\dh{f}'-\sigma_1^{N-2}\dt{f}'+|\wh{N-4}, N-2, \dh{\Ph}(N-2),\dt{\Ph}(N-2)|,
\\
&\sigma_1^{N-2}\sigma_2^{N-2}\sigma^-_{12}\Bigl(\,\dth{s}'+\sigma^+_{12}\dth{g}'
+(\sigma^{+^2}_{12}-\sigma_1\sigma_2)\dth{f}' \Bigr) \nn \\
&~~~ = \sigma_1\sigma_2^{N-2}\dh{f}'-\sigma_1^{N-2}\sigma_2\dt{f}'+|\wh{N-5}, N-3, N-2,\dh{\Ph}(N-2),\dt{\Ph}(N-2)|, \\
&\sigma_1^{N-2}\sigma_2^{N-2}\sigma^-_{12}\Big(\, \dth{\nu}'+\sigma^+_{12}\dth{s}'
+(\sigma^{+^2}_{12}-\sigma_1\sigma_2)\dth{g}'
+\sigma^+_{12}(\sigma_1^2+\sigma_2^2)\dth{f}' \Big) \nn \\
&~~~ =\sigma_1^2 \sigma_2^{N-2} \dh{f}'-\sigma_1^{N-2}\sigma_2^2 \dt{f}'+|\wh{N-6}, N-4, N-3,N-2,\dh{\Ph}(N-2),\dt{\Ph}(N-2)|.
\end{align*}
For the last term in \eqref{dthH18}, we replace its coefficient $\sigma^-_{12}R(\sigma_1,\sigma_2)$ using \eqref{sig-S}:
\[\sigma^-_{12}R(\sigma_1,\sigma_2)=S(\sigma_1)-S(\sigma_2),\]
where $S(\sigma_i)$ is defined in \eqref{eq:Qt}.
Thus we obtain
\begin{align}
\label{h28 hff}
\sigma^{-}_{12}\sigma_1^{N-2}\sigma_2^{N-2}\dth{\mathcal{H}}_{18}^\delta
=~& f'|\wh{N-6},N-4,N-3,N-2,\dh{\Ph}(N-2),\dt{\Ph}(N-2)| \nn \\
&+(\epsilon_1f'-g')|\wh{N-5},N-3,N-2,\dh{\Ph}(N-2),\dt{\Ph}(N-2)| \nn\\
&+(\epsilon_2f'-\epsilon_1g'+h')|\wh{N-4},N-2,\dh{\Ph}(N-2),\dt{\Ph}(N-2)|\nn\\
&+(\epsilon_3f'-\epsilon_2g'+\epsilon_1h'-\mu')|\wh{N-3},\dh{\Ph}(N-2),\dt{\Ph}(N-2)| \nn \\
&+\sigma_1^{N-2}\dt{f}'[\sigma_2^{N+2}\dh{f}'-\sigma_2^2 f'
+\sigma_2 g'-h'+\epsilon_1(\sigma_2^{N+1}\dh{f}'-\sigma_2f'+g') \nn \\
&+\epsilon_2(\sigma_2^N\dh{f}'-f')+\epsilon_3\sigma_2^{N-1}\dh{f}']
-\sigma_2^{N-2}\dh{f}'[\sigma_1^{N+2}\dt{f}'-\sigma_1^2 f' \nn \\
& +\sigma_1 g'-h'+\epsilon_1(\sigma_1^{N+1}\dt{f}'-\sigma_1 f'+g')+\epsilon_2(\sigma_1^N\dt{f}'-f')+\epsilon_3\sigma_1^{N-1}\dt{f}'].
\end{align}
Let us sketch the rest part of the proof.
We replace the term $f'|\wh{N-6},N-4,N-3,N-2,\dh{\Ph}(N-2),\dt{\Ph}(N-2)|$ by
\[\sigma_1^{N-2}\dt{f'}|\wh{N-6},N-4,N-3,N-2,N-1,\dh{\Ph}(N-2)|
-\sigma_2^{N-2}\dh{f'}|\wh{N-6},N-4,N-3,N-2,N-1,\dt{\Ph}(N-2)|\]
through \eqref{eq:B2fN-6}, and proceed in the same manner
for the replacements through \eqref{eq:B2fN-5}-\eqref{eq:vN-3}.
For $\sigma_1^{N-1}\dt{f}'$ and $\sigma_2^{N-1}\dh{f}'$, we use \eqref{B2:fd} for $(i,j)=(1,N-1)$
and $(i,j)=(2,N-1)$, respectively.
After that,  for $\sigma_1^N\dt{f}'-f'$, $\sigma_2^N\dh{f}'-f'$, $\sigma_1^{N+1}\dt{f}'-\sigma_1 f'+g'$,
$\sigma_2^{N+1}\dh{f}'-\sigma_2 f'+g'$, $\sigma_1^{N+2}\dt{f}'-\sigma_1^2 f'+\sigma_2 g'-h'$ and
$\sigma_2^{N+2}\dh{f}'-\sigma_2^2 f' +\sigma_2 g'-h'$, we use \eqref{fd-N} with $i=1,2$,
\eqref{fd-N+1} with $i=1,2$, and \eqref{fd-N+2} with $i=1,2$, respectively.
Finally, using the identity \eqref{Ga4 B2} with $i=1,2$,
we can reduce \eqref{h28 hff} to
\begin{align}
\sigma^{-}_{12}\sigma_1^{N-2}\sigma_2^{N-2}\dth{\mathcal{H}}_{18}^\delta
=\sigma_1^{N-2}\dt{f'}(\text{tr}(\bK))\sigma_2^{N-2}\dh{f'}-\sigma_2^{N-2}\dh{f'}
(\text{tr}(\bK))\sigma_1^{N-2}\dt{f'}=0,
\end{align}
where $\text{tr}(\bK)$ represents the trace of $\bK$.

The proof is thereby completed.

\end{proof}

\subsubsection{Examples of exact solutions}\label{sec-4-1-3}

In the following we will briefly consider some typical examples of
the exact solutions of  the $\delta$-extended GD-4 (B-2) equation \eqref{eq:B2-2d}.
Since the transformation \eqref{B2-B2_delta}  connects
the GD-4 (B-2) equation \eqref{eq:B2o} and its $\delta$-extension  \eqref{eq:B2-2d} is invertible,
these solutions can be easily reduced to those for the GD-4 (B-2) equation \eqref{eq:B2o}.

The  first step is to solve the equation set \eqref{CES-B2}.
From \eqref{CES-B2} we have a formal expression for its solution:
\begin{equation}
\label{Ph}
\Ph_{[\bK]}=\sum_{j=1}^4\Rho_{j}^{(0)}
(\delta\bI-\Ome_j(\bK))^l(p\bI-\Ome_j(\bK))^n(q\bI-\Ome_j(\bK))^m(a\bI
-\Ome_j(\bK))^\alpha(b\bI-\Ome_j(\bK))^\beta,
\end{equation}
in which $\Ome_j(\bK), j=1,2,3,4$ are solutions to the matrix equation
\begin{align}
\label{G-curve-4-ma}
\bG(\Ome,\bK):=\Ome^4-\bK^4+\alpha_3(\Ome^3-\bK^3)+\alpha_2(\Ome^2-\bK^2)+\alpha_1(\Ome-\bK)=0,
\end{align}
where $\bI$ is the $N$-th order identity matrix and $\Rho_{j}^{(0)}, j=1,2,3,4$ are constant column vectors.
Note that $\{\Ome_j(\bK)\}_{j=1}^{4}$ commute with respect to matrix multiplication.

$\bK$ is an arbitrary constant matrix in \eqref{Ph}.
We claim that $\bK$ and its similar matrix $\bL$ yield same solutions to
the $\delta$-extended GD-4 (B-2) equation \eqref{eq:B2-2d}.
In fact, we suppose $\bL$ is any matrix that is similar to $\bK$, i.e.,
\begin{align}
\label{Ga-K}
\bL=\bT\bK\bT^{-1},
\end{align}
where $\bT$ is the transform matrix. Taking \eqref{Ga-K} into \eqref{Ph} and noting that
$\Ome_j(\bL)=\bT\Ome_j(\bK)\bT^{-1}$, we replace $\Rho_{j}^{(0)}\bT^{-1}\rightarrow \Rho_{j}^{(0)}$ and then have
$\Ph_{[\bL]}=\Ph_{[\bK]}\bT^{-1}$.
It is apparent that $\Ph_{[\bL]}$ still satisfies   \eqref{CES-B2}
with $\bK\rightarrow \bL$.
In light of the following relations
\begin{align*}
& f'(\Ph_{[\bL]})=|\bT|f'(\Ph_{[\bK]}),~~ g'(\Ph_{[\bL]})=|\bT|g'(\Ph_{[\bK]}),\\
&
 h'(\Ph_{[\bL]})=|\bT|h'(\Ph_{[\bK]}),~~ s'(\Ph_{[\bL]})=|\bT|s'(\Ph_{[\bK]}), \\
& \mu'(\Ph_{[\bL]})=|\bT|\mu'(\Ph_{[\bK]}),~~
\nu'(\Ph_{[\bL]})=|\bT|\nu'(\Ph_{[\bK]}),~~ \theta'(\Ph_{[\bL]})=|\bT|\theta'(\Ph_{[\bK]}),
\end{align*}
as well as the transformation \eqref{eq:tr-GD-4'}
(or the gauge property of bilinear equations),
one can easily find that
the above Casoratians composed by $\Ph_{[\bK]}$ and by $\Ph_{[\bL]}$
lead to same solutions for the $\delta$-extended GD-4 (B-2) equation \eqref{eq:B2-2d}.
In what follows we can suppose $\bK$ in \eqref{Ph} is its canonical form.

Now let us start from the following column vector:
\begin{equation}
\label{Ph-Ga}
\Ph=\sum_{j=1}^4\Rho_{j}^{(0)}
(\delta\bI-\Ome_j(\Ga))^l(p\bI-\Ome_j(\Ga))^n(q\bI-\Ome_j(\Ga))^m(a\bI-\Ome_j(\Ga))^\alpha
(b\bI-\Ome_j(\Ga))^\beta,
\end{equation}
where  $\bK$ in \eqref{Ph} is already replaced by its  canonical form $\Ga$.
Then various exact solutions, including soliton solutions, Jordan-block solutions and quasi-rational solutions,
for the $\delta$-extended GD-4 (B-2) equation \eqref{eq:B2-2d}
can be derived from different eigenvalue structures of matrix $\Ga$.
Let us list them case by case.

\vskip 6pt

\noindent\textbf{Soliton solutions:} If $\Ga$ is a diagonal matrix
\begin{subequations}\label{diag-K}
\begin{align}
\label{Ga-diag}
\Ga=\text{diag}(\gamma_1, \gamma_2, \dots, \gamma_N)
\end{align}
with
\begin{align}
\label{e in Gamma}
\gamma_s=-(\delta^4-k_s^4)+\alpha_3(\delta^3-k_s^3)-\alpha_2(\delta^2-k_s^2)+\alpha_1(\delta-k_s),
\quad s=1, 2 ,\dots, N,
\end{align}
\end{subequations}
then the basic elements $\{\phi_s{(n,m,\alpha,\beta,l)}\}$
are given by \eqref{psi-jd-sec-4}. The resulting Casoratians \eqref{eq:moB2} yield
the $N$-soliton solutions. In particular, the one-soliton solution reads
\begin{align}
\label{B2-1ss}
u'_{0}=\frac{\tau^{[1]}_1}{\tau_1}, \quad u'_{1,0}=\frac{\tau^{[2]}_1}{\tau_1}, \quad u'_{2,0}=\frac{\tau^{[3]}_1}{\tau_1}, \quad
u'_{0,1}=0,\quad u'_{0,2}=0,
\end{align}
where
\begin{subequations}\label{3.30}
\begin{equation}\label{3.30-a}
\tau_i=\sum_{j=1}^4\rho_{i,j},~~~ \tau^{[\ell]}_i=\sum_{j=1}^4\zeta^{\ell}_{i,j}\rho_{i,j},
~~~\zeta_{i,j}=\delta-\omega_j(k_i)
\end{equation}
and $\rho_{i,j}$ are the plane wave factors defined as
\begin{align}
\label{pwf}
\rho_{i,j}=\rho_{i,j}^{(0)}(p-\omega_j(k_i))^n(q-\omega_j(k_i))^m(a-\omega_j(k_i))^\alpha(b-\omega_j(k_i))^\beta,
~~  \rho_{i,j}^{(0)}\in \mathbb{C}.
\end{align}
\end{subequations}
Besides, the two-soliton solutions can be given by \eqref{eq:moB2} together with
\begin{subequations}
\label{B2d-fghs-2SS}
\begin{align}
\label{B2d-f-2SS}
& f'=\sum_{i=1}^2(-1)^{i+1}\tau_i\tau^{[1]}_{i+1}, \quad
g'=\sum_{i=1}^2(-1)^{i+1}\tau_i\tau^{[2]}_{i+1}, \quad
h'=\sum_{i=1}^2(-1)^{i+1}\tau_i\tau^{[3]}_{i+1}, \\
\label{B2d-s-2SS}
& \mu'=\sum_{i=1}^2(-1)^{i+1}\tau_i\tau^{[4]}_{i+1}, \quad
s'=\sum_{i=1}^2(-1)^{i+1}\tau^{[1]}_{i}\tau^{[2]}_{i+1}, \quad
\theta'=\sum_{i=1}^2(-1)^{i+1}\tau^{[1]}_{i}\tau^{[3]}_{i+1},
\end{align}
\end{subequations}
and $\nu'=0$, where $\zeta_{3,j}=\zeta_{1,j}$ and $\rho_{3,j}=\rho_{1,j}$.
It is worthy to note that all Casoratians are non-zero for $N\geq 3$.

\vskip 6pt
\noindent\textbf{Jordan-block solutions:} If $\Ga$ is a lower triangular Toeplitz matrix of the form
\begin{align}
\label{Ga-Jor}
\Ga=(\gamma_{s,l}(k_1))_{N \times N} , \quad
\gamma_{s,l}(k_1) = \left\{
\begin{aligned}
& \frac{\partial_{k_{1}}^{s-l}\gamma_1}{(s-l)!}, &&  s \geq l, \\
& 0,&& s < l,
\end{aligned}
\right.
\end{align}
then the basic column vector $\Ph$ can be expressed as
\begin{align}
\label{Ph-Jor}
\Ph=\sum_{j=1}^4 \mathcal{A}_j \mathcal{P}_j,
\end{align}
where $\mathcal{P}_j=(\mathcal{P}_{j,0},\mathcal{P}_{j,1},\dots,\mathcal{P}_{j,N-1})^{\st}$ with
\begin{align*}
\mathcal{P}_{j,s}=\frac{1}{s!}\partial_{k_{1}}^{s}[\rho^{(0)}_{j,1}(\delta-\omega_j(k_1))^l
(p-\omega_j(k_1))^n
(q-\omega_j(k_1))^m(a-\omega_j(k_1))^\alpha(b-\omega_j(k_1))^\beta],
\end{align*}
and $\{\mathcal{A}_j\}$ are lower triangular Toeplitz matrices of the form \eqref{A}.
In this case, Jordan-block solutions (i.e. multiple-pole solutions)
can be obtained.

\vskip 6pt

\noindent\textbf{Quasi-rational solutions:} In principle, this type of solutions can be obtained
from the Jordan-block solutions by taking limit $k_1\rightarrow 0$ (cf. \cite{NZSZ}). Thus we
take $\Ga$ as
\begin{align}
\label{Ga-Jor-0}
\Ga=(\gamma_{s,l}(0))_{N \times N} , \quad
\gamma_{s,l}(0) = \left\{
\begin{aligned}
& \frac{\partial_{k_{1}}^{s-l}\gamma_1}{(s-l)!}\bigg|_{k_1=0}, &&  s \geq l, \\
& 0,&& s < l.
\end{aligned}
\right.
\end{align}
In this case, quasi-rational solutions can be generated via
\begin{align}
\label{qrpsi}
\Ph=\sum_{j=1}^4 \mathcal{A}_j \mathcal{Q}_j,
\end{align}
in which $\mathcal{Q}_j=(\mathcal{Q}_{j,0},\mathcal{Q}_{j,1},\dots,\mathcal{Q}_{j,N-1})^{\st}$ with
\begin{align*}
\mathcal{Q}_{j,s}=\frac{1}{s!}\partial_{k_{1}}^{s}[\rho^{(0)}_{j,1}(\delta-\omega_j(k_1))^{l+l_0}
(p-\omega_j(k_1))^n
(q-\omega_j(k_1))^m(a-\omega_j(k_1))^\alpha(b-\omega_j(k_1))^\beta]|_{k_1=0},
\end{align*}
where $l_0$ is either an integer large enough or a non-integer so that the derivative
$\partial_{k_{1}}^{s}(\delta-\omega_j(k_1))^{l+l_0}\neq 0$ (cf. \cite{NZSZ}).
Here $\{\mathcal{A}_j\}$ are lower triangular Toeplitz matrices.

\subsection{GD-4 (A-2) equation}\label{sec-4-2}

Soliton solutions of Casoratian type for the bilinear GD-4 (A-2) system \eqref{eq:A21-h} are
presented in the following theorem.
 Its proof will be given in later together with Theorem \ref{prop Sigm4-A21}.
\begin{thm}
\label{ss-A21}
The bilinear GD-4 (A-2) system \eqref{eq:A21-h} admits Casoratian solutions
\begin{align}
\label{ss-A21-f}
f=|\wh{N-1}|,\quad g=|\wh{N-2},N|,\quad h=|\wh{N-2},N+1|,\quad s=|\wh{N-3},N-1,N|
\end{align}
composed by the vector $\Ps$ in \eqref{Ps} with entries \eqref{eq:psij-sec-4}
\end{thm}

Let us consider $\delta$-extension of the above Casoratians and related equations.
Note that in Theorem \ref{ss-A21} the Casoratians $f, g, h$ and $s$
are the same as those in Theorem \ref{ss-B2} (see \eqref{ss-B2-f}).
This means we can introduce $f', g',  h'$ and $s'$ as in \eqref{eq:moB2'}
and define (cf. \eqref{A21-tran})
\begin{align}
\label{A21-tran'}
& u'_{0}=\frac{g'}{f'}, \quad u'_{1,0}=\frac{h'}{f'}, \quad
v'_a=\frac{\dc{f'}}{f'}, \quad s'_a=\frac{\sigma_3\dc{f'}+\dc{g'}}{f'}, \quad r'_a=\frac{\sigma_3^2\dc{f'}+\sigma_3\dc{g'}+\dc{s'}}{f'}.
\end{align}
Note also that the relations \eqref{eq:GD4-B22-d-1} and \eqref{eq:GD4-B22-d-2}
hold as well for this case.
Then we find
\begin{subequations}
\label{eq:B2d_B2}
\begin{align}
& u_{0}=u'_{0}-N\delta, \\
& u_{1,0}=u'_{1,0}-(N+1)\delta u'_{0}+\frac{(N+1)N}{2}\delta^2, \quad u_{0,1}=u'_{0,1}-(N-1)\delta  u'_{0}+\frac{N(N-1)}{2}\delta^2,\\
\label{eq:vava'}
& v_a=v'_a, \quad s_a=s'_a-(N-1)\delta v'_a, \quad r_a=r'_a-(N-2)\delta s'_a+\frac{(N-1)(N-2)}{2}v'_a.
\end{align}
\end{subequations}
As a result, we obtain the $\delta$-extension of \eqref{eq:A21-h}, which takes the form
\begin{subequations}
\label{eq:A2hdf}
\begin{align}
\label{eq:A21hd_a}
\mathcal{H}_{21}^{\delta}\equiv &f'(\sigma_3\wt{\dc{f}}'+\wt{\dc{g}}')+\sigma^-_{13}\wt{f}'\dc{f}'
-(\sigma_1 f'+g')\wt{\dc{f}}'=0,\\
\label{eq:A21hd_b}
\mathcal{H}_{22}^{\delta}\equiv &f'(\sigma_3\wh{\dc{f}}'+\wh{\dc{g}}')+\sigma^-_{23}\wh{f}'\dc{f}'
-(\sigma_2 f'+g')\wh{\dc{f}}'=0,\\
\label{eq:A21hd_c}
\mathcal{H}_{23}^{\delta}\equiv &f'(\sigma_3^2\wt{\dc{f}}'+\sigma_3\wt{\dc{g}}'+\wt{\dc{s}}')
-\sigma_1 f'(\sigma_3\wt{\dc{f}}'+\wt{\dc{g}}')+\sigma^-_{13}\wt{f}'(\sigma_3\dc{f}'+\dc{g}')-\wt{\dc{f}}'s'=0, \\
\label{eq:A21hd_d}
\mathcal{H}_{24}^{\delta}\equiv& f'(\sigma_3^2\wh{\dc{f}}'+\sigma_3\wh{\dc{g}}'+\wh{\dc{s}}')
-\sigma_2 f'(\sigma_3\wh{\dc{f}}'+\wh{\dc{g}}')+\sigma^-_{23}\wh{f}'(\sigma_3\dc{f}'+\dc{g}')-\wh{\dc{f}}'s'=0, \\
\label{eq:A21hd_e}
\mathcal{H}_{25}^{\delta}\equiv &\wt{f}'(\sigma_1 g'+h')-\wt{g}'(\sigma_1 f'+g')+f'\wt{s}'=0,\\
\label{eq:A21hd_f}
\mathcal{H}_{26}^{\delta}\equiv &\wh{f}'(\sigma_2 g'+h')-\wh{g}'(\sigma_2 f'+g')+f'\wh{s}'=0, \\
\label{eq:A21hd_g}
\mathcal{H}_{27}^{\delta}\equiv& \sigma^-_{12}(\wt{f}'\wh{f}'-f'\wth{f}')+\wt{f}'\wh{g}'-\wh{f}'\wt{g}'=0,
\\
\label{eq:A21hd_h}
\mathcal{H}_{28}^{\delta}\equiv&\big[Q(\sigma_1,\sigma_2)\dc{f}'\wth{f}'+\dc{f}'\wth{h}'
-(\sigma^+_{12}+\epsilon_1)\dc{f}'\wth{g}'+((\sigma^+_{12}+\epsilon_1)\wth{f}'-\wth{g}') (\sigma_3\dc{f}'+\dc{g}')\nn\\
& +\wth{f}'(\sigma_3^2\dc{f}'+\sigma_3\dc{g}'+\dc{s}')\big]\sigma^-_{12}
-R(\sigma_1,\sigma_3)\wh{f}'\wt{\dc{f}}'+R(\sigma_2,\sigma_3)\wt{f}'\wh{\dc{f}}'=0,
\end{align}
\end{subequations}
and the nonlinear form, namely, the $\delta$-extened GD-4 (A-2) equation \eqref{eq:A21o}, reads
\begin{subequations}
\label{eq:A21-o2}
\begin{align}
\label{eq:GD4-A21-o2-ex_a}
A_{11}^{\delta}\equiv&\wt{s}'_a-(\sigma_1+u'_{0})\wt{v}'_{a}+\sigma^-_{13}v'_{a}=0,\\
\label{eq:GD4-A21-o2-ex_b}
A_{12}^{\delta}\equiv&\wh{s}'_a-(\sigma_2+u'_0)\wh{v}'_a+\sigma^-_{23} v'_a=0, \\
\label{eq:GD4-A21-o2-ex_c}
A_{13}^{\delta}\equiv& (\wt{r}'_a-\sigma_1\wt{s}'_a)/\wt{v}'_a
+\sigma^-_{13}s'_a/\wt{v}'_a-(\wh{r}'_a-\sigma_2\wh{s}'_a)/\wh{v}'_a-\sigma^-_{23}s'_a/\wh{v}'_a=0, \\
\label{eq:GD4-A21-o2-ex_d}
A_{14}^{\delta}\equiv&2 \wth{r}'_a-\sigma^+_{12}\wth{s}'_a+\sigma^-_{13}\wh{s}'_a
 +\sigma^-_{23}\wt{s}'_a-\wth{v}'_a[\sigma_1\wt{u}'_0+\sigma_2\wh{u}'_0 \nn\\
& -(\sigma^+_{12} -\wt{u}'_0-\wh{u}'_0)u'_0-2u'_{1,0}]=0, \\
\label{eq:GD4-A21-o2-ex_e}
A_{15}^{\delta}\equiv&\wh{u}'_{1,0}-\wt{u}'_{1,0}-(\sigma^-_{12}+\wh{u}'_0-\wt{u}'_0)\wth{u}'_0
+\sigma_1\wh{u}'_0-\sigma_2\wt{u}'_0=0, \\
\label{eq:GD4-A21-o2-ex_f}
A_{16}^{\delta}\equiv&\big[Q(\sigma_1,\sigma_2)+\wth{u}'_{1,0}-(\sigma^+_{12}+\epsilon_1)\wth{u}'_{0}
+((\sigma^+_{12}+\epsilon_1-\wth{u}'_{0})s'_a+r'_a)/v'_a\big] \nn \\
&\cdot (\sigma^-_{12}+\wh{u}'_0-\wt{u}'_0)-(R(\sigma_1,\sigma_3)\wt{v}'_a-R(\sigma_2,\sigma_3)\wh{v}'_a)/v'_a=0.
\end{align}
\end{subequations}
We may also recover Hietarinta's version (cf. Eq.\eqref{eq:A2-111}):
\begin{subequations}
\label{eq:A2-111'}
\begin{align}
\label{eq:A2-ex_a'}
& \wt{y}'=z'\wt{x}'-x', \quad \wh{y}'=z'\wh{x}'-x', \\
\label{eq:A2-ex_b'}
& \wt{\eta}'\wh{x}'-\wh{\eta}'\wt{x}'=y'(\wt{x}'-\wh{x}'), \\
\label{eq:A2-ex_bb'}
& \wh{\xi}'-\wt{\xi}'=(\wh{z}'-\wt{z}')\wth{z}', \\
\label{eq:A2-ex_eta-th'}
& \wth{\eta}'=\frac{(\wt{y}'\wh{z}'-\wh{y}'\wt{z}')z'-\xi'(\wt{y}'-\wh{y}')}{z'(\wt{z}'-\wh{z}')},\\
\label{eq:A2-ex_c'}
& \eta'=-x'\wth{\xi}'+y'\wth{z}'-\epsilon_1(y'-x'\wth{z}')-\epsilon_2x'
+\frac{P(\sigma_1,\sigma_3)\,\wt{x}'-P(\sigma_2,\sigma_3)\,\wh{x}'}{\wh{z}'-\wt{z}'},
\end{align}
\end{subequations}
which is linked to \eqref{eq:A21-o2} via the transformation
\begin{subequations}
\label{eq:xyz-A2'}
\begin{align}
& v'_a=x'/x'_{a}, \quad u'_0=z'-z'_0, \quad s'_a=(y'-v'_ay'_a)/x'_a, \\
& u'_{1,0}=\xi'-z'_0u'_0-\xi'_0, \quad r'_a =(\eta'-z'_0y'+\xi'_0x')/x'_a,
\end{align}
\end{subequations}
where
\begin{subequations}
\label{eq:xyz-0-A2'}
\begin{align}
& x'_a=(\sigma_1-\sigma_3)^{-n}(\sigma_2-\sigma_3)^{-m}c'_0, \quad z'_0=-\sigma_1n-\sigma_2m-c'_1, \quad y'_a=x'_az'_0, \\
& \xi'_0= \left((z'_0)^2+(n\sigma_1^2+m\sigma_2^2+c'_2)\right)/2+c'_3,
\end{align}
\end{subequations}
and $\{c'_i\}$ are constants. The CAC property of \eqref{eq:A2-111'} can be checked as well.

The more general Casoratian solutions of the $\delta$-extended bilinear form \eqref{eq:A2hdf} are
presented in the following theorem.

\begin{thm}
\label{prop Sigm4-A21}
The functions
\begin{align}
\label{eq:ss-A21d}
f'=|\wh{N-1}|,\quad g'=|\wh{N-2},N|,\quad h'=|\wh{N-2},N+1|, \quad s'=|\wh{N-3},N-1,N|
\end{align}
solve the bilinear equation \eqref{eq:A2hdf}, provided that the column vector $\Ph$ is defined by
the equation set  \eqref{CES-B2}.
\end{thm}

\begin{proof}
Equations \eqref{eq:A21hd_a}, \eqref{eq:A21hd_b} and \eqref{eq:A21hd_g}
have the  form as same as the lattice BSQ case (cf.\cite{NZSZ}),
which can be proved along the same line.
Equations \eqref{eq:A21hd_e} and \eqref{eq:A21hd_f} have been proved in Sec.\ref{sec-3-2}.
Therefore, here we just prove the Casoratian solutions of the bilinear equations \eqref{eq:A21hd_c} and \eqref{eq:A21hd_h},
while we will skip \eqref{eq:A21hd_d} since it is the symmetric version of \eqref{eq:A21hd_c}.
Some formulae presented in Appendix \ref{gd4A2:formulae} will be used in the proof.

We prove \eqref{eq:A21hd_c} in the following form
\begin{equation}
\label{eq:A21hd_c-ut}
\dt{f}'(\sigma_3^2\dc{f}'+\sigma_3\dc{g}'+\dc{s}')-\sigma_1\dt{f}'(\sigma_3\dc{f}'+\dc{g}')
+\sigma^-_{13}f'(\sigma_3\dtc{f}'+\dtc{g}')-\dc{f}'\dt{s}'=0,
\end{equation}
which is a down-tilde-shifted version of the original one. For $\dt{f}'$, $\dc{f}'$, $\dt{s}'$, $\sigma_3\dc{f}'+\dc{g}'$,
$\sigma^-_{13}(\sigma_3\dtc{f}'+\dtc{g}')$ and $\sigma_3^2\dc{f}'+\sigma_3\dc{g}'+\dc{s}'$,
we use \eqref{A21f-i} with $i=1,3$, \eqref{A21s-i} with $i=1$, \eqref{A21-sigi-Eifg} with $i=3$,
\eqref{A21-sigi3-Eifgdc} with $(i,j)=(1,3)$, \eqref{A21-Eifgs} with $i=3$, respectively. Then we have
\begin{align*}
\dt{H}_{23}^{\delta}=&|\wh{N-1}||\wh{N-4},N-2,\dc{\Ph}(0),\dt{\Ph}(0)|
-|\wh{N-2},\dc{\Ph}(0)||\wh{N-4},N-2,N-1,\dt{\Ph}(0)| \\
&+|\wh{N-2},\dt{\Ph}(0)||\wh{N-4},N-2,N-1,\dc{\Ph}(0)|
-\sigma_1(|\wh{N-1}||\wh{N-3},\dc{\Ph}(0),\dt{\Ph}(0)| \\
&-|\wh{N-2},\dc{\Ph}(0)||\wh{N-3},N-1,\dt{\Ph}(0)|
+|\wh{N-2},\dt{\Ph}(0)||\wh{N-3},N-1,\dc{\Ph}(0)|),
\end{align*}
which is zero in light of \eqref{plu-r-1} with
\begin{align*}
& \bM=(\wh{N-4},N-2),\quad \ba=\Ph(N-3), \quad \bb=\Ph(N-1), \quad \bc=\dc{\Ph}(0), \quad \bd=\dt{\Ph}(0), \\
& \bM=(\wh{N-3}),\quad \ba=\Ph(N-2),\quad \bb=\Ph(N-1), \quad \bc=\dc{\Ph}(0),\quad \bd=\dt{\Ph}(0).
\end{align*}

We next prove \eqref{eq:A21hd_h}. For this purpose, we first reformulate its down-tilde-hat-shifted version
\begin{align*}
\dth{\mathcal{H}}_{28}^{\delta}
=&\sigma^-_{12}\big[f'\big((\sigma_{1}\sigma^+_{12}
+\sigma_{2}\sigma^+_{23}+\sigma_{3}\sigma^+_{13})\dthc{f}'
+\sigma^+_{123}\dthc{g}'+\dthc{s}'\big)-(g'-\epsilon_1f')(\sigma^+_{123}\dthc{f}'+\dthc{g}')  \\
& +(h'-\epsilon_1g'+\epsilon_2f')\dthc{f}'\big]-(R(\sigma_1,\sigma_3)\dt{f}'\dhc{f}'
-R(\sigma_2,\sigma_3)\dh{f}'\dtc{f}').
\end{align*}
Multiplying both sides by $\sigma^-_{13}\sigma^-_{23}$ and identifying the formulas
\eqref{A21:f-dthc}-\eqref{A21:fgs-dthc} with $(i,j,\kappa)=(1,2,3)$ to calculate
$\sigma^{-}_{123}\dthc{f}'$, $\sigma^{-}_{123}(\sigma^+_{123}\dthc{f}'+\dthc{g}')$ and
$\sigma^{-}_{123}\big((\sigma_{1}\sigma^+_{12}+\sigma_{2}\sigma^+_{23}
+\sigma_{3}\sigma^+_{13})\dthc{f}'+\sigma^+_{123}\dthc{g}'+\dthc{s}'\big)$.
Furthermore,
\eqref{A21-r-fdthc} with $(i,j,\kappa)=(1,2,3)$ and the relation \eqref{sig-S} are used to
manipulate the component $\sigma^-_{13}\sigma^-_{23}(R(\sigma_1,\sigma_3)\dt{f}'\dhc{f}'-R(\sigma_2,\sigma_3)\dh{f}'\dtc{f}')$.
Thus we know
\begin{align*}
& \sigma^-_{13}\sigma^-_{23}\dth{\mathcal{H}}_{28}^{\delta}  \\
=~& (-1)^{N}\big[ f'|\wh{N-6},N-4,N-3, \dt{\Ph}(0),\dh{\Ph}(0),\dc{\Ph}(0)|-(g'-\epsilon_1f')|
\wh{N-5},N-3, \dt{\Ph}(0), \\
&\dh{\Ph}(0),\dc{\Ph}(0)| + (h'-\epsilon_1g'+\epsilon_2f')|\wh{N-4}, \dt{\Ph}(0), \dh{\Ph}(0), \dc{\Ph}(0)|\big] \\
&-(S(\sigma_1)\sigma^-_{23}\dt{f}'\dhc{f}'- S(\sigma_2)\sigma^-_{13}\dh{f}'\dtc{f}'
+ S(\sigma_3)\sigma^-_{12} \dc{f}'\dth{f}'),
\end{align*}
where the function $S$ is defined as in \eqref{2.18}.
In addition, for $S(\sigma_1)\dt{f}'$, $S(\sigma_2)\dh{f}'$ and $S(\sigma_3)\dc{f}'$, equation
\eqref{B2:fd} with $i=1,2,3$ are used.
To deal with $f'|\wh{N-6},N-4,N-3, \dt{\Ph}(0),\dh{\Ph}(0),\dc{\Ph}(0)|$, \eqref{A21-fN-6}
with $(i,j,\kappa)=(1,2,3)$ are applied,
and \eqref{A21-fN-5}-\eqref{A21-hN-4} are adopted in the same way. To proceed,
we use \eqref{A21-gamma4} with $i=1,2,3$ to calculate $\text{tr}(\bK)\dt{f}'$, $\text{tr}(\bK)\dh{f}'$
and $\text{tr}(\bK)\dc{f}'$,
respectively. With the aid of \eqref{A21-sigij-Eijf} with $(i,j)=(1,2),~(i,j)=(1,3)$ and $(i,j)=(2,3)$, respectively,
we can replace $\sigma^{-}_{12}\dth{f}'$, $\sigma^{-}_{13}\dtc{f}'$ and $\sigma^{-}_{23}\dhc{f}'$.
Therefore,
the following result can be derived finally
\begin{align*}
&(-1)^{N-1}\big[\text{tr}(\bK)(\sigma^-_{12}\dth{f}'\dc{f}'-\sigma^-_{13}\dtc{f}'\dh{f}'
+\sigma^-_{23}\dhc{f}'\dt{f}')
+\sigma^-_{13}\sigma^-_{23}\dth{\mathcal{H}}_{28}^{\delta}\big]  \\
=~&|\wh{N-3},\dh{\Ph}(0),\dt{\Ph}(0)||\wh{N-3},\bK\Ph(N-2), \dc{\Ph}(0)|
-|\wh{N-3},\dc{\Ph}(0),\dt{\Ph}(0)|  \\
&\cdot |\wh{N-3},\bK\Ph(N-2),\dh{\Ph}(0)|
+|\wh{N-3},\dc{\Ph}(0),\dh{\Ph}(0)||\wh{N-3},\bK\Ph(N-2), \dt{\Ph}(0)|,
\end{align*}
which vanishes in light of \eqref{plu-r-1}, in which
\[\bM=(\wh{N-3}),\quad \ba=\dh{\Ph}(0),\quad \bb=\dt{\Ph}(0), \quad \bc=\bK\Ph(N-2), \quad \bd=\dc{\Ph}(0).\]
Thus from \eqref{A21-r-fdthc} with $i=1,2,3$, we conclude that
$$\sigma^{-}_{13}\sigma^{-}_{23}\dth{\mathcal{H}}_{28}^{\delta}
=-\text{tr}(\bK)\big(\sigma^{-}_{12}\dth{f'}\dc{f'}-\sigma^{-}_{13}\dtc{f'}\dh{f'}
+\sigma^{-}_{23}\dhc{f'}\dt{f'}\big)=0,$$
which completes the proof.

\end{proof}

Finally, as for one-soliton and two-soliton solutions,
we briefly provide them through the following Remark.

\begin{remark}\label{R-1}
Since the column vector $\Ph$ used to define the Caosratians in \eqref{eq:ss-A21d}
is still formulated by the equation set \eqref{CES-B2},
it will  share the same expressions as listed in Sec.\ref{sec-4-1-3}.
Solutions can be given through \eqref{A21-tran'}.
For one-soliton solution of the system \eqref{eq:A21-o2}, $u'_{0}$ and $u'_{1,0}$ are the same as \eqref{B2-1ss},
and in addition,
\begin{align}
\label{A21-oss}
v'_a=\bigg(\sum_{j=1}^4\frac{\rho_{1,j}}{\sigma_3+\zeta_{1,j}}\bigg)/\tau_1,\quad s'_a=1, \quad r'_a=\sigma_3,
\end{align}
where $\tau_1$, $\rho_{i,j}$ and $\zeta_{i,j}$ are defined as in \eqref{3.30}.

For two-soliton solutions, $f'$, $g'$, $h'$ and $s'$ have been provided in \eqref{B2d-fghs-2SS},
and  we also have
\begin{subequations}
\begin{align}
\label{A2d-df-2SS}
& \dc{f}'
=\sum_{i=1}^2(-1)^{i+1}\bigg(\sum_{j=1}^4\frac{\rho_{i,j}}{\sigma_3+\zeta_{i,j}}\bigg)
\cdot \bigg(\sum_{j=1}^4\frac{\zeta_{i+1,j}\rho_{i+1,j}}{\sigma_3+\zeta_{i+1,j}}\bigg), \\
\label{A2d-sg3df-dg-2SS}
& \sigma_3\dc{f}'+\dc{g}'=\sum_{i=1}^2(-1)^{i+1}
\bigg(\sum_{j=1}^4\frac{\rho_{i,j}}{\sigma_3+\zeta_{i,j}}\bigg)
\cdot\bigg(\sum_{j=1}^4\zeta_{i+1,j}\rho_{i+1,j}\bigg).
\end{align}
\end{subequations}
One can substitute them into \eqref{A21-tran'} to get solutions.
It is also notable that in this case $\sigma^2_{3}\dc{f}'+\sigma_3\dc{g}'+\dc{s}'=f'$, which implies
$r'_a=1$.
\end{remark}

\subsection{Alternative GD-4 (A-2) equation}\label{sec-4-3}

Similar to the GD-4 (A-2) equation, for the alternative  GD-4 (A-2) equation,
let us first present its soliton solutions in Casoratians.
The proof will be included in a more general case  for Theorem \ref{theo-Sigm4-A22d}.

\begin{thm}
\label{ss-A22}
The functions
\begin{align}
\label{ss-A22-f}
f=|\wh{N-1}|,\quad g=|\wh{N-2},N|,\quad h=|\wh{N-2},N+1|, \quad s=|\wh{N-3},N-1,N|
\end{align}
generate   soliton solutions to the bilinear  alternative  GD-4 (A-2) system \eqref{eq:A2'-b0},
where the Casoratians are composed by the vector $\Ps$ in \eqref{Ps} with entries \eqref{eq:psij-sec-4}.
\end{thm}

To get $\delta$-extended equations, we   introduce $f', g',  h'$ and $s'$ as in \eqref{eq:moB2'}
and (cf. Eq.\eqref{A22-tran})
\begin{align}
\label{A22-tran'}
u'_{0}=\frac{g'}{f'}, \quad u'_{0,1}=\frac{s'}{f'}, \quad w'_b=\frac{\d{f}'}{f'},
\quad t'_b=\frac{\d{g}'-\sigma_4\d{f}'}{f'}, \quad z'_b=\frac{\d{h}'-\sigma_4\d{g}'+\sigma_4^2\d{f}'}{f'}.
\end{align}
Then, making use of  \eqref{eq:dand realtions} we obtain
the $\delta$-extended form of \eqref{eq:A2'-b0}
\begin{subequations}
\label{eq:A22hdf}
\begin{align}
\label{eq:A22hd_a}
\mathcal{H}_{31}^{\delta}\equiv &\wt{f}'(\sigma_1\d{f}'+\d{g}')-\sigma^{-}_{14}f'\wtd{f}'-(\sigma_4\wt{f}'+\wt{g}')\d{f}'=0,\\
\label{eq:A22hd_b}
\mathcal{H}_{32}^{\delta}\equiv &\wh{f}'(\sigma_2\d{f}'+\d{g}')-\sigma^{-}_{24}f'\whd{f}'-(\sigma_4\wh{f}'+\wh{g}')\d{f}'=0,\\
\label{eq:A22hd_c}
\mathcal{H}_{33}^{\delta}\equiv& \wt{f}'(\sigma_4^2\d{f}'-\sigma_4\d{g}'+\d{h}')+\sigma_1\wt{f}'(\d{g}'-\sigma_4\d{f}')
-\sigma^{-}_{14}f'(\wtd{g}'-\sigma_4\wtd{f}')-\d{f}'\wt{h}'=0,\\
\label{eq:A22hd_d}
\mathcal{H}_{34}^{\delta}\equiv&  \wh{f}'(\sigma_4^2\d{f}'-\sigma_4\d{g}'+\d{h}')+\sigma_2\wh{f}'(\d{g}'-\sigma_4\d{f}')
-\sigma^{-}_{24}f'(\whd{g}'-\sigma_4\whd{f}')-\d{f}'\wh{h}'=0,\\
\label{eq:A22hd_e}
\mathcal{H}_{35}^{\delta}\equiv&  \wt{f}'(\sigma_1 g'+h')-\wt{g}'(\sigma_1 f'+g')+f'\wt{s}'=0,\\
\label{eq:A22hd_f}
\mathcal{H}_{36}^{\delta}\equiv & \wh{f}'(\sigma_2 g'+h')-\wh{g}'(\sigma_2 f'+g')+f'\wh{s}'=0,\\
\label{eq:A22hd_g}
\mathcal{H}_{37}^{\delta}\equiv& \sigma^{-}_{12}(\wt{f}'\wh{f}'-f'\wth{f}')+\wt{f}'\wh{g}'-\wh{f}'\wt{g}'=0,\\
\label{eq:A22hd_h}
\mathcal{H}_{38}^{\delta}\equiv&\big[Q(\sigma_1,\sigma_2)f'\wthd{f}'
+s'\wthd{f}'+(\sigma^+_{12}+\epsilon_1)g'\wthd{f}'-((\sigma^+_{12}+\epsilon_1)f'+g')
(\wthd{g}'-\sigma_4\wthd{f}')\nn\\
&+f'(\wthd{h}'-\sigma_4\wthd{g}'+\sigma_4^2\wthd{f}')\big]\sigma^{-}_{12}
-R(\sigma_1,\sigma_4)\wt{f}'\whd{f}'+R(\sigma_2,\sigma_4)\wh{f}'\wtd{f}'=0,
\end{align}
\end{subequations}
and relations
\begin{subequations}
\label{A22o-A22d}
\begin{align}
& u_{0}=u_{0}'-N\delta, \quad u_{0,1}=u_{0,1}'-(N-1)\delta u_{0}'+\frac{N(N-1)}{2}\delta^2, \\
\label{wbwb'}
&  w_b=w_b', \quad t_b=t_b'-(N+1)\delta w_b', \quad z_b=z_b'-(N+2)\delta t_b'+\frac{(N+2)(N+1)}{2}\delta^2 w_b'.
\end{align}
\end{subequations}
The latter recovers the $\delta$-extended form of the alternative GD-4 (A-2) equation \eqref{eq:A22o}:
\begin{subequations}
\label{eq:A22-o2}
\begin{align}
\label{eq:GD4-A22-o2-ex_a}
A_{21}^{\delta}\equiv & t'_b-\sigma^{-}_{14}\wt{w}'_b+(\sigma_1-\wt{u}'_0)w'_b=0,
\\
\label{eq:GD4-A22-ex_b}
A_{22}^{\delta}\equiv &  t'_b-\sigma^{-}_{24}\wh{w}'_b+(\sigma_2-\wh{u}'_0)w'_b=0, \\
\label{eq:GD4-A22-ex_c}
A_{23}^{\delta}\equiv& (\wh{z}'_b+\sigma_1\wh{t}'_b)/\wh{w}'_b-\sigma^{-}_{14}\wth{t}'_b/\wh{w}'_b
-(\wt{z}'_b+\sigma_2\wt{t}'_b)/\wt{w}'_b+\sigma^{-}_{24}\wth{t}'_b/\wt{w}'_b=0, \\
\label{eq:GD4-A22-ex_d}
A_{24}^{\delta}\equiv& \wh{u}'_{0,1}-\wt{u}'_{0,1}-(\sigma^{-}_{12}+\wh{u}'_0-\wt{u}'_0)u'_0
+\sigma_1\wt{u}'_0-\sigma_2\wh{u}'_0=0, \\
\label{eq:GD4-A22-ex_e}
A_{25}^{\delta}\equiv &  2z'_b-\sigma^{-}_{14}\wt{t}'_b-\sigma^{-}_{24}\wh{t}'_b
+\sigma^{+}_{12}t'_b-w'_b((\sigma^{+}_{12}+\wt{u}'_0
+\wh{u}'_0)\wh{\wt{u}}'_0-2\wh{\wt{u}}'_{0,1}
-\sigma_2\wt{u}'_0-\sigma_1\wh{u}'_0)=0, \\
\label{eq:GD4-A22-ex_f}
A_{26}^{\delta}\equiv& \big[Q(\sigma_1,\sigma_2)+u'_{0,1}+(\sigma^{+}_{12}+\epsilon_1)u'_{0}
- \big((\sigma^{+}_{12}+\epsilon_1+u'_{0})\wth{t}'_b-\wth{z}'_b\big)/\wth{w}'_b\big]\nn\\
& \cdot(\sigma^{-}_{12}+\wh{u}'_0-\wt{u}'_0)-(R(\sigma_1,\sigma_4)\wh{w}'_b
-R(\sigma_2,\sigma_4)\wt{w}'_b)/\wth{w}'_b=0.
\end{align}
\end{subequations}
One may also obtain   Hietarinta's from
\begin{subequations}
\label{eq:A2-222'}
\begin{align}
\label{eq:A2-ex-N4-1 a'}
& y'=\wt{z}' x'-\wt{x}', \quad y'=\wh{z}'x'-\wh{x}',\\
\label{eq:A2-ex-N4-1 b'}
& \wt{\eta}'\wh{x}'-\wh{\eta}'\wt{x}'=\wth{y}'(\wt{x}'-\wh{x}'), \\
\label{eq:A2-ex-N4-1 bb'}
& \wh{\xi}'-\wt{\xi}'=(\wh{z}'-\wt{z}')z', \\
\label{eq:alt-A2-ex_eta'}
& \wth{\xi}'=\frac{-\eta(\wh{x}'-\wt{x}')+y'(\wh{y}'-\wt{y}')+ \wt{x}'\wh{y}'-\wh{x}'\wt{y}'}{x'(\wh{x}'-\wt{x}')},\\
\label{eq:A2-ex-N4-1 c'}
&\wth{\eta}'=-\wth{x}'\xi'+\wth{y}'z'+\epsilon_1(\wth{y}'-\wth{x}'z')-\epsilon_2\wth{x}'
- \frac{P(\sigma_1,\sigma_4)\wh{x}'-P(\sigma_2,\sigma_4)\wt{x}'}{\wh{z}'-\wt{z}'},
\end{align}
\end{subequations}
through the transformation
\begin{subequations}
\label{eq:xyz-A2-1'}
\begin{align}
& w'_b=x'/x'_b, \quad u'_0=z'-z'_0, \quad t'_b=(y'-w'_by'_b)/x'_b, \\
& u'_{0,1}=\xi'-z'_0u'_0-\xi'_0, \quad z'_b=(\eta'-z'_0y'+x'\xi'_0)/x'_b,
\end{align}
\end{subequations}
with
\begin{subequations}
\label{eq:xyz-0-A2-1'}
\begin{align}
& x'_b =(-\sigma_1+\sigma_4)^{n}(-\sigma_2+\sigma_4)^{m}c_0, \quad z'_0=-\sigma_1n-\sigma_2m-c'_1, \quad y'_b=x'_bz'_0,\\
& \xi'_0 = \left((z'_0)^2-(n\sigma_1^2+m\sigma_2^2+c'_2)\right)/2-c'_3,
\end{align}
\end{subequations}
and constants $\{c'_i\}$.

Casoratian solutions with more general setting of the bilinear system \eqref{eq:A22hdf}
can be described below.

\begin{thm}
\label{theo-Sigm4-A22d}
The functions
\begin{align}
\label{eq:ss-A22d}
f'=|\wh{N-1}|,\quad g'=|\wh{N-2},N|,\quad h'=|\wh{N-2},N+1|, \quad s'=|\wh{N-3},N-1,N|
\end{align}
solve the bilinear $\delta$-extended  alternative GD-4 (A-2) system \eqref{eq:A22hdf},
provided that the column vector $\Ph$
is defined by equation set \eqref{CES-B2}.
\end{thm}

\begin{proof}
Among the equations in \eqref{eq:A22hdf}, $\mathcal{H}_{35}^{\delta}$, $\mathcal{H}_{36}^{\delta}$
and $\mathcal{H}_{37}^{\delta}$ already appeared in Sec.\ref{sec-4-1-2}.
Here we just prove \eqref{eq:A22hd_a} and \eqref{eq:A22hd_c}.
We skip \eqref{eq:A22hd_h} as its  verification
is quite similar to the one for \eqref{eq:A21hd_h}.
Formulae given in Appendix \ref{A2:formulae} will be used in the proof.

For \eqref{eq:A22hd_a}, its  down-tilde shift   reads
\begin{align}
\label{H31-del-dt}
\mathcal{\dt{H}}_{31}^{\delta}\equiv f'(\sigma_1\dt{\d{f}}'
+\dt{\d{g}}')-\sigma^{-}_{14}\dt{f}'\d{f}'-(\sigma_4 f'+g')\dt{\d{f}}'.
\end{align}
For $f',~\d{\dt{f}}',~\sigma^{-}_{14}\dt{f}',~\sigma_4 f'+g',
~\sigma_1\dt{\d{f}}'+\dt{\d{g}}'$, we use \eqref{A22 f-f dot}, \eqref{A22 dEf} with $i=1$, \eqref{A22 p-b f-i}
with $i=1$, \eqref{A22 bf g-dot}, \eqref{A22 dotfg-i} with $i=1$. Then we recognize
\begin{align*}
&\mathcal{\dt{H}}_{31}^{\delta}=|\d{\wh{N-2}},\Ph(0)||\d{\wh{N-3}},\d{\Ph}(N-1),\d{\dt{\Ph}}(0)|
- |\d{\wh{N-1}}||\d{\wh{N-3}},\Ph(0),\d{\dt{\Ph}}(0)|  \\
&\qquad\quad +|\d{\wh{N-2}},\d{\dt{\Ph}}(0)||\d{\wh{N-3}},\Ph(0),\d{\Ph}(N-1)|,
\end{align*}
which is zero in light of identity \eqref{plu-r-1} with
\[\bM=(\d{\wh{N-3}}), \quad \ba=\d{\Ph}(0),\quad \bb=\Ph(0), \quad \bc=\d{{\Ph}}(N-1), \quad \bd=\d{\dt{{\Ph}}}(0).\]

For the bilinear equation \eqref{eq:A22hd_c}, we still consider its down-tilde version, i.e.,
\begin{align}
\dt{\mathcal{H}}_{33}^{\delta}\equiv
f'(\sigma_4^2\dt{\d{f}}'-\sigma_4\dt{\d{g}}'+\dt{\d{h}}')+\sigma_1f'(\dt{\d{g}}'-\sigma_4\dt{\d{f}}')
-\sigma^{-}_{14}\dt{f}'(\d{g}'-\sigma_4\d{f}')-\d{\dt{f}}'h'=0,
\end{align}
where $\d{f}'$ is defined as \eqref{H31-del-dt},
$f'$, $h'$  and $\d{g}'$ are provided by formulas \eqref{A22 f-f dot},
\eqref{A22-h}, and \eqref{A22-dg}, respectively. For $\d{\dt{f}}'$, $\dt{\d{g}}'$, $\dt{\d{h}}'$  and
$\sigma^{-}_{14}\dt{f}'$, we use \eqref{A22 dEf} with $i=1$, \eqref{A22-dEg}
with $i=1$, \eqref{A22-dEh} with $i=1$, and \eqref{A22 p-b f-i}
with $i=1$, respectively. Then we get
\begin{align*}
\dt{\mathcal{H}}_{33}^{\delta}=&|\d{\wh{N-2}},\Ph(0)||\d{\wh{N-3}},\d{\d{\Ph}}(N-1),\d{\dt{\Ph}}(0)|
- |\d{\wh{N-2}},\d{\d{\Ph}}(N-1)||\d{\wh{N-3}},\Ph(0),\d{\dt{\Ph}}(0)| \\
&+|\d{\wh{N-2}},\d{\dt{\Ph}}(0)||\d{\wh{N-3}},\Ph(0),\d{\d{\Ph}}(N-1)|
+2\sigma_4(|\d{\wh{N-2}},\Ph(0)||\d{\wh{N-3}},\d{\dt{\Ph}}(0),\d{\Ph}(N-1)|  \\
&-|\d{\wh{N-2}},\d{\dt{\Ph}}(0)||\d{\wh{N-3}},\Ph(0),\d{\Ph}(N-1)|
+|\d{\wh{N-1}}||\d{\wh{N-3}},\Ph(0),\d{\dt{\Ph}}(0)|),
\end{align*}
which turns out to be zero in light of \eqref{plu-r-1}, in which
\begin{align*}
& \bM=(\d{\wh{N-3}}),\quad \ba=\d{\Ph}(N-2),\quad \bb=\Ph(0), \quad \bc=\d{\d{\Ph}}(N-1), \quad \bd=\d{\dt{\Ph}}(0), \\
& \bM=(\d{\wh{N-3}}),\quad \ba=\d{\Ph}(N-2),\quad \bb=\Ph(0), \quad
\bc=\d{\dt{\Ph}}(0),\quad \bd=\d{\Ph}(N-1).
\end{align*}

Thus the proof is completed.

\end{proof}

For explicit examples of solutions, we have the following remark.

\begin{remark}\label{R-2}
The functions
\eqref{eq:ss-A22d} with $\Ph$ given by \eqref{psi-jd-sec-4}, \eqref{Ph-Jor} and \eqref{qrpsi}
provide soliton solutions, Jordan-block solutions and quasi-rational solutions for the $\delta$-extended
alternative GD-4 (A-2) equation \eqref{eq:A22-o2}, respectively.
Solutions are given through \eqref{A22-tran'}.
 For  one-soliton solution,
in addition to $u'_{0}$ and $u'_{1,0}$ as listed in \eqref{B2-1ss}, we also have
\begin{align}
\label{wb-1ss}
& w'_b=\tau^{[1]}_1/\tau_1+\sigma_4, \quad t'_b=\tau^{[2]}_1/\tau_1-\sigma_4^2, \quad z'_b=\tau^{[3]}_1/\tau_1+\sigma_4^3.
\end{align}
Here and below for notations one should refer to \eqref{3.30}.

For two-soliton solutions, based on $f'$, $g'$, $h'$ and $s'$ listed in \eqref{B2d-fghs-2SS}, we also have
\begin{subequations}
\begin{align}
\label{A22d-df-2SS}
& \d{f}'=\sum_{i=1}^2(-1)^{i+1}\bigg(\sum_{j=1}^4(\sigma_4+\zeta_{i,j})\rho_{i,j}\bigg)
\cdot \bigg(\sum_{j=1}^4(\sigma_4+\zeta_{i+1,j})\zeta_{i+1,j}\rho_{i+1,j}\bigg), \\
\label{A22d-dg-2SS}
& \d{g}'=\sum_{i=1}^2(-1)^{i+1}\bigg(\sum_{j=1}^4(\sigma_4+\zeta_{i,j})\rho_{i,j}\bigg)
\cdot\bigg(\sum_{j=1}^4(\sigma_4+\zeta_{i+1,j})\zeta^2_{i+1,j}\rho_{i+1,j}\bigg),\\
\label{A22d-dh-2SS}
& \d{h}'=\sum_{i=1}^2(-1)^{i+1}\bigg(\sum_{j=1}^4(\sigma_4+\zeta_{i,j})\rho_{i,j}\bigg)
\cdot\bigg(\sum_{j=1}^4(\sigma_4+\zeta_{i+1,j})\zeta^3_{i+1,j}\rho_{i+1,j}\bigg),
\end{align}
\end{subequations}
which provide two-soliton solutions through \eqref{A22-tran'}.

\end{remark}

\subsection{GD-4 (C-3) equation}\label{sec-4-4}

For the bilinear GD-4 (C-3) system \eqref{eq:C3-b0}, its
soliton solutions in terms of  Casoratians are presented below.

\begin{thm}
\label{ss-C3}
The functions
\begin{align}
\label{ss-C3-f}
f=|\wh{N-1}|,\quad g=|\wh{N-2},N|
\end{align}
generate  soliton solutions to the bilinear GD-4 (C-3) system \eqref{eq:C3-b0},
where the Casoratians are composed by the vector $\Ps$ in \eqref{Ps} with entries \eqref{eq:psij-sec-4}
\end{thm}

This theorem will be proved later together with the $\delta$-extended case.
For the $\delta$-extension of the GD-4 (C-3) equation,
note that there are only $f$ and $g$ which are involved in the bilinear form.
For $f'$ and $g'$ defined as in \eqref{eq:moB2'},
making use of the relations in \eqref{fafd},
from \eqref{eq:C3-b0} we get the $\delta$-extended bilinear form
\begin{subequations}
\label{eq:C3hdf}
\begin{align}
\label{eq:C3hd_a}
& \mathcal{H}_{41}^\delta\equiv \sigma^{-}_{13}\wt{f}'\d{\dc{f}}'-\sigma^{-}_{14}f'\wtd{\dc{f}}'
+ \sigma^{-}_{34}\d{f}'\wt{\dc{f}}'=0,\\
\label{eq:C3hd_b}
& \mathcal{H}_{42}^\delta\equiv \sigma^{-}_{23}\wh{f}'\d{\dc{f}}'-\sigma^{-}_{24}f'\whd{\dc{f}}'+\sigma^{-}_{34}\d{f}'\wh{\dc{f}}'=0, \\
\label{eq:C3hd_c}
& \mathcal{H}_{43}^{\delta}\equiv f'(\sigma_3\wt{\dc{f}}'+\wt{\dc{g}}')
+\sigma^{-}_{13}\wt{f}'\dc{f}'-(\sigma_1 f'+g')\wt{\dc{f}}'=0,
\\
\label{eq:C3hd_d}
& \mathcal{H}_{44}^{\delta}\equiv f'(\sigma_3\wh{\dc{f}}'+\wh{\dc{g}}')
+\sigma^{-}_{23}\wh{f}'\dc{f}'-(\sigma_2 f'+g')\wh{\dc{f}}'=0,\\
\label{eq:C3hd_e}
& \mathcal{H}_{45}^{\delta}\equiv \wt{f}'(\sigma_1\d{f}'+\d{g}')
-\sigma^{-}_{14}f'\wtd{f}'-(\sigma_4\wt{f}'+\wt{g}')\d{f}'=0, \\
\label{eq:C3hd_f}
& \mathcal{H}_{46}^{\delta}\equiv \wh{f}'(\sigma_2\d{f}'+\d{g}')
-\sigma^{-}_{24}f'\whd{f}'-(\sigma_4\wh{f}'+\wh{g}')\d{f}'=0, \\
\label{eq:C3hd_g}
& \mathcal{H}_{47}^{\delta}\equiv\sigma^{-}_{14}\wtd{f}'\wh{f}'-\sigma^{-}_{24}\whd{f}'\wt{f}'
- \sigma^{-}_{12}\wth{f}'\d{f}'=0, \\
\label{eq:C3hd_h}
& \mathcal{H}_{48}^\delta\equiv (\sigma_3\dc{f}'+\dc{g}')\wthd{f}'-\dc{f}'(\wthd{g}'-\sigma_4\wthd{f}')
- \frac{R(\sigma_1,\sigma_4)}{\sigma^{-}_{12}\sigma^{-}_{13}}\whd{f}'\wt{\dc{f}}'
+ \frac{R(\sigma_2,\sigma_4)}{\sigma^{-}_{12}\sigma^{-}_{23}}\wtd{f}'\wh{\dc{f}}' \\
&\qquad\quad +(\sigma^{+}_{12}+\epsilon_1)\dc{f}'\wthd{f}'
- \frac{R(\sigma_3,\sigma_4)}{\sigma^{-}_{13}\sigma^{-}_{23}}f'\wthd{\dc{f}}'=0,
\end{align}
\end{subequations}
and meanwhile, from
\begin{equation}
\label{C3-tran'}
v'_a=\frac{\dc{f'}}{f'},\quad w'_b=\frac{\d{f'}}{f'}, \quad s'_a=\frac{\sigma_3\dc{f'}+\dc{g'}}{f'}, \quad t'_b=\frac{\d{g}'-\sigma_4\d{f'}}{f'},
\quad  S'_{a,b}=\frac{1}{\sigma_4-\sigma_3}\frac{\d{\dc{f'}}}{f'},
\end{equation}
we get \eqref{eq:vava'}, \eqref{wbwb'} and  $S_{a,b}=S'_{a,b}$,
which lead to the $\delta$-extended form of the GD-4 (C-3) equation \eqref{eq:C3o2}:
\begin{subequations}
\label{eq:C3-o2}
\begin{align}
\label{eq:GD4-C3-o2-ex_a}
C_{31}^{\delta}\equiv& \sigma^{-}_{13}S'_{a,b}-\sigma^{-}_{14}\wt{S}'_{a,b}-\wt{v}'_a w'_b=0,\\
\label{eq:GD4-C3-o2-ex_b}
C_{32}^{\delta}\equiv& \sigma^{-}_{23}S'_{a,b}-\sigma^{-}_{24}\wh{S}'_{a,b}-\wh{v}'_a w'_b=0,\\
\label{eq:GD4-C3-o2-ex_c}
C_{33}^{\delta}\equiv& \wt{s}'_a/\wt{v}'_a-\wh{s}'_a/\wh{v}'_a
-\sigma^{-}_{12}+v'_a(\sigma^{-}_{13}/\wt{v}'_a-\sigma^{-}_{23}/\wh{v}'_a)=0, \\
\label{eq:GD4-C3-o2-ex_d}
C_{34}^{\delta}\equiv&  \wt{t}'_b/\wt{w}'_b-\wh{t}'_b/\wh{w}'_b
-\sigma^{-}_{12}+\wh{\wt{w}}'_b(\sigma^{-}_{14}/\wh{w}'_b-\sigma^{-}_{24}/\wt{w}'_b)=0, \\
\label{eq:GD4-C3-o2-ex_e}
C_{35}^{\delta}\equiv& v'_{a}\wh{\wt{t}}'_b-s'_a\wh{\wt{w}}'_b+w'_b\frac{\frac{R(\sigma_1,\sigma_4)}{\sigma^{-}_{13}}
 \wh{w}'_b\wt{v}'_a
- \frac{R(\sigma_2,\sigma_4)}{\sigma^{-}_{23}}\wt{w}'_b\wh{v}'_a}{\sigma^{-}_{14}
\wt{w}'_b-\sigma^{-}_{24}\wh{w}'_b}
-(\sigma^{+}_{12}+\epsilon_1)v'_a\wh{\wt{w}}'_b\nn\\
&-\frac{\sigma^{-}_{34}R(\sigma_3,\sigma_4)}{\sigma^{-}_{13}\sigma^{-}_{23}}\wh{\wt{S}}'_{a,b}=0.
\end{align}
\end{subequations}
The corresponding Hietarinta's form reads
\begin{subequations}
\label{eq:xyz-MSBSQ'}
\begin{align}
\label{eq:xyz-MSBSQ-a'}
& x'-\wt{x}'=\wt{y}'z', \quad x'-\wh{x}'=\wh{y}'z', \\
\label{eq:xyz-MSBSQ-b'}
& (\wt{y}'_1+y')\wh{y}'=(\wh{y}'_1+y')\wt{y}', \quad (\wh{z}'_1-\wth{z}')\wt{z}'
=(\wt{z}'_1-\wth{z}')\wh{z}', \\
\label{eq:xyz-MSBSQ-c'}
& y'\wth{z}'_1-y'_1\wth{z}'=
z'\frac{P(\sigma_1,\sigma_4)\wh{z}'\wt{y}'-P(\sigma_2,\sigma_4)\wt{z}'\wh{y}'}{\wh{z}'-\wt{z}'}
+\epsilon_1y'\wth{z}'
+P(\sigma_3,\sigma_4)\wth{x}',
\end{align}
\end{subequations}
which is connected with \eqref{eq:C3-o2} via the transformation
\begin{subequations}
\label{eq:MSBSQ-tran}
\begin{align}
& S'_{a,b}=\big((\sigma_1-\sigma_3)/(\sigma_1-\sigma_4)\big)^n \big((\sigma_2-\sigma_3)/(\sigma_2-\sigma_4)\big)^m x', \\
& v'_a=(\sigma_1-\sigma_3)^n(\sigma_2-\sigma_3)^m y', \quad w_b=(\sigma_1-\sigma_4)^{-n}(\sigma_2-\sigma_4)^{-m} z', \\
& s'_a=(\sigma_1-\sigma_3)^n(\sigma_2-\sigma_3)^m(y'_1-z'_0y'), \\
& t'_b=(\sigma_1-\sigma_4)^{-n}(\sigma_2-\sigma_4)^{-m}(z'_1-z'_0z'),
\end{align}
\end{subequations}
where $z'_0=-\sigma_1n-\sigma_2m-c_{1}$ with a constant $c_1$.
The CAC property of \eqref{eq:xyz-MSBSQ'} can be checked as well.

The more general  Casoratian solutions for the $\delta$-extended bilinear equation \eqref{eq:C3-o2}
are given in the following.

\begin{thm}
\label{theo-Sigm4-C1}
The functions
\begin{subequations}
\label{eq:ss-A22d-C3}
\begin{align}
& f'=|\wh{N-1}|,\quad g'=|\wh{N-2},N|
\end{align}
\end{subequations}
solve the bilinear equation \eqref{eq:C3-o2}, provided that the column vector $\Ph$ satisfies
the equation set \eqref{CES-B2}.
\end{thm}

\begin{proof}
We only prove \eqref{eq:C3hd_g} and \eqref{eq:C3hd_h}.
In fact, some equations in the bilinear form \eqref{eq:C3hdf}
have been verified (Eqs.\eqref{eq:C3hd_c}-\eqref{eq:C3hd_f} are $\mathcal{H}_{21}^{\delta}$,
$\mathcal{H}_{22}^{\delta}$, $\mathcal{H}_{31}^{\delta}$ and $\mathcal{H}_{32}^{\delta}$, respectively),
and Eqs.\eqref{eq:C3hd_a} and \eqref{eq:C3hd_b}   have appeared in \cite{NZSZ}
and can be proved along the same line.
Formulae given in Appendix \ref{gd4A2:formulae} and Appendix \ref{C3:formulae}
will be used in the proof.

For \eqref{eq:C3hd_g}, after a down-tilde-hat-bullet shift we obtain
\begin{align}
\sigma^{-}_{14}\dtd{f}'\dh{f}'-\sigma^{-}_{24}\dhd{f}'\dt{f}'-\sigma^{-}_{12}\dth{f}'\dd{f}'=0.
\end{align}
By applying formulas \eqref{A21f-i} with $i=1,~2,~4$ to get $\dt{f}'$, $\dh{f}'$, $\dd{f}'$ and \eqref{A21-sigij-Eijf}
with $(i,j)=(1,2),~(i,j)=(1,4),~(i,j)=(2,4)$ to get $\sigma^{-}_{12}\dth{f}'$, $\sigma^{-}_{14}\dtd{f}'$, $\sigma^{-}_{24}\dhd{f}'$, then
we derive
\begin{align*}
&\dthd{\mathcal{H}}_{47}^{\delta}=(-1)^{N}(|\wh{N-2},\dh{\Ph}(0)||\wh{N-3},\dt{\Ph}(0),\dd{\Ph}(0)|
- |\wh{N-2},\dt{\Ph}(0)||\wh{N-3},\dh{\Ph}(0),\dd{\Ph}(0)| \\
&\quad\quad\quad+|\wh{N-2},\dd{\Ph}(0)||\wh{N-3},\dh{\Ph}(0),\dt{\Ph}(0)|).
\end{align*}
The application of \eqref{plu-r-1} leads  this to be zero, where
\begin{align*}
\bM=(\wh{N-3}), \quad \ba=\Ph(N-2), \quad \bb=\dh{\Ph}(0), \quad \bc=\dt{\Ph}(0), \quad \bd=\dd{\Ph}(0).
\end{align*}

For \eqref{eq:C3hd_h},
the proof is a bit of complicated. We first use \eqref{sig-S} with $(i,j)=(1,4)$, $(i,j)=(2,4)$ and $(i,j)=(3,4)$
to handle the term in \eqref{eq:C3hd_h}: $$\sigma^{-}_{1234}\bigg(-\frac{R(\sigma_1,\sigma_4)}{\sigma^{-}_{12}\sigma^{-}_{13}}\d{\dt{f}}'\dhc{f}'
+ \frac{R(\sigma_2,\sigma_4)}{\sigma^{-}_{12}\sigma^{-}_{23}}\d{\dh{f}}'\dtc{f}'
- \frac{R(\sigma_3,\sigma_4)}{\sigma^{-}_{13}\sigma^{-}_{23}}\dth{f}'\d{\dc{f}}'\bigg).$$
Then
$\sigma^{-}_{1234}\dth{\mathcal{H}}_{48}^{\delta}$ can be rewritten as follows
\begin{align*}
\sigma^{-}_{1234}\dth{\mathcal{H}}_{48}^{\delta}
=&\sigma^{-}_{1234}\bigg[\d{f}'(\dthc{g}'
+\sigma^{+}_{1234}\dthc{f}')+(\epsilon_1\d{f}'-\d{g}')\dthc{f}'\bigg]  \\
&-(S(\sigma_1)-S(\sigma_4))\sigma^{-}_{234}\dhc{f}'\d{\dt{f}}'
+(S(\sigma_2)-S(\sigma_4))\sigma^{-}_{134}\dtc{f}'\d{\dh{f}}'
-(S(\sigma_3)-S(\sigma_4))\sigma^{-}_{124}\dth{f}'\d{\dc{f}}'.
\end{align*}
Subsequently, for $\sigma^{-}_{1234}\dthc{f}'$, $\sigma^{-}_{1234}(\dthc{g}'+\sigma^{+}_{1234}\dthc{f}')$
and
$S(\sigma_4)(\sigma^{-}_{234}\dhc{f}'\d{\dt{f}}'-\sigma^{-}_{134}\dtc{f}'\d{\dh{f}}'
+\sigma^{-}_{124}\dth{f}'\d{\dc{f}}')$,
one can refer to \eqref{C3-sig1234-dthf}, \eqref{C3-sig1234-dthgf} and \eqref{C3-fbf0}. Then we have
\begin{align*}
&\sigma^{-}_{1234}\dth{\mathcal{H}}_{48}^{\delta} \\
=~&\d{f}'|\d{\wh{N-6}},\d{\Ph}(N-4),\d{\dt{\Ph}}(0),\d{\dh{\Ph}}(0),\d{\dc{\Ph}}(0),\Ph(0)|
+(\epsilon_1\d{f}'-\d{g}')|\d{\wh{N-5}},\d{\dt{\Ph}}(0),\d{\dh{\Ph}}(0),\d{\dc{\Ph}}(0),\Ph(0)| \\
&-S(\sigma_1)\sigma^{-}_{234}\d{\dt{f}}'\dhc{f}'+S(\sigma_2)\sigma^{-}_{134}\d{\dh{f}}'\dtc{f}'
- S(\sigma_3)\sigma^{-}_{124}\d{\dc{f}}'\dth{f}'+S(\sigma_4)\sigma^{-}_{123}\d{\dthc{f}}'f'.
\end{align*}
Moreover, we use \eqref{C3-fN-6} to deal with $\d{f'}|\d{\wh{N-6}},\d{\Ph}(N-4),\d{\dt{\Ph}}(0),\d{\dh{\Ph}}(0),\d{\dc{\Ph}}(0),\Ph(0)|$,
and employ \eqref{C3-fN5} and \eqref{C3-gN-5} to calculate $(\epsilon_1\d{f}'-\d{g}')|\d{\wh{N-5}},\d{\dt{\Ph}}(0),\d{\dh{\Ph}}(0),\d{\dc{\Ph}}(0),\Ph(0)|$.
The further operations (\eqref{C3-sigi-j-f} with $i=1,2,3,4$ are used to obtain $S(\sigma_1)\d{\dt{f}}'$,
$S(\sigma_2)\d{\dh{f}}'$,  $S(\sigma_3)\d{\dc{f}}'$,  $S(\sigma_4)f'$; \eqref{G4-C3} with $i=1,2,3,4$
are utilized to derive $\text{tr}(\bK)\d{\dt{f}}'$, $\text{tr}(\bK)\d{\dh{f}}'$, $\text{tr}(\bK)\d{\dc{f}}'$, $\text{tr}(\bK)f'$;
\eqref{C3-sigij4-Eijf} with $(i,j)=(1,2),~(i,j)=(1,3),~(i,j)=(2,3)$ are utilized to
get $\sigma^{-}_{124}\dth{f}'$, $\sigma^{-}_{134}\dtc{f}'$, $\sigma^{-}_{234}\dhc{f}'$,
and \eqref{C3-sig123-dfdthc} is applied to get $\sigma^{-}_{123}\d{\dthc{f}}'$)
yield
\begin{align}
& \text{tr}(\bK)(\sigma^{-}_{123}\d{\dthc{f}}'f'-\sigma^{-}_{124}\dth{f}'\d{\dc{f}}'
+ \sigma^{-}_{134}\dtc{f}'\d{\dh{f}}'-\sigma^{-}_{234}\dhc{f}'\d{\dt{f}}')
- \sigma^{-}_{1234}\dth{\mathcal{H}}_{48}^{\delta}\nn\\
=~&|\d{\wh{N-4}},\d{\dt{\Ph}}(0),\d{\dh{\Ph}}(0),\d{\dc{\Ph}}(0)|\nn \\
& ~~~ \cdot \left (|\d{\wh{N-4}},\d{\Ph}(N-2),\bK\d{\Ph}(N-3),\Ph(0)|
     -|\d{\wh{N-3}},\bK\d{\Ph}(N-2),\Ph(0)| \right)\nn\\
&-|\d{\wh{N-4}},\Ph(0),\d{\dh{\Ph}}(0),\d{\dc{\Ph}}(0)|\nn \\
& ~~~ \cdot \left(|\d{\wh{N-4}},\d{\Ph}(N-2),\bK\d{\Ph}(N-3),\d{\dt{\Ph}}(0)|
      -|\d{\wh{N-3}},\bK\d{\Ph}(N-2),\d{\dt{\Ph}}(0)|\right) \nn\\
& +|\d{\wh{N-4}},\Ph(0),\d{\dt{\Ph}}(0),\d{\dc{\Ph}}(0)|\nn \\
& ~~~ \cdot \left(|\d{\wh{N-4}},\d{\Ph}(N-2),\bK\d{\Ph}(N-3),\d{\dh{\Ph}}(0)|
 -|\d{\wh{N-3}},\bK\d{\Ph}(N-2),\d{\dh{\Ph}}(0)|\right)\nn \\
& -|\d{\wh{N-4}},\Ph(0),\d{\dt{\Ph}}(0),\d{\dh{\Ph}}(0)|\nn \\
& ~~~ \cdot \left(|\d{\wh{N-4}},\d{\Ph}(N-2),\bK\d{\Ph}(N-3),
 \d{\dc{\Ph}}(0)|-|\d{\wh{N-3}},\bK\d{\Ph}(N-2),\d{\dc{\Ph}}(0)|\right) \nn \\
=~&0,
\end{align}
where we have utilized \eqref{plu-r-2} with $k=4$,
$\bP_1=(\d{\wh{N-4}},\d{\Ph}(N-2),\Ga\d{\Ph}(N-3))$, $\ba_1=\Ph(0)$, $\ba_2=\d{\dt{\Ph}}(0)$,
$\ba_3=\d{\dh{\Ph}}(0)$, $\ba_4=\d{\dc{\Ph}}(0)$, $\bQ_1=(\d{\wh{N-4}})$ and $\bP_2=(\d{\wh{N-3}},\Ga\d{\Ph}(N-2))$,
$\bb_1=\Ph(0)$, $\bb_2=\d{\dt{\Ph}}(0)$, $\bb_3=\d{\dh{\Ph}}(0)$, $\bb_4=\d{\dc{\Ph}}(0)$, $\bQ_2=(\d{\wh{N-4}})$.
In addition in light of \eqref{C3-fbf0},
we know that $\mathcal{H}_{48}^{\delta}=0$.
Thus we have completed the proof for all bilinear equations.

\end{proof}

Note that the only distinction between the systems of GD-4 (C-3) \eqref{eq:xyz-MSBSQ}
and its alternative form \eqref{eq:xyz-MSBSQ-1} lies in the final equation
(also see Eqs.(3.29d) and (3.30) in \cite{Tela}). Here we skip the discussion
for the alternative form since both of
the GD-4 (C-3) type equations share the same dependent variables.

Finally, we mention the following about one- and two-soliton solutions.

\begin{remark}\label{R-3}
One-soliton and two-soliton solution formulae for variables $v'_{a}$, $w'_b$,
$s'_a$ and $t'_b$ have been indicated in Remarks \ref{R-1}  and \ref{R-2}.
Therefore, here we just list one-soliton and two-soliton formula
of the variable $S'_{a,b}$ which is defined by $f$. Notations can be referred to \eqref{3.30}.
The one-soliton solution is
\begin{align}
\label{C3-oss}
S'_{a,b}=\bigg(\sum_{j=1}^4\frac{\sigma_4+\zeta_{1,j}}{\sigma_3+\zeta_{1,j}}\rho_{1,j}\bigg)/\tau_1,
\end{align}
and the two-soliton solution is determined by (cf.\eqref{C3-tran'})
\begin{align}
\label{C3-ddcf-2ss}
&\d{\dc{f'}}=\sum_{i=1}^2(-1)^{i+1}
\bigg(\sum_{j=1}^4\frac{\sigma_4+\zeta_{i,j}}{\sigma_3+\zeta_{i,j}}\rho_{i,j}\bigg)
\cdot\bigg(\sum_{j=1}^4\frac{\sigma_4+\zeta_{i+1,j}}{\sigma_3+\zeta_{i+1,j}}\zeta_{i+1,j}\rho_{i+1,j}\bigg).
\end{align}

\end{remark}

\section{Conclusions} \label{sec-6}

In this paper we investigated bilinearization and derivation of Casoratian solutions
for all the lattice GD-4 type equations that arose from the direct linearization scheme \cite{Tela},
including GD-4 type (A-2), (B-2) and (C-3) equations.
These equations can be presented in different forms,
for example, both \eqref{eq:B222} and \eqref{eq:B2o} are for the GD-4 (B-2) equation,
up to a set of point transformations \eqref{B2-tr}.
Taking the GD-4 (B-2) equation as an example,
we started from equation \eqref{eq:B2o},
introduced the transformation \eqref{eq:tr-GD-4} and got its bilinear form \eqref{eq:B2-h1}.
Soliton solutions were given in terms of Casoratian in Theorem \ref{ss-B2}.
By replacing the plane wave factor $\psi_s$ with $\phi_s$
which contains a parameter $\delta$,
we introduced new Casoratians $f', g', h', \cdots$ as in \eqref{eq:moB2'},
which are composed by $\phi_s$.
It is shown that the new and old Casoratians are connected as shown in
the relations presented in \eqref{eq:dand realtions}.
This then  enabled us to get the new bilinear forms \eqref{eq:B2-h2}
and new nonlinear forms \eqref{eq:B2-2d}.
The later is the $\delta$-extended GD-4 (B-2) equation (see also \eqref{eq:B222'}).
Then we proved in Theorem \ref{prop Sigm4-B2} that  the $\delta$-extended bilinear equation
set \eqref{eq:B2-h2}  allows more general Casoratian solutions,
of which the composing column vector $\Ph$ is determined by the equation set \eqref{CES-B2}.
The obtained solutions include not only solitons but also Jordan block (or multiple-pole solutions)
and quasi-rational solutions, as shown in Sec.\ref{sec-4-1-3}.
Then, along this line for the GD-4 (B-2) equation,
we also investigated the GD-4 (A-2) equation and its alternative from and
the GD-4 (C-3) equation.
We have derived their bilinear forms and $\delta$-extensions and Casoratian solutions.
It is also notable that the $\delta$-extended equations are all consistent around cube.

As we have mentioned before,
compared with the continuous and differential-difference case,
the discrete GD type equations play  important roles
in the research of (1+1)-dimensional discrete integrable systems,
because not only their bilinear forms but also their nonlinear forms
are difficult to obtain as reductions from those of the lattice KP equation.
We hope the results we obtained in this paper
can bring more insights on the lattice GD-4 type equations
as well as (1+1)-dimensional discrete integrable systems.

\vskip 20pt
\subsection*{Data availability statement}
Data sharing not applicable to this article as no datasets were generated or analysed during the current study.

\subsection*{Conflict of interest statement }

The authors declare that they have no conflict of interest.

\subsection*{Acknowledgments}
This project is supported by the National Natural Science Foundation of China (Grant Nos. 12071432 and 12271334)
and the Natural Science Foundation of Zhejiang Province (Grant No. LZ24A010007).

\vskip 20pt

\appendix

\section{Casoratian formulas for the GD-4 (B-2) equation}
\label{gd4:formulae}

In this section and also the succeeding three sections,
we collect some shift formulae
for the Casoratians defined in \eqref{eq:moB2}
where the column vector $\Ph$ is subject to the equation set \eqref{CES-B2}.
For convenience, we introduce some shorthand notations
\begin{align}\label{E}
E_1\Ph\doteq \dt{\Ph},\quad E_2\Ph\doteq \dh{\Ph},\quad E_3\Ph\doteq \dc{\Ph},\quad E_4\Ph\doteq \dd{\Ph},
\end{align}
and denote $\Ph(j)\doteq \Ph(n,m,\alpha,\beta, l=j)$.
In addition, for the notations like $\sigma_i$, etc, one can refer to \eqref{2.18}.
Some basic shift formulae are (for $i=1,2,3,4$)
\begin{subequations}
\label{Formula-BI}
\begin{align}
\label{B2:fd}
&\sigma_i^{j} E_i f'=(-1)^{N+1-j}|\wh{N-2},E_i\Ph(j)|, ~~~0\leq j\leq N-1,\\
\label{fd-N}
&\sigma_i^{N} E_i f'- f'=-|\wh{N-2},E_i\Ph(N)|, \\
\label{eq:fg-mu}
&-\sigma_i^{N-2} E_i \big(g' + \sigma_i f' \big) = |\wh{N-3}, N-1, E_i \Ph(N-2)|, \\
\label{eq:hv-mu}
&-\sigma_i^{N-2} E_i \big(\mu'+\sigma_i h' \big) = |\wh{N-3}, N+1, E_i \Ph(N-2)|, \\
\label{eq:gs-mu}
&\sigma_i^{N-3}E_i s' = |\wh{N-4}, N-2, N-1, E_i \Ph(N-3)|+|\wh{N-3}, N-1, E_i \Ph(N-2)|, \\
\label{eq:ssita-mu}
& \sigma_i^{N-3} E_i (\theta'+\sigma_i s') = |\wh{N-4}, N-2, N, E_i\Ph(N-3)|+|\wh{N-3}, N, E_i \Ph(N-2)|, \\
\label{fd-N+1}
&\sigma_i^{N+1} E_i f'-\sigma_i f'+g'=|\wh{N-2},E_i\Ph(N+1)|, \\
\label{fd-N+2}
&\sigma_i^{N+2} E_i f'-\sigma_i^2 f'+\sigma_i g'-h'=-|\wh{N-2},E_i\Ph(N+2)|, \\
\label{eq:f-th}
& \sigma_1^{N-2}\sigma_2^{N-2}\sigma^-_{12}\dth{f}'= |\wh{N-3}, \dh{\Ph}(N-2),\dt{\Ph}(N-2)|, \\
\label{eq:fg-th}
&\sigma_1^{N-2}\sigma_2^{N-2}\sigma^-_{12}(\dth{g}'+\sigma^+_{12}\dth{f}')=
\sigma_2^{N-2}\dh{f}'-\sigma_1^{N-2}\dt{f}'+|\wh{N-4}, N-2, \dh{\Ph}(N-2),\dt{\Ph}(N-2)|,
\\
\label{eq:fgs-th}
&\sigma_1^{N-2}\sigma_2^{N-2}\sigma^-_{12}(\dth{s}'+\sigma^+_{12}\dth{g}'
+(\sigma^{+^2}_{12}-\sigma_1\sigma_2)\dth{f}') \nn \\
&= \sigma_1\sigma_2^{N-2}\dh{f}'-\sigma_1^{N-2}\sigma_2\dt{f}'+|\wh{N-5}, N-3, N-2,\dh{\Ph}(N-2),\dt{\Ph}(N-2)|, \\
\label{eq:fgws-th}
&\sigma_1^{N-2}\sigma_2^{N-2}\sigma^-_{12}\big(\dth{\nu}'+\sigma^+_{12}\dth{s}'
+(\sigma^{+^2}_{12}-\sigma_1\sigma_2)\dth{g}'
+\sigma^+_{12}(\sigma_1^2+\sigma_2^2)\dth{f}'\big) \nn \\
&=\sigma_1^2 \sigma_2^{N-2} \dh{f}'-\sigma_1^{N-2}\sigma_2^2 \dt{f}'+|\wh{N-6}, N-4, N-3,N-2,\dh{\Ph}(N-2),\dt{\Ph}(N-2)|.
\end{align}
\end{subequations}
These formulae can be derived along the line of \cite{HZ-Bi}.
In addition, we also have the following formulas
\begin{subequations}
\label{Formula-B2}
\begin{align}
& f'|\wh{N-6},N-4,N-3,N-2,\dh{\Ph}(N-2),\dt{\Ph}(N-2)|\nn \\
& ~~ =\sigma_1^{N-2}\dt{f'}|\wh{N-6},N-4,N-3,N-2,N-1,\dh{\Ph}(N-2)| \nn \\
& \qquad -\sigma_2^{N-2}\dh{f'}|\wh{N-6},N-4,N-3,N-2,N-1,\dt{\Ph}(N-2)|,
\label{eq:B2fN-6}
\\
& f'|\wh{N-5},N-3,N-2,\dh{\Ph}(N-2),\dt{\Ph}(N-2)|\nn \\
& ~~ =\sigma_1^{N-2}\dt{f'}|\wh{N-5},N-3,N-2,N-1,\dh{\Ph}(N-2)| \nn\\
& \qquad -\sigma_2^{N-2}\dh{f'}|\wh{N-5},N-3,N-2,N-1,\dt{\Ph}(N-2)|,
\label{eq:B2fN-5}\\
& f'|\wh{N-4},N-2,\dh{\Ph}(N-2),\dt{\Ph}(N-2)|\nn \\
& ~~=\sigma_1^{N-2}\dt{f'}|\wh{N-4},N-2,N-1,\dh{\Ph}(N-2)|
-\sigma_2^{N-2}\dh{f'}|\wh{N-4},N-2,N-1,\dt{\Ph}(N-2)|,
\label{eq:fN-4}\\
&f'|\wh{N-3},\dh{\Ph}(N-2),\dt{\Ph}(N-2)|\nn \\
& ~~ =\sigma_1^{N-2}\dt{f'}|\wh{N-3},N-1,\dh{\Ph}(N-2)|
 -\sigma_2^{N-2}\dh{f'}|\wh{N-3},N-1,\dt{\Ph}(N-2)|,
\label{eq:fN-3}\\
& g'|\wh{N-5},N-3,N-2,\dh{\Ph}(N-2),\dt{\Ph}(N-2)|\nn \\
&~~ =\sigma_1^{N-2}\dt{f'}|\wh{N-5},N-3,N-2,N,\dh{\Ph}(N-2)| \nn\\
& \qquad-\sigma_2^{N-2}\dh{f'}|\wh{N-5},N-3,N-2,N,\dt{\Ph}(N-2)|,
\label{eq:gN-5} \\
& g'|\wh{N-4},N-2,\dh{\Ph}(N-2),\dt{\Ph}(N-2)| \nn \\
& ~~ =\sigma_1^{N-2}\dt{f'}|\wh{N-4},N-2,N,\dh{\Ph}(N-2)|
 -\sigma_2^{N-2}\dh{f'}|\wh{N-4},N-2,N,\dt{\Ph}(N-2)|,
 \label{eq:gN-4}\\
& g'|\wh{N-3},\dh{\Ph}(N-2),\dt{\Ph}(N-2)| \nn \\
& ~~ =\sigma_1^{N-2}\dt{f'}|\wh{N-3},N,\dh{\Ph}(N-2)|
 -\sigma_2^{N-2}\dh{f'}|\wh{N-3},N,\dt{\Ph}(N-2)|,
\label{eq:gN-3}\\
& h'|\wh{N-4},N-2,\dh{\Ph}(N-2),\dt{\Ph}(N-2)| \nn \\
& ~~ =\sigma_1^{N-2}\dt{f'}|\wh{N-4},N-2,N+1,\dh{\Ph}(N-2)|
-\sigma_2^{N-2}\dh{f'}|\wh{N-4},N-2,N+1,\dt{\Ph}(N-2)|,
\label{eq:hN-4}\\
\label{eq:hN-3}
& h'|\wh{N-3},\dh{\Ph}(N-2),\dt{\Ph}(N-2)|\nn \\
& ~~=\sigma_1^{N-2}\dt{f'}|\wh{N-3},N+1,\dh{\Ph}(N-2)|
 -\sigma_2^{N-2}\dh{f'}|\wh{N-3},N+1,\dt{\Ph}(N-2)|, \\
\label{eq:vN-3}
& \mu'|\wh{N-3},\dh{\Ph}(N-2),\dt{\Ph}(N-2)|\nn \\
& ~~ =\sigma_1^{N-2}\dt{f'}|\wh{N-3},N+2,\dh{\Ph}(N-2)|
 -\sigma_2^{N-2}\dh{f'}|\wh{N-3},N+2,\dt{\Ph}(N-2)|,  \\
\label{Ga4 B2}
&-\text{tr}(\bK)\sigma_i^{N-2}E_i f'=-|\wh{N-6},N-4,N-3,N-2,N-1,E_i \Ph(N-2)|\nn\\
& \qquad +|\wh{N-5},N-3,N-2,N,E_i \Ph(N-2)|+|\wh{N-4},N-2,N+1,E_i \Ph(N-2)|\nn\\
& \qquad +|\wh{N-3},N+2,E_i \Ph(N-2)|-|\wh{N-2},E_i \Ph(N+2)|\nn \\
& \qquad  -\epsilon_1\Bigl (|\wh{N-5},N-3,N-2,N-1,E_i \Ph(N-2)|-|\wh{N-4},N-2,N,E_i \Ph(N-2)| \nn\\
& \qquad  +|\wh{N-3},N+1,E_i \Ph(N-2)|+|\wh{N-2},E_i \Ph(N+1)|\Bigr)\nn\\
& \qquad -\epsilon_2\Bigl (|\wh{N-4},N-2,N-1,E_i\Ph(N-2)|
-|\wh{N-3},N,E_i \Ph(N-2)|-|\wh{N-2},E_i \Ph(N)|\Bigr)\nn\\
& \qquad -\epsilon_3\Bigl(|\wh{N-3},N-1,E_i \Ph(N-2)|+|\wh{N-2},E_i \Ph(N-1)|\Bigr).
\end{align}
\end{subequations}
Here $\text{tr}(\bK)$ stands for the trace of matrix $\bK$.
Equations \eqref{eq:B2fN-6}-\eqref{eq:vN-3} are due to the identity \eqref{plu-r-1}. 
For example,
equation \eqref{eq:B2fN-6} is the identity \eqref{plu-r-1} with
$$\bM=(\wh{N-6},N-4,N-3,N-2), ~\ba=\Ph(N-5),~ \bb=\Ph(N-1),~\bc=\dh{\Ph}(N-2),~\bd=\dt{\Ph}(N-2).$$
To prove the formula \eqref{Ga4 B2}, we need to make use of Proposition \ref{thm-2-3-1}.
For convenience, let us introduce the following notation to denote the shift in $l$-direction 
(cf.\eqref{E} and \eqref{shifts}):
\begin{align}\label{E}
E_5\Ph\doteq \underline{\Ph} \, ,
\end{align}
under which we can rewrite the relation \eqref{bKPh} as
\begin{equation}
\bK\Ph=\left(E_5^{-4}-\epsilon_1 E_5^{-3}+ \epsilon_2 E_5^{-2} -\epsilon_3 E_5^{-1}\right ) \Ph ~. 
\end{equation}
Recall Proposition \ref{thm-2-3-1}, where we take 
\begin{equation}\label{Omega-ij}
\Omega_{i,j}\equiv E_5^{-4}-\epsilon_1 E_5^{-3}+ \epsilon_2 E_5^{-2} -\epsilon_3 E_5^{-1} 
\end{equation}
and $\Xi=\big(\wh{N-2},E_i\Ph(N-2)\big)$.
Then, the r.h.s. of equation \eqref{id-w-2} yields
$\text{tr}(\bK)|\wh{N-2},E_i\Ph(N-2)|$,
i.e., $-\text{tr}(\bK)\sigma_i^{N-2}E_i f'$ in light of \eqref{B2:fd} with $j=N-1$,
while from the l.h.s. of equation \eqref{id-w-2} we get nothing but the r.h.s. of  \eqref{Ga4 B2}.
Thus we obtain the formula \eqref{Ga4 B2}.

\section{Casoratian formulas for the GD-4 (A-2) equation}
\label{gd4A2:formulae}

The basic shift formulas are
\begin{subequations}
\label{Formula-a21b}
\begin{align}
\label{A21f-i}
& E_i f'=(-1)^{N-1}|\wh{N-2},E_i\Ph(0)|,\\
\label{A21s-i}
& E_i s'=(-1)^{N-1}(|\wh{N-4},N-2,N-1,E_i\Ph(0)|-\sigma_i|\wh{N-3},N-1,E_i\Ph(0)|), \\
\label{A21-sigij-Eijf}
&\sigma^{-}_{ij}E_iE_jf'=|\wh{N-3},E_j\Ph(0),E_i\Ph(0)|,\\
\label{A21-sigi-Eifg}
& E_i(\sigma_i f'+ g')=(-1)^{N-1}|\wh{N-3},N-1,E_i\Ph(0)|,
\\
\label{A21-sigi3-Eifgdc}
&\sigma^-_{ij} E_i E_j(\sigma_j f'+g')=|\wh{N-4},N-2,E_j\Ph(0),E_i{\Ph}(0)|-\sigma_i|\wh{N-3},E_j\Ph(0),E_i{\Ph}(0)|,\\
%\displaybreak\\
\label{A21-Eifgs}
& E_i(\sigma_i^2 f'+\sigma_i g'+s')=(-1)^{N-1}|\wh{N-4},N-2,N-1,E_i \Ph(0)|, \\
\label{A21:f-dthc}
&\sigma^-_{ij\kappa}E_iE_jE_\kappa{f}'=(-1)^{N}|\wh{N-4},E_i\Ph(0),E_j\Ph(0),E_\kappa \Ph(0)|,\\
\label{A21:fg-dthc}
&\sigma^-_{ij\kappa} E_iE_jE_\kappa(g'+\sigma^+_{ij\kappa}f')=(-1)^{N}|\wh{N-5},N-3,E_i\Ph(0),E_j\Ph(0),E_\kappa \Ph(0)|,\\
\label{A21:fgs-dthc}
&\sigma^-_{ij\kappa}E_iE_jE_\kappa\big[(\sigma_i\sigma^+_{ij}+\sigma_j\sigma^+_{j\kappa}
+\sigma_\kappa\sigma^+_{i\kappa})f'+\sigma^+_{ij\kappa}g'+s'\big]\nn\\
& \qquad =(-1)^{N}|\wh{N-6},N-4,N-3,E_i\Ph(0),E_j\Ph(0),E_\kappa \Ph(0)|,
\end{align}
\end{subequations}
where \eqref{A21f-i} is actually the equation \eqref{B2:fd} with $j=N-1$, 
and equation \eqref{A21-sigij-Eijf} is the generalized version of
\eqref{eq:f-th}. 
The other equations can be verified in a way  similar to those for the equations \eqref{Formula-BI}. 
Some more formulas are
\begin{subequations}
\label{Formula-A21m}
\begin{align}
\label{A21-fN-6}
& f'|\wh{N-6},N-4,N-3,E_i\Ph(0),E_j\Ph(0),E_\kappa\Ph(0)|\nn \\
& ~~=-\sigma^-_{ij}E_iE_jf'|\wh{N-6},N-4,N-3, N-2,N-1,E_\kappa\Ph(0)|\nn\\
& \qquad+\sigma^-_{ik}E_iE_\kappa f'|\wh{N-6},N-4,N-3,N-2,N-1,E_j\Ph(0)| \nn\\
& \qquad -\sigma^-_{jk}E_jE_\kappa f'|\wh{N-6},N-4,N-3,N-2,N-1,E_i\Ph(0)|, \\
\label{A21-fN-5}
& f'|\wh{N-5},N-3,E_i\Ph(0),E_j\Ph(0),E_\kappa \Ph(0)| \nn \\
& ~~ =-\sigma^-_{ij}E_iE_jf'|\wh{N-5},N-3,N-2,N-1,E_\kappa \Ph(0)| \nn\\
& \qquad+\sigma^-_{i\kappa}E_iE_\kappa f'|\wh{N-5},N-3,N-2,N-1,E_j\Ph(0)|\nn \\
&\qquad -\sigma^-_{j\kappa}E_jE_\kappa f'|\wh{N-5}, N-3,N-2,N-1,E_i\Ph(0)|, \\
\label{A21-fN-4}
&f'|\wh{N-4},E_i\Ph(0),E_j\Ph(0),E_\kappa \Ph(0)| \nn \\
& ~~ =-\sigma^-_{ij}E_iE_jf'|\wh{N-4},N-2,N-1,E_\kappa \Ph(0)|
+\sigma^-_{i\kappa}E_iE_\kappa f'|\wh{N-4},N-2,N-1,E_j\Ph(0)| \nn \\
& \qquad  -\sigma^-_{j\kappa}E_jE_\kappa f'|\wh{N-4},N-2,N-1,E_i\Ph(0)|,
\\
\label{A21-gN-5}
&g'|\wh{N-5},N-3,E_i\Ph(0),E_j\Ph(0),E_\kappa \Ph(0)|\nn\\
& ~~ =-\sigma^-_{ij}E_iE_jf'|\wh{N-5},N-3,N-2,N, E_\kappa \Ph(0)|
+\sigma^-_{i\kappa}E_iE_\kappa f'|\wh{N-5},N-3,N-2,N,E_j\Ph(0)|\nn \\
&\qquad -\sigma^-_{j\kappa}E_jE_\kappa f'|\wh{N-5},N-3, N-2,N,E_i\Ph(0)|,\\
\label{A21-gN-4}
& g'|\wh{N-4},E_i\Ph(0),E_j\Ph(0),E_\kappa \Ph(0)|\nn\\
& ~~ =
-\sigma^-_{ij}E_iE_jf'|\wh{N-4},N-2,N,E_\kappa\Ph(0)|
+\sigma^-_{i\kappa}E_iE_\kappa f'|\wh{N-4},N-2,N,E_j\Ph(0)| \nn \\
& \qquad -\sigma^-_{j\kappa}E_jE_\kappa f'|\wh{N-4},N-2,N,E_i\Ph(0|, \\
\label{A21-hN-4}
& h'|\wh{N-4},E_i\Ph(0),E_j\Ph(0),E_\kappa\Ph(0)|\nn \\
& ~~ = -\sigma^-_{ij}E_iE_jf'|\wh{N-4},N-2,N+1,E_\kappa\Ph(0)|
+\sigma^-_{i\kappa} E_iE_\kappa f'|\wh{N-4},N-2,N+1,E_j\Ph(0)|\nn \\
& \qquad
-\sigma^-_{j\kappa}E_jE_\kappa f'|\wh{N-4},N-2,N+1,E_i\Ph(0|, \\
\label{A21-r-fdthc}
& E_\kappa f'\sigma^-_{ij}E_iE_jf'-E_jf' \sigma^-_{i\kappa}E_iE_\kappa f'
+E_if'\sigma^-_{j\kappa}E_jE_\kappa f'\nn\\
& ~~ =(-1)^{N-1}\Bigl (|\wh{N-2},E_\kappa \Ph(0)||\wh{N-3},E_j\Ph(0),E_i\Ph(0)|\nn \\
&\qquad -|\wh{N-2},E_j\Ph(0)||\wh{N-3},E_\kappa\Ph(0),E_i\Ph(0)|\nn\\
&\qquad +|\wh{N-2},E_i\Ph(0)||\wh{N-3},E_\kappa\Ph(0),E_j\Ph(0)|\Bigr)=0, \\
\label{A21-gamma4}
&(-1)^{N-1}\text{tr}(\bK) E_i f'=\bigg[-|\wh{N-6},N-4,N-3,N-2,N-1,E_i\Ph(0)| \nn\\
&\qquad +|\wh{N-5},N-3,N-2,N,E_i\Ph(0)|-|\wh{N-4},N-2,N+1,E_i\Ph(0)| \nn\\
& \qquad +|\wh{N-3},\psi(N+2),E_i\Ph(0)|+|\wh{N-2},E_i\Ph(4)|\nn \\
& \qquad -\epsilon_1\Bigl(|\wh{N-5},N-3,N-2,N-1, E_i\Ph(0)|-|\wh{N-4},N-2,N,E_i\Ph(0)|\nn\\
&\qquad+|\wh{N-3},N+1,E_i\Ph(0)| +|\wh{N-2},E_i\Ph(3)|\Bigr)\nn \\
&\qquad +\epsilon_2\Bigl(-|\wh{N-4},N-2, N-1,E_i\Ph(0)|+|\wh{N-3},N,E_i\Ph(0)|+|\wh{N-2},E_i\Ph(2)|\Bigr)\nn\\
&\qquad -\epsilon_3\Bigl(|\wh{N-3},N-1,E_i\Ph(0)|+|\wh{N-2},E_i\Ph(1)|\Bigr)\bigg].
\end{align}
\end{subequations}
Equations \eqref{A21-fN-6}-\eqref{A21-hN-4} are due to the identity \eqref{plu-r-2} with $k=4$ rather than the identity \eqref{plu-r-1}.
For example,  equation \eqref{A21-fN-6} corresponds to
\begin{align*}
&\bP=(\wh{N-6},N-4,N-3,N-2,N-1),~~\bQ=(\wh{N-6},N-4,N-3),\\
& \ba_1=\Ph(N-5),~~\ba_2=E_i\Ph(0),~~\ba_3=E_j\Ph(0),~~\ba_4=E_{\kappa}\Ph(0).
\end{align*}
Equation \eqref{A21-r-fdthc} is a case of   \eqref{plu-r-1} where
\[\bM=(\wh{N-3}),~~ \ba=\Ph(N-2), ~~\bb=E_{\kappa}\Ph(0),~~ \bc=E_j\Ph(0), ~~ \bd=E_i\Ph(0).\]
 For the verification of \eqref{A21-gamma4}, one can refer to the one for
\eqref{Ga4 B2}, while here taking $\Omega_{i,j}$ as \eqref{Omega-ij} but
$\Xi=\big(\wh{N-2},E_i\Ph(0) \Big)$ and making use of the expression \eqref{A21f-i}.

\section{Casoratian formulas for the alternative GD-4 (A-2) equation}
\label{A2:formulae}

The basic shift formulas are
\begin{subequations}
\label{Formula-a221}
\begin{align}
\label{A22 f-f dot}
& f'=(-1)^{N-1}|\d{\wh{N-2}},\Ph(0)|,\\
\label{A22-h}
&h'=(-1)^{N-2}(|\d{\wh{N-3}},\Ph(0),\d{\d{\Ph}}(N-1)|-2\sigma_4|\d{\wh{N-3}},\Ph(0),\d{\Ph}(N-1)|)+\sigma_4^2 f',\\
\label{A22-dg}
& \d{g}'=|\d{\wh{N-2}},\d{\d{\Ph}}(N-1)|-\sigma_4 \d{f}',\\
\label{A22 dEf}
&E_i\d{f}'=(-1)^{N-1}|\d{\wh{N-2}},E_i\d{\Ph}(0)|,\\
\label{A22-dEg}
&E_i \d{g}'=(-1)^{N-1}|\d{\wh{N-3}},\d{\Ph}(N-1),E_i\d{\Ph}(0)|-\sigma_i E_i\d{f}',\\
\label{A22-dEh}
&E_i\d{h}'=(-1)^{N-2}(|\d{\wh{N-3}},E_i\d{\Ph}(0),\d{\d{\Ph}}(N-1)|
-\sigma_4|\d{\wh{N-3}},E_i\d{\Ph}(0),\d{\Ph}(N-1)|)-\sigma_iE_i\d{g}',\\
\label{A22 p-b f-i}
&\sigma^{-}_{i4}E_i f'=|\d{\wh{N-3}},\Ph(0),E_i \d{\Ph}(0)|,\\
\label{A22 bf g-dot}
& \sigma_4 f'+g'=(-1)^{N-1}|\d{\wh{N-3}},\d{\Ph}(N-1),\Ph(0)|,\\
\label{A22 dotfg-i}
& E_i(\sigma_i\d{f}'+\d{g}')=(-1)^{N-1}|\d{\wh{N-3}},\d{\Ph}(N-1),E_i\d{\Ph}(0)|.
\end{align}
\end{subequations}

Equations \eqref{A22 f-f dot}-\eqref{A22-dg} can be proved directly in light of the definition \eqref{eq:moB2}
and shift relation \eqref{sig4-dd}. \eqref{A22 dEf}-\eqref{A22-dEh} and \eqref{A22 dotfg-i} are the dot shift results of $E_if'$,
$E_ig'$, $E_ih'$ and $E_i(\sigma_i f'+g')$, where $E_if'$ and $E_i(\sigma_i f'+g')$ have been given in \eqref{A21f-i} and \eqref{A21-sigi-Eifg}, respectively.

\section{Casoratian formulas for the GD-4 (C-3) equation}
\label{C3:formulae}

The basic shift formulas are
\begin{subequations}
\label{Formula-C3b}
\begin{align}
\label{C3-sigi-j-f}
&\sigma_i^{j}E_i\d{f}'=(-1)^{N-1+j}|\d{\wh{N-2}},E_i\d{\Ph}(j)|,\quad j\leq N-1,\\
\label{C3-sigij4-Eijf}
&\sigma^{-}_{ij4}E_iE_j f'=(-1)^{N}|\d{\wh{N-4}},E_i\d{\Ph}(0),E_j\d{\Ph}(0),\Ph(0)|,\\
\label{C3-sig123-dfdthc}
&\sigma^{-}_{123}\d{\dthc{f}}'=(-1)^N|\d{\wh{N-4}},\d{\dt{\Ph}}(0),\d{\dh{\Ph}}(0),\d{\dc{\Ph}}(0)|, \\
\label{C3-sig1234-dthf}
&\sigma^{-}_{1234}\dthc{f}'=|\d{\wh{N-5}},\d{\dt{\Ph}}(0),\d{\dh{\Ph}}(0),\d{\dc{\Ph}}(0),\Ph(0)|,
\\
\label{C3-sig1234-dthgf}
&\sigma^{-}_{1234}(\dthc{g}'+\sigma^{+}_{1234}\dthc{f}')
=|\d{\wh{N-6}},\d{\Ph}(N-4),\d{\dt{\Ph}}(0),\d{\dh{\Ph}}(0),\d{\dc{\Ph}}(0),\Ph(0)|.
\end{align}
\end{subequations}
Equations \eqref{C3-sigi-j-f} and \eqref{C3-sig123-dfdthc} are the dot shift of \eqref{B2:fd}
and \eqref{A21:f-dthc} (with $i=1,j=2,\kappa=3$), respectively.
Proofs of the remaining ones can be manipulated similarly to the relations \eqref{A21:f-dthc} and \eqref{A21:fg-dthc}.
Some more formulas are
\begin{subequations}
\label{Formula-c32}
\begin{align}
\label{C3-fN-6}
&\d{f}'|\d{\wh{N-6}},\d{\Ph}(N-4),\d{\dt{\Ph}}(0),\d{\dh{\Ph}}(0),\d{\dc{\Ph}}(0),\Ph(0)| \nn \\
& ~~ =(-1)^{N}\Bigl[\sigma^{-}_{123}\d{\dthc{f}}'|\d{\wh{N-6}},\d{\Ph}(N-4),\d{\Ph}(N-3),
\d{\Ph}(N-2),\d{\Ph}(N-1),\Ph(0)| \nn\\
&\qquad
-\sigma^{-}_{124}\dth{f}'|\d{\wh{N-6}},\d{\Ph}(N-4),\d{\Ph}(N-3),
\d{\Ph}(N-2),\d{\Ph}(N-1),\d{\dc{\Ph}}(0)|\nn\\
&\qquad
+\sigma^{-}_{134}\dtc{f}'|\d{\wh{N-6}},\d{\Ph}(N-4),\d{\Ph}(N-3),\d{\Ph}(N-2),
\d{\Ph}(N-1),\d{\dh{\Ph}}(0)|\nn\\
&\qquad
-\sigma^{-}_{234}\dhc{f}'
|\d{\wh{N-6}},\d{\Ph}(N-4),\d{\Ph}(N-3),\d{\Ph}(N-2),\d{\Ph}(N-1),\d{\dt{\Ph}}(0)|\Bigr],
 \\
\label{C3-fN5}
& \d{f}'|\d{\wh{N-5}},\d{\dt{\Ph}}(0),\d{\dh{\Ph}}(0),\d{\dc{\Ph}}(0),\Ph(0)|\nn \\
& ~~ =(-1)^{N}\Bigl[
\sigma^{-}_{123}\d{\dthc{f}}'|\d{\wh{N-5}},\d{\Ph}(N-3),\d{\Ph}(N-2),\d{\Ph}(N-1),\Ph(0)| \nn\\
&\qquad-\sigma^{-}_{124}\dth{f}'|\d{\wh{N-5}},\d{\Ph}(N-3),\d{\Ph}(N-2),\d{\Ph}(N-1),\d{\dc{\Ph}}(0)|\nn\\
&\qquad  +\sigma^{-}_{134}\dtc{f}'
\d{\wh{N-5}},\d{\Ph}(N-3),\d{\Ph}(N-2),\d{\Ph}(N-1),\d{\dh{\Ph}}(0)|\nn \\
&\qquad -\sigma^{-}_{234}\dhc{f}'|\d{\wh{N-5}},\d{\Ph}(N-3),
\d{\Ph}(N-2),\d{\Ph}(N-1),\d{\dt{\Ph}}(0)|\Bigr],
\\
\label{C3-gN-5}
&\d{g}'|\d{\wh{N-5}},\d{\dt{\Ph}}(0),\d{\dh{\Ph}}(0),\d{\dc{\Ph}}(0),\Ph(0)| \nn \\
& ~~ =(-1)^{N}\Bigl[\sigma^{-}_{123}\d{\dthc{f}}'|\d{\wh{N-5}},\d{\Ph}(N-3),\d{\Ph}(N-2),\d{\Ph}(N),\Ph(0)|
\nn\\
&\qquad -\sigma^{-}_{124}\dth{f}'|\d{\wh{N-5}},\dot{\Ph}(N-3),\d{\Ph}(N-2),\d{\Ph}(N),\d{\dc{\Ph}}(0)|\nn \\
&\qquad +\sigma^{-}_{134}\dtc{f}'|\d{\wh{N-5}},\d{\Ph}(N-3),
\d{\Ph}(N-2),\d{\Ph}(N),\d{\dh{\Ph}}(0)| \nn \\
& \qquad -\sigma^{-}_{234}\dhc{f}'|\d{\wh{N-5}},\d{\Ph}(N-3),\d{\Ph}(N-2),\d{\Ph}(N),\d{\dt{\Ph}}(0)|\Bigr],
\\
\label{G4-C3}
&(-1)^{N-1}\text{tr}(\bK)E_i\d{f}'\nn\\
&=-|\d{\wh{N-6}},\d{\Ph}(N-4),\d{\Ph}(N-3),\d{\Ph}(N-2),\d{\Ph}(N-1),E_i \d{\Ph}(0)| \nn \\
& \qquad +|\d{\wh{N-5}},\d{\Ph}(N-3), \d{\Ph}(N-2), \d{\Ph}(N),E_i\d{\Ph}(0)|\nn \\
& \qquad -|\d{\wh{N-4}},\d{\Ph}(N-2),\d{\Ph}(N+1),E_i\d{\Ph}(0)|
+|\d{\wh{N-3}}, \d{\Ph}(N+2),E_i\d{\Ph}(0)|+|\d{\wh{N-2}},E_i\d{\Ph}(4)|\nn \\
& \qquad -\epsilon_1\Bigl(|\d{\wh{N-5}},\d{\Ph}(N-3),\d{\Ph}(N-2),\d{\Ph}(N-1), E_i\d{\Ph}(0)| \nn \\
& \qquad -|\d{\wh{N-4}},\d{\Ph}(N-2),\d{\Ph}(N),E_i\d{\Ph}(0)|+|\d{\wh{N-3}},\d{\Ph}(N+1),E_i\d{\Ph}(0)|
+|\d{\wh{N-2}},E_i\d{\Ph}(3)|\Bigr)\nn \\
& \qquad +\epsilon_2\Bigl(-|\d{\wh{N-4}},\d{\Ph}(N-2),\d{\Ph}(N-1),E_i\d{\Ph}(0)|
+|\d{\wh{N-3}},\d{\Ph}(N),E_i\d{\Ph}(0)|+|\d{\wh{N-2}},E_i\d{\Ph}(2)|\Bigr)\nn \\
& \qquad -\epsilon_3\Bigl(|\d{\wh{N-3}},\d{\Ph}(N-1),E_i\d{\Ph}(0)|+|\d{\wh{N-2}},E_i\d{\Ph}(1)|\Bigr), \\
\label{C3-fbf0}
&\sigma^{-}_{123}f'\d{\dthc{f}}'-\sigma^{-}_{124}\d{\dc{f}}'{\dth{f}'}
+\sigma^{-}_{134}
\d{\dh{f}}'{\dtc{f}'}-\sigma^{-}_{234}\d{\dt{f}}'{\dhc{f}'}\nn \\
& ~~ =-(|\d{\wh{N-2}},\Ph(0)||\d{\wh{N-2}},
\dt{\d{\Ph}}(0),\dh{\d{\Ph}}(0),\dc{\d{\Ph}}(0)|\nn\\
&\qquad-|\d{\wh{N-2}},\d{\dt{\Ph}}(0)||\d{\wh{N-2}},\Ph(0),\dh{\d{\Ph}}(0),\dc{\d{\Ph}}(0)|
+|\d{\wh{N-2}},\d{\dh{\Ph}}(0)||\d{\wh{N-2}},\Ph(0),\dt{\d{\Ph}}(0),\dc{\d{\Ph}}(0)|\nn\\
&\qquad-|\d{\wh{N-2}},\d{\dc{\Ph}}(0)||\d{\wh{N-2}},\Ph(0),\dt{\d{\Ph}}(0),\dh{\d{\Ph}}(0)|)=0.
\end{align}
\end{subequations}
Note that equations \eqref{C3-fN-6}-\eqref{C3-gN-5} are due to the identity \eqref{plu-r-2} with $k=5$. For instance,
equation \eqref{C3-fN-6} is obtained from \eqref{plu-r-2} with $k=5$ and
\begin{align*}
& \bP=(\d{\wh{N-6}},\d{\Ph}(N-4),\d{\Ph}(N-3), \d{\Ph}(N-2),\d{\Ph}(N-1)), 
~~\bQ=(\d{\wh{N-6}},\d{\Ph}(N-4)),\\
& \ba_1=\d{\dt{\Ph}}(0),~~ \ba_2=\d{\dh{\Ph}}(0),~~ 
 \ba_3=\d{\dc{\Ph}}(0),~~ \ba_4=\Ph(0),~~ \ba_5=\d{\Ph}(N-5).
\end{align*}

\end{document}